\documentclass[11pt,a4paper]{article}
\usepackage[margin=1in]{geometry}

\usepackage{graphicx}
\usepackage{amsmath,amssymb}
\usepackage{amsfonts}
\usepackage{amsthm}

\usepackage{xcolor}
\definecolor{darkgreen}{rgb}{0, 0.5, 0}
\usepackage[colorlinks,citecolor=darkgreen]{hyperref}

\def\mpfile#1#2{\includegraphics{#1-#2.eps}}

\newcommand*{\poly}{\mathop{\mathrm{poly}}}
\def\op#1{\mathop{\mathrm{#1}}\nolimits}
\def\Wd{\op{width}}
\def\Dom{\op{Dom}}
\def\Ht{h}
\def\Htmax{H}
\newcommand{\sset}{\subseteq}
\def\vc#1{\boldsymbol{ #1}}
\let\ld\lambda
\def\F{{\mathcal F}}

\def\A{{\mathcal A}}
\def\root{\mathop{\mathrm{root}}}
\def\sm{\setminus}

\let\epsilon\varepsilon
\let\eps\varepsilon
\let\al\alpha

\let\phi\varphi

\let\ld\lambda

\def\R{{\mathcal R}}
\def\Z{{\mathcal Z}}
\def\B{{\mathcal B}}

\def\FF{\mathbb F}

\def\es{\varnothing}

\theoremstyle{plain}
\newtheorem{theorem}{Theorem}
\newtheorem{lemma}{Lemma}
\newtheorem{fact}{Fact}
\newtheorem{prop}{Claim}
\newtheorem{cor}{Corollary}
\theoremstyle{definition}
\newtheorem{defi}{Definition}
\newtheorem{ex}{Example}
\theoremstyle{remark}
\newtheorem{Remark}{Remark}

\begin{document}

\title{Re-pairing brackets}

\author{%
    Dmitry Chistikov${}^1$
    \and
    Mikhail Vyalyi${}^2$
}
\date{%
$^1$Centre for Discrete Mathematics and its Applications\\(DIMAP) \&
Department of Computer Science,\\University of Warwick\\
Coventry, United Kingdom\\
\texttt{d.chistikov@warwick.ac.uk}\\
\mbox{}\\
$^2$National Research University Higher School 
of Economics,\\
 Moscow Institute of Physics and Technology,\\
 Dorodnicyn Computing Centre, FRC CSC RAS,\\
 Moscow, Russia\\
\texttt{vyalyi@gmail.com}}
\maketitle

\begin{abstract}
Consider the following one-player game. Take a well-formed sequence
of opening and closing brackets.
As a move, the player can \emph{pair} any opening bracket
with any closing bracket to its right, \emph{erasing} them.
The goal is to re-pair (erase) the entire sequence,
and the complexity of a strategy is measured by its \emph{width:} the maximum number of
nonempty segments of symbols (separated by blank space) seen
during the play.

For various initial sequences, we prove upper and lower bounds on the
minimum width sufficient for re-pairing. (In particular, the sequence
associated with the complete binary tree of height $n$ admits a strategy
of width sub-exponential in $\log n$.)
Our two key contributions are
(1) lower bounds on the width and
(2) their application in automata theory:
    quasi-polynomial lower bounds on the translation from one-counter automata to Parikh-equivalent nondeterministic finite automata.
The latter result answers a question by Atig et~al.~(2016).
\end{abstract}

\newpage

\section{Introduction}
\label{s:intro}

Consider the following one-player game. Take a well-formed sequence
of opening and closing brackets; that is, a word in the Dyck language.
As a move, the player can \emph{pair} any opening bracket
with any closing bracket to its right,
 \emph{erasing} them.
The two brackets do not need to be adjacent or matched to each other
in the word.
The goal is to re-pair (erase) the entire word,
and the complexity of a play is measured by its \emph{width:} the maximum number of
nonempty segments of symbols (`islands' separated by blank space) seen
during the play.
Here is an example:
\newcommand{\ope}{\text{\texttt{(}}\hspace{-1em}}
\newcommand{\clo}{\text{\texttt{)}}\hspace{-1em}}
\newcommand{\emp}{\text{\texttt{\ }}\hspace{-1em}}
\renewcommand{\le}{\leq}
\renewcommand{\ge}{\geq}
\begin{quote}
$
\begin{array}{rcccccccc}
\text{Initial word:}     & \ope & \ope & \clo & \ope & \clo & \clo\\
\text{After move $1$:}   & \ope & \emp & \emp & \ope & \clo & \clo\\
\text{After move $2$:}   & \ope & \emp & \emp & \emp & \clo & \emp\\
\text{After move $3$:}   & \emp & \emp & \emp & \emp & \emp & \emp
\end{array}
$
\end{quote}
Note that move~$2$ pairs up two brackets not matched to each
other in the initial word; such moves are permitted without restrictions.
At the beginning, there is a single segment, which splits into two
after the first move.
Both segments
disappear simultaneously after the third move;
the width of the play is equal to~$2$.
In this example, width~$1$ is actually sufficient:
a~better strategy is to erase the endpoints first and then erase the two
matched pairs in either order.

For a word $\sigma$, the \emph{width} of~$\sigma$ is the minimum
width sufficient for re-pairing it.
Is it true that all well-formed (Dyck) words, no matter how long,
have bounded complexity, i.e.,
can be re-paired using width at most $c$, where $c$ is independent of the word?
The answer to this simply formulated combinatorial question turns out to be negative,
but there does not appear to be a simple proof for this:
strategies, perhaps surprisingly,
turn out quite intricate.
In the present paper, we study this and related questions.

\smallskip
\emph{Motivation.}
First of all, we find the re-pairing problem interesting in its
own right---as a curious combinatorial game,
easily explained using just the very basic concepts in discrete mathematics.
So our motivation is, in part, driven by the appeal of the problem itself.

We have first identified the re-pairing problem when studying
an open question in automata theory, the complexity of
translation of one-counter automata (OCA) on finite words
into Parikh-equivalent nondeterministic finite automata (NFA)~\cite{AtigCHKSZ16}.
%
This translation arose in model checking,
namely in the context of
availability expressions~\cite{AbdullaAMS15,HoenickeMO10}.
It is more generally motivated by recent lines of research
on
the classic Parikh theorem~\cite{Parikh66}:
on its applications in verification
(see, e.g.,~\cite{%
      Ganty:2012:AVA:2160910.2160915,
      EsparzaGP14,
      DBLP:conf/cav/HagueL12%
      })
and on its
extensions and refinements%
required for them~\cite{EsparzaGKL11,KT10,Kopczynski15}.
It has been unknown~\cite{AtigCHKSZ16}
whether the translation in question can be made
polynomial, and in this paper we answer this question
negatively: using our results on the re-pairing problem,
we obtain a quasi-polynomial lower bound
on the blowup in the translation.

The re-pairing problem
is a curious case study in the theory
of \emph{non-uniform models of computation.}
As it turns out, 
restricted strategies,
where paired brackets
are always matched to each other in the word,
have a close connection with
the black-and-white pebble game on binary trees,
a classic setting in computational complexity
(see, e.g., surveys by Nordstr\"om~\cite{lmcs:1111,Nordstrom-survey}).
We show that
unrestricted strategies in the re-pairing game
make it significantly more complex than pebbling,
strengthening this model of computation.

\medskip

Our two key contributions are
(i) \textbf{lower bounds on the width} (against this model) and
(ii) the connection to automata theory:
     \textbf{lower bounds on the translation from OCA to Parikh-equivalent NFA.}

\emph{Our lower bounds on the width}
are obtained by bounding the
set of (Dyck) words that can be re-paired using width~$k$ (for each~$k$).
To put this into context,
classic models of matrix grammars~\cite{Rozoy-j} and deterministic two-way transducers~\cite{Rozoy-c}, as well as a more recent model of streaming string transducers~\cite{AlurC11,AlurC10},
extend our model of computation with additional finite-state memory.
(We refer the reader to surveys~\cite{FiliotR16,Muscholl17,MuschollP19} for more details.)
In terms of transducers, our technique would
correspond to determining the expressive power of machines of bounded
size. Existing results of this kind (see~\cite{FiliotR16,Muscholl17})
%
%
apply to variants of the model where \emph{concatenation is restricted:}
with~\cite{AlurR13} or without~\cite{DaviaudRT16} restrictions on output
and in the presence of nondeterminism~\cite{BaschenisGMP16}.
%
%
Our result is,
to the best of our knowledge,
the first lower bound against the \emph{unrestricted} model.

In the \emph{application of the re-pairing problem to automata theory,}
our lower bounds on the size of NFA for the Parikh image
can be viewed as lower bounds on the size of commutative NFA,
i.e., nondeterministic automata over the free commutative monoid
(cf.~\cite{Huynh83,Huynh85,Esparza97,KT10,Kopczynski15,HaaseH16}).
To the best of our knowledge, we are the first to develop lower bounds
(on description size) for this simple non-uniform model of computation.
It is well-known that, even for usual NFA,
obtaining lower bounds on the size (number of states) is
challenging and, in fact, provably (and notoriously) hard;
and that the available toolbox of techniques is limited
(see, e.g.,~\cite{GruberH06eccc,HolzerK08} and~\cite{HromkovicPS09}).
From the `NFA perspective',
we first develop a lower bound for a stronger model of computation
and then import this result to NFA
using combinatorial tools,
which we thus bring to automata theory:
the Birkhoff---von Neumann theorem on doubly stochastic matrices
(see, e.g.,~\cite[p.~301]{Schrijver03})
and the Nisan---Wigderson construction of a family of sets
with pairwise low intersection~\cite{NW}.
The obtained lower bounds point to
a limitation of NFA that does not seem to have the form of
the usual communication complexity bottleneck
(cf.~\cite[Theorem~3.11.4]{Shallit2nd}, \cite{HromkovicPS09},
and the book by Hromkovic~\cite{Hromkovic97});
exploring and exploiting this further is a possible direction
for future research.

\subsection*{Our contribution}

In this paper, we define and study the \emph{re-pairing} problem (game),
as sketched above. Our main results are as follows:

\begin{enumerate}

\item
We show 
that every well-formed (Dyck) word $\sigma$ has a 
re-pairing of width $O(\log |\sigma|)$,
where $|\sigma|$ is the length of~$\sigma$.

This re-pairing always pairs up brackets that are matched to each other
in~$\sigma$; we call re-pairings with this property \emph{simple}. 
It is standard that well-formed words are associated with trees;
for words~$\sigma$ associated with \emph{binary} trees,
we show that
the minimum width of a simple re-pairing is equal
(up to a constant factor) to the
\emph{minimum number of pebbles in the black-and-white pebble game}
on the associated tree, a quantity that has been studied
in computational
complexity~\cite{CookS76,LengauerT80,lmcs:1111,Nordstrom-survey} 
and captures the amount of space used by nondeterministic computation.
%

In particular, this means~\cite{Loui79,Meyer-a-d-H81,LengauerT80} that
for the word $Z(n)$ associated with
a complete binary tree of height~$n$,
the minimum width of simple re-pairings is $\Theta(n)$,
which is logarithmic in the length of $Z(n)$.

\item
For $Z(n)$, we show 
how to beat this bound,
giving a (non-simple) recursive re-pairing strategy
of width $2^{O(\sqrt{\log n})}$.
This is a function sub-exponential in $\log n$;
it grows faster than all $(\log n)^k$, but
slower than all $n^\eps$, $\eps > 0$.

\item\label{intro:lb}
For $Z(n)$ and for a certain `stretched' version of it,
$Y(\ell)$,
we prove 
lower bounds on the width of re-pairings:
\begin{equation}
 \label{eq:lb}
 \begin{aligned}
   \Wd(Z(n))
   & =
   \Omega\!\left(
         \frac{     \log \log |Z(n)|}%
              {\log \log \log |Z(n)|}
     \right)
   =
   \Omega\!\left(\frac{\log n}{\log\log n}\right)\!,
   \\
   \Wd(Y(\ell)) & = \Omega\!\left(
    \sqrt{
    \frac{     \log |Y(\ell)|}%
         {\log \log |Y(\ell)|}
    }
    \,
   \right)
   =
   \Omega(\ell).
 \end{aligned}
\end{equation}
\item
As an application of our lower bounds, we prove 
that there is no polynomial-time translation from one-counter automata
(OCA) on finite words to Parikh-equivalent nondeterministic finite
automata (NFA). This shows that optimal translations must be
quasi-polynomial,
answering a question by Atig et~al.~\cite{AtigCHKSZ16}.

To prove this result,
we consider
OCA from a specific \emph{complete} family, $(\mathcal H_n)_{n \ge 2}$,
identified by Atig et~al.~\cite{AtigCHKSZ16}.
(There is a polynomial
translation from (any) OCA to Parikh-equivalent NFA if and only if
these OCA~$\mathcal H_n$ have Parikh-equivalent NFA of polynomial size.)
We prove, for every Dyck word $\sigma_n$ of length $O(\sqrt{n})$, 
a~lower bound of $n^{\Omega(\Wd(\sigma_n))}$
on the minimum size of NFA accepting regular languages Parikh-equivalent
to~$L(\mathcal H_n)$.
Based on the words $Y(\ell)$, we get a lower bound of
\begin{equation*}
n^{\Omega\left(\sqrt{\log n / \log \log n}\right)}
\end{equation*}
on the size of NFA.
Note that this holds for NFA that accept not just a specific
regular language,
but \emph{any} language Parikh-equivalent to
the one-counter language $L(\mathcal H_n)$ (there are infinitely many such
languages for each~$n$).
\end{enumerate}

\subsection*{Background and related work}

\paragraph*{Parikh image of one-counter languages.}
The problem of re-pairing brackets in well-formed words
is linked to the following problem in automata theory.

The \emph{Parikh image} (or commutative image) of a word~$u$ over
an alphabet $\Sigma$ is a vector of dimension $|\Sigma|$ in which
the components specify how many times each letter from $\Sigma$
occurs in~$u$. The Parikh image of a language $L \sset \Sigma^*$
is the set of Parikh images of all words $u \in L$.
It is well-known~\cite{Parikh66} that for every context-free language $L$ there exists
a regular language~$R$ with the same Parikh image
(\emph{Parikh-equivalent} to~$L$).
If $L$ is generated by a context-free grammar of size $n$, then
there is a nondeterministic finite automaton (NFA) of size exponential in~$n$
that accepts such a regular language~$R$ (see~\cite{EsparzaGKL11}); the
exponential in this translation is necessary in the worst case.

When applying this translation
to a language from a proper subclass of context-free languages,
it is natural to ask whether this blowup in description size can be avoided.
For languages recognized by
\emph{one-counter automata} (OCA; a fundamental subclass of pushdown automata),
the exponential construction is 
suboptimal~\cite{AtigCHKSZ16}.
If an alphabet~$\Sigma$ is fixed, then for every OCA with $n$~states over~$\Sigma$
there exists a Parikh-equivalent NFA of polynomial size (the degree of this polynomial
depends on $|\Sigma|$).
And even in general, if the alphabet is not fixed,
for every OCA with~$n$ states over an alphabet of cardinality at most~$n$
there exists a Parikh-equivalent NFA of size $n^{O(\log n)}$,
quasi-polynomial in~$n$.
Whether this quasi-polynomial construction is optimal has been unknown, and we prove in the present paper a quasi-polynomial lower bound.

We note that the gap between NFA of polynomial and quasi-polynomial size
grows to exponential
when the translation is applied iteratively, as 
is the case
in Abdulla et~al.~\cite{AbdullaAMS15}.

\paragraph*{Matrix grammars of finite index and transducers.}
The question of whether all well-formed (Dyck) words can be re-paired
using bounded width can be linked to a question on \emph{matrix grammars},
a model of computation studied 
since the 1960s~\cite{Abraham65}. 
Matrix grammars are a generalization of context-free grammars
in which productions are applied in `batches'
prescribed by the grammar.
This formalism subsumes many classes of
rewriting systems, including
controlled grammars,
L~systems, etc.
(see, e.g.,~\cite{DassowPS97}).
%

The \emph{index} of a derivation in a matrix grammar is the maximum
number of nonterminals in a sentential form in this derivation
(this definition applies to ordinary context-free grammars as well)%
~\cite{Brainerd67,GinsburgS68}.
Bounding the index of derivations, i.e.,
restricting grammars to \emph{finite index} is known to
reduce the class of generated languages;
this holds both
for ordinary context-free~\cite{GinsburgS68,Salomaa69index,Gruska71index}
and
matrix grammars~\cite{Brainerd67}.
Languages generated by finite-index matrix grammars
have many characterizations:
as languages output by deterministic two-way transducers
with one-way output tape~\cite{Rajlich72},
or produced by EDT0L systems of finite index~\cite[Proposition~I.2]{Latteux79};
images of monadic second-order logic (MSO)
transductions~\cite{EngelfrietH01};
and, most recently, output languages of streaming string transducers~\cite{AlurC11,AlurC10}.
(See also the survey by Filiot and Reynier~\cite{FiliotR16}.)


Encoding the rules of our re-pairing problem in the matrix grammar formalism
leads to a simple sequence of grammars with index~$k = 2, 3, \ldots$
for subsets of the Dyck language~$D_1$;
the question of whether all Dyck words can be re-paired using bounded width is
the same as asking if any of these grammars has in fact
\emph{(bounded-index) derivations for all Dyck words}.
A~1987 paper by Rozoy~\cite{Rozoy-j}
is devoted to the proof that, in fact, no matrix grammar
can generate all words in $D_1$ using bounded-index derivations
without also generating some words outside~$D_1$. This amounts
to saying that no finite-index matrix grammar generates~$D_1$;
and a non-constant lower bound on the width in the re-pairing problem
could be extracted from the proof.

Unfortunately, the proof in that paper seems to be flawed.
Fixing
the argument does not seem possible, and
we are not aware of an alternative proof.
(We discuss the details and compare the proof to our construction
 in Appendix~\ref{app:rozoy}).

\section{Basic definitions}
\label{s:def}

\paragraph*{The Dyck language.}

We use non-standard notation for brackets
in words from the Dyck language:
the opening bracket is denoted by~$+$
and the closing bracket by~$-$;
we call these symbols pluses and minuses, accordingly.
Moreover, in some contexts it is convenient
to interpret $+$ and $-$ as integers $+1$ and $-1$.

Let $N$ be an even integer.
A word $\sigma = (\sigma(1),\dots,\sigma(N))$, $\sigma(i)\in\{+1,-1\}$,
is  a \emph{Dyck word}
(or a \emph{well-formed word})
if it has an equal number of $+1$ and $-1$
and for every $1\leq k\leq N$ the inequality
$
\sum_{i=1}^k \sigma(i) \geq0
$
is satisfied.
The \emph{height} of a position~$i$ in a well-formed word~$\sigma$
is
$
\Ht(i) = \sum_{j=1}^{i}\sigma(j)
$.
As usual, $|\sigma|$ denotes the length of the word~$\sigma$
(the number of symbols in it).

Dyck words are naturally \emph{associated} with
\emph{ordered rooted forests} (i.e., with sequences of ordered
rooted trees). E.g.,  words~$Z(n)$ defined by 
\begin{equation}\label{eq:Zn-def}
Z(1) = {+}{-};\quad Z(n+1) = {+}Z(n)Z(n){-}
\end{equation}
can be associated with  complete binary trees of
height~$n-1$. Recall that the height of a rooted tree is the maximum
length of a path (number of edges) from the root to a leaf.

Note that we described a re-pairing of the word~$Z(2)$ in section~\ref{s:intro}.


\paragraph*{Re-pairings and their width.}

{\def\theenumi{R\arabic{enumi}}\def\labelenumi{(\theenumi)}
A \emph{re-pairing} of a well-formed word $\sigma$
is a sequence of pairs
\begin{equation*}
p = (p_1,\dots, p_{N/2}),\quad \text{where } p_i = (\ell_i,r_i)
\end{equation*}
and the following properties are satisfied:
\begin{enumerate}
\item\label{coloring} $\sigma(\ell_i)=+1$, $\sigma(r_i)=-1$, $\ell_i<r_i$ for all $i$;
\item every number from the interval $[1,N]$ occurs in exactly one pair~$p_j$.
\end{enumerate}
}
(We use the word `interval' to refer to a set of the form
$[a,b] = \{x \in \mathbb Z: a\leq x\leq b\}$.)
 
The intuition is that the index~$i$ corresponds to discrete time,
and at time~$i$ the two symbols $\sigma(\ell_i)$ and $\sigma(r_i)$
are \emph{(re-)paired} (or \emph{erased}).
Denote by
$B_t(p) = \{b\in [1,N]: (b = \ell_i)\lor(b=r_i),\ i\leq t\}$
the set of points from  $[1,N]$
that correspond to symbols erased at times~$[1, t]$.

It is easy to see that re-pairings exist for every well-formed word.
By induction on the length of the word one can prove a stronger statement:
a word $(p_1,\dots, p_{t})$ can be extended to a re-pairing iff 
all numbers in the pairs $p_i= (\ell_i,r_i)$ are different,
the property~\eqref{coloring} is satisfied, 
and the remaining signs (those which have not been erased)
constitute a well-formed word.
We now define the following quantities:
\begin{itemize}
\item
\emph{The width of a set $S$ of integers}, $\Wd(S)$, is the smallest
number of intervals the union of which is equal to~$S$.
\item
\emph{The width of a re-pairing~$p$ at time $t$} is $\Wd(B_t(p))$.
\item
\emph{The width of a re-pairing~$p$} of a well-formed word~$\sigma$,
$\Wd(p)$,
is $\max_t \Wd(B_t(p))$,
i.e., the maximum of the width of this re-pairing
over all time points.
\item
\emph{The width of a well-formed word~$\sigma$}, $\Wd(\sigma)$,
is $\min_p \Wd(p)$,
where the minimum is over all re-pairings of~$\sigma$.
\end{itemize}
We will look into how big the width of a well-formed word of length~$N$ can be,
that is, we are interested in $\max_\sigma \Wd(\sigma)$,
where the maximum is over all well-formed words of length~$N$.

\begin{Remark}
  Section~\ref{s:intro} discussed the minimization of the maximum number
  of the ``surviving'' (non-erased) intervals.
  This quantity cannot differ from
  the width defined above by more than~$1$.
\end{Remark}

\begin{Remark}
A tree-based representation of re-pairings is described in
Section~\ref{s:lwr-bnd} and, in more details, in Appendix,
Section~\ref{app:s:trees}. 
\end{Remark}

\section{Simple bounds and simple re-pairings}
\label{s:simple}

In this section we establish several basic facts on the width of
well-formed words and re-pairings (proofs are provided in Appendix, Section~\ref{app:s:simple}).
A careful use of bisection leads to the following upper bound:

\begin{theorem}\label{log-upbnd}
  $\Wd(\sigma) = O(\log |\sigma|)$
  for all well-formed words~$\sigma$.
\end{theorem}

We call a re-pairing of a well-formed word~$\sigma$ \emph{simple}
if at all times it pairs up two signs that are matching
%
in the word~$\sigma$.
The re-pairing that the proof of Theorem~\ref{log-upbnd} constructs
is simple.

We now show a link between simple re-pairings and strategies
in the following game.
Let $G$ be an acyclic graph
(in our specific case it will be a tree
 with edges directed from leaves to root).
Define a \emph{black-and-white pebble game} on~$G$
(see, e.g.,~\cite{LengauerT80,Nordstrom-survey}) as follows.
There is only one player,
and black and white pebbles are placed on the nodes of the graph.
The following moves are possible:
\begin{enumerate}
\renewcommand{\theenumi}{(M\arabic{enumi})}
\renewcommand{\labelenumi}{\theenumi}
\item\label{bw-move:place-black}
place a black pebble on a node,
provided that all its immediate predecessors
carry pebbles;
\item\label{bw-move:remove-black}
remove a black pebble from any node;
\item
place a white pebble on any node; and
\item\label{bw-move:remove-white}
remove a white pebble from a node,
provided that all its immediate predecessors
carry pebbles.
\end{enumerate}
(In a tree, immediate predecessors are immediate descendants,
 i.e., children.
 Rules~\ref{bw-move:place-black} and~\ref{bw-move:remove-white}
 are applicable to all sources, i.e., leaves of~$G$.)
At the beginning there are no pebbles on any nodes.
A sequence of moves in the game is a \emph{strategy};
it is \emph{successful} if it achieves the goal:
reaching a configuration in which
all sinks of the graph carry pebbles
and there are no white pebbles on any nodes.
By $\mathrm{bw}(G)$ we will denote
the minimum number of pebbles sufficient
for a successful strategy in the black-and-white pebble game
on~$G$.

%

\begin{theorem}
\label{th:peb}
Suppose the tree~$D$ associated with a well-formed word~$\sigma$
is binary.
Then the minimum width of a simple re-pairing for~$\sigma$
is $\Theta(\mathrm{bw}(D))$.
\end{theorem}

Since $D$ is a tree,
it follows from the results of the papers~\cite{Loui79,Meyer-a-d-H81,LengauerT80}
(see also~\cite[pp.~526--528]{savage})
that the value of $\mathrm{bw}(D)$
at most doubles
if the strategies are not allowed any white pebbles.
The optimal number of (black) pebbles in such strategies
is determined by the so-called
Strahler number
(see, e.g.,~\cite{EsparzaLS14} and~\cite{LengauerT80}):

\begin{cor}
\label{cor:strahler}
For binary trees,
the following two quantities
are within a constant factor from each other:
the minimum width of a simple re-pairing for $\sigma$
and the maximum height of a complete binary tree
which is a 
graph-theoretic 
minor of the tree~$D$.
\end{cor}

By Corollary~\ref{cor:strahler},
the upper bound in Theorem~\ref{th:peb} has
the same order of magnitude as
(or lower than)
the upper bound from Theorem~\ref{log-upbnd}.
The latter gives a simple re-pairing too,
but also holds for non-binary trees~$D$.

The lower bound in Theorem~\ref{th:peb}
relies on the re-pairing being simple.
For instance, for the word $Z(n)$ associated with a complete binary tree
(see~\eqref{eq:Zn-def}), 
the minimum width of a simple re-pairing is $\Theta(n)$,
but the (usual) width is $o(n^\eps)$
 for all $\eps > 0$ (section~\ref{upbndZ}).

\section{Upper bound for complete binary trees}
\label{upbndZ}

Recall
the words~$Z(n)$, defined by equation~\eqref{eq:Zn-def}
on page~\pageref{eq:Zn-def}.

\begin{theorem}\label{Zn-upbnd}
$ \Wd(Z(n))= 2^{O(\sqrt{\log n})}$.
\end{theorem}

The upper bound from the previous section
gives
$\Wd(Z(n)) = O(n)$, whilst
the functions $f(n) = a \cdot 2^{b \sqrt n}$ for $a, b > 0$
are such that $(\log n)^k = o(f(n))$ and
$f(n) = o(n^\eps)$ for all $k, \eps > 0$.


To prove Theorem~\ref{Zn-upbnd} we need a family of \emph{framed
words} $Z(n)^{(k)}$. 
Denote by $\sigma^{(k)}$ the word
\begin{equation}
\label{eq:fr}
\underbrace{++\ldots++}_{\text{$k$}} \sigma
\underbrace{--\ldots--}_{\text{$k$}}.
\end{equation}
Using the brackets terminology,
this is the word $\sigma$ which is enclosed by
$k$~pairs of openings and closing brackets.
We will call such words \emph{$k$-framed}.

\begin{Remark}
If $k\geq |\sigma| / 2$, then $\Wd(\sigma^{(k)})\leq 2$,
because a re-pairing can erase the signs of $\sigma$ from left
to right,
pairing each ${-}$ with a ${+}$ from the prefix
    and each ${+}$ with a ${-}$ from the suffix.
This re-pairing is, of course, \emph{not} simple.
\end{Remark}

We construct a family of re-pairings~$p(q,n,k)$ of framed words
$Z(n)^{(k)}$, where 
 $k\leq n$ and $1\leq q\leq (n+1)/2$ is a parameter. 
The definition will be recursive,
and $q$ will control the `granularity' of the recursion.

\paragraph*{Oveview.}

On each step of the re-pairing~$p(q,n,k)$
the leftmost remaining ${-}$ is erased. 
For $n\leq 2$,
it is paired
with the leftmost remaining ${+}$.
For $n>2$,
it is paired with the ${+}$
that we choose using the following recursive definition.

At each step of the re-pairing $p(q,n,k)$,
we define an auxiliary subsequence of  the word $Z(n)^{(k)}$ that
forms a word $Z(q)^{(k')}$. If the leftmost remaining minus is not in
the subsequence, then we pair it with the \emph{leftmost} non-erased
plus. Otherwise we consider the re-pairing $p(q',q,k')$ of the
word $Z(q)^{(k')}$,
where we pick $q'$ and $k'$ below,
and pair the minus using this re-pairing (more details to follow). 

\paragraph*{Stages of the re-pairing $p(q, n, k)$.}

The re-pairing $p(q,n,k)$ is divided into
stages, indexed by $t = 1, \ldots, 2^{n-q}$.
Denote by $Z_t$, $1\leq t\leq N= 2^{n-q}$, the $t$th leftmost
occurrence (factor) of~$Z(q)$ in the word~$Z(n)^{(k)}$.
Stage~$t$ begins at the moment when all minuses to the
left of the start position $i_t$ of the factor $Z_t$ are erased,  
and ends when stage~$t+1$ begins.
Define
an integer
sequence $k_t$ as follows:
\begin{equation}\label{k_t-def}
k_{2^q \cdot s +1} =0, \ k_{2^q \cdot s +t} = \lceil\log_2t\rceil -1\ 
\text{for } 1<t\leq 2^q, \ 0\leq s.
\end{equation}
At the beginning of stage~$t$, the subsequence $Z'_t$ of
$Z(n)^{(k)}$ is formed by the $k' = k_t$ rightmost non-erased pluses to the
left of $i_t$; followed by the symbols of the factor $Z_t$; followed by the $k_t$
leftmost non-erased minuses to the right of the end position of
$Z_t$.  The symbols of $Z'_t$ written together,
form the word $Z(q)^{(k_t)}$.

Choose $1\leq q'\leq q/3$ such that the width of the re-pairing
$p(q',q,k_t)$ is minimal. 
At the first part of stage~$t$, the re-pairing $p(q,n,k)$ pairs the
signs in $Z'_t$ according to the re-pairing $p(q',q,k_t)$. The first part ends when 
either all minuses to the left of the 
factor $Z_{t+1}$ 
are erased or 
the sequence $Z'_t$ is exhausted (whichever is earlier). In the latter case the final part of
stage~$t$ is started. At each step of this part, the leftmost
non-erased minus is paired with the leftmost non-erased plus.

\begin{prop}
Re-pairings $p(q, n, k)$ are well-defined.
\end{prop}

%
%
Define
$W_n = \min_{q} \max_{0\leq k \leq n}\Wd(p(q,n,k))$,
where the minimum is over
$15\leq q\leq n/3$ for $n\geq 45$ and over
$1\leq q\leq n/3$ for $3\leq n<45$.

\begin{prop}\label{claim:work-rec}
$ W_{n}\leq\min\limits_{15\leq q\leq n/3}\Big(
  \displaystyle\frac{2n}{q} + 2W_q + 3\Big)$ for $n\geq45$.
\end{prop}

Somewhat strangely, we have been unable to find solutions
to recurrences of this form in the literature.

\begin{prop}\label{Wn-ub}
$W_n = 2^{O(\sqrt{\log n})}$.
\end{prop}

Since $\Wd(Z(n))\leq W_n$, 
this implies the upper bound of Theorem~\ref{Zn-upbnd}.

\paragraph*{Proof idea for Claim~\ref{claim:work-rec}.}

(For complete proofs of these claims see Appendix, Section~\ref{app:upbndZ}.)
Assume $15\leq q \leq n/3$.
We notice that at each step at most two of the factors
$Z_t$, $Z_{t+1}$ are
partially erased.
(All other factors $Z_{t'}$ either have been erased completely
($t' < t$) or are yet untouched ($t' > t+1$).)
Furthermore, non-erased signs  to the
left of partially erased factors form  several intervals; each of them,
except  possibly the leftmost, has size at least~$q$.

Note that, at each moment in time,
the non-erased signs form a well-formed (Dyck) word,
so the height of each position in $Z(n)$ with respect
to these signs only is nonnegative.
Since the height of positions in the word~$Z(n)^{(k)}$ cannot exceed~$n + k \le 2 n$,
it follows that a partially erased factor $Z_t$ can be preceded by at most
$2 n/q+1$~non-erased intervals  (runs of pluses).
This leads to the recurrence of Claim~\ref{claim:work-rec}.

%

\section{Lower bounds}\label{s:lwr-bnd}

\begin{theorem}
\label{lwr-sqrtlog}
There exists a sequence of well-formed words~$W_n$
such that
\begin{equation*}
\Wd(W_n) = \Omega(\sqrt{\log |W_n|/\log\log |W_n|}).
\end{equation*}
\end{theorem}

The words in this sequence are similar to the words~$Z(n)$ associated with
complete binary trees.
They are associated with a `stretched' version of the complete
binary tree, i.e., one in which every edge is subdivided into several edges.
More precisely, let $a_0$, $a_1$, \ldots, $a_k$ be a finite sequence
of positive integers. Define the following sequence of well-formed words
inductively:
\begin{equation*}
  \begin{aligned}
    &X(a_0) = +^{a_0}-^{a_0},\\ 
    &X(a_0,\dots, a_k) =
    +^{a_k}X(a_0,\dots, a_{k-1})X(a_0,\dots, a_{k-1})-^{a_k}.    
  \end{aligned}
\end{equation*}
%
%
%
%
%
The words we use to prove Theorem~\ref{lwr-sqrtlog}
have the form
$
Y(m,\ell) = X(a_0,\dots, a_{m\ell-1})
$,
where
$a_i = 2^{\lfloor i/\ell\rfloor}$,
$ m \geq 1 $, and $ \ell \geq 1 $.
In particular,
$Y(\ell) = Y\big(\lfloor \ell\cdot\log\ell\rfloor, \ell\big)$.
(Notice that $Z(n) = X(1, \ldots, 1) = Y(1, n)$.)
%
%
Our method applies both to $Y(\ell)$ and~$Z(n)$, giving
the bounds in equation~\eqref{eq:lb} on page~\pageref{eq:lb}.

We give a proof overview below, details are provided in Appendix,
Section~\ref{app:s:lwr-bnd}.  
We use a \textbf{tree representation
of re-pairings.}
Informally, this tree tracks
the sequence of mergers of erased intervals.
This sequence, indeed, is naturally depicted as an ordered
rooted binary tree as shown in Fig.~\ref{pic:col-tree->word}(a). 
This tree is essentially a derivation tree for a word
in an appropriate matrix grammar.
Edges of a rooted tree are divided into levels according to the
distance  to the root. We think of this distance as a~moment of \emph{time} in the derivation process. 
The derived word can be read off the tree
by following the left-to-right depth-first traversal.

For formal definitions see Appendix, Section~\ref{app:s:trees:col}.

\noindent
\begin{figure}[!h]
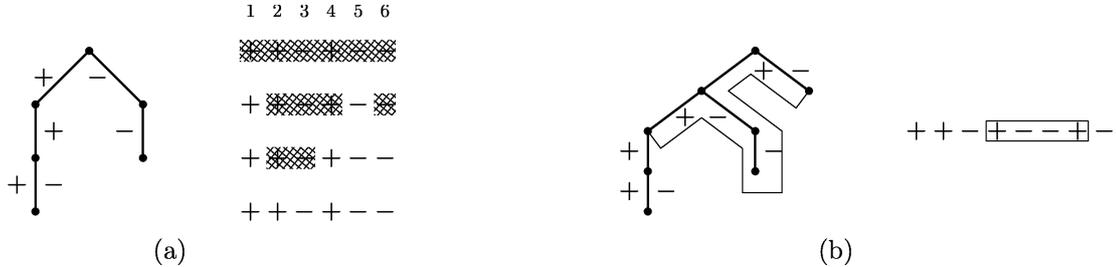

\begin{minipage}[b]{0.45\textwidth}
  \centering
  \mpfile{tr}{22}
\end{minipage}
\hfill
\begin{minipage}[b]{0.45\textwidth}
   \centering
  \mpfile{tr}{23} 
\end{minipage}
\par
%
%
  \caption{(a) Tree representation of the re-pairing
           $(2, 3), (4, 6), (1, 5)$ for the word $Z(2)$;
           (b)
           for the word $Z(2){+}{-}$,
           a fragment of the tree associated with the factor
           ${+}{-}{-}{+} $.
  }
  \label{pic:col-tree->word}
\end{figure}

Our proof is inductive,
and one of the ideas is \emph{what the induction should be over}.
Observe that
every factor $w$ of a Dyck word $W$ induces a connected subgraph,
which we call
a \emph{fragment}; see
Fig.~\ref{pic:col-tree->word}(b). 
The width of a tree or a fragment is defined in natural way: it is the
maximal number of edges at a level of the tree. E.g., the fragment
shown in Fig.~\ref{pic:col-tree->word}(b) has width~2. 

Our \textbf{inductive statement} applies to fragments.
Fix a well-formed word~$W$; in the sequel we specialize the argument
to $Z(n)$ and $Y(\ell)$.
Denote by $ L(W, k) $
the maximum length of a factor~$w$
associated with a fragment of width at most~$k$
in trees that derive the word~$W$.
Put differently,
given~$W$,
consider all possible trees that derive~$W$.
Fragments
of width at most~$k$ in these trees
are associated with factors of the word~$W$,
and $L(W, k)$ is the maximum length of such a factor.
Note that in this definition
the width of the (entire) trees is not restricted.

It is clear from the definition that the sequence of numbers~$L(W,k)$
is non-decreasing:
$
L(W,1)\leq L(W,2)\leq \dots \leq L(W,k)
$.
We obtain \emph{upper} bounds on the numbers~$L(W,k)$ by induction.
For $Z(n)$, we show that
$L(Z(n), k) = L(Z(n),k-1)^{O(k)}$
for big enough~$n$ and~$k$. Here and below, implicit constants in
the asymptotic notation do not depend on $n$ and $k$.
From this we get $L(Z(n),k) = 2^{2^{O(k)} \cdot k!}$.
We observe that if $\Wd(W)\leq k$, then $|W|\leq L(W,k+1)$.
Since $|Z(n)| = \Theta(2^n)$, it follows that
every derivation tree of the word~$Z(n)$
must have width~$k$ satisfying  $n =2^{O(k)} \cdot k!$,
that is, $\Wd(Z(n)) = \Omega(\log n / \log \log n)$.

For $Y(\ell)$, we show a stronger inequality,
$L(Y(\ell), k) \le \poly(\ell, k) \cdot (c k)^\ell \cdot L(Y(\ell),k-1)$,
which is sufficient for a lower bound
$\Wd(Y(\ell)) = \Omega(\ell)$.

To prove  the inductive upper bound
on $L(W, k)$ we need to show that narrow fragments cannot be
associated with long factors. For this purpose we use two ideas.

\paragraph*{Combinatorial properties of increases and drops in $Z(n)$ and
  $Y(\ell)$. }

Denote by
$\Delta(u)$ the difference $\Ht(j)-\Ht(i) $, where $i$ and $j$ are the
start and end positions of a factor $u$. The value $\Delta(u)$ is the increase
in height on the factor.

The first property is that
every factor of~$Z(n)$ of length~$x$
contains a sub-factor~${-}^d$ with $d \ge \log x - O(1)$.
%
%
The second combinatorial property of~$Z(n)$ is as follows:
for sufficiently large~$x$ and every two factors~$u$ and ${-}^x$ of the word~$Z(n)$,
if $\Delta(u) \ge x$ and $u$~is located to the left of ${-}^x$,
then 
the distance between these factors is at least~$2^x$.
Here and below the \emph{distance} between the factors
is the length of the smallest  factor of $W$ containing both of them.

For the word $Y(\ell)$, similar properties hold,
but  the functions $\log x$ and $2^x$ are replaced by the functions
$\Omega(\ell\cdot (x/9)^{1/\ell})$ and $\Omega((x/2\ell)^\ell)$, respectively.

\paragraph*{Balance within a single time period.}

Consider a factor~$w$ of the word~$W$ associated with a fragment of
width at most~$k$ (in a tree derivation that generates~$W$).
Denote this fragment~$F$.


Notice that, in a~Dyck word,
every ${-}$ is matched by a ${+}$ somewhere
to the left of it. Thus, for a~factor ${-}^d$,
there exists a factor~$u$ to the left of ${-}^d$ with a
matching height increase: $\Delta(u) \ge d$.
We strengthen this balance observation to identify a~pair of matching
factors $-^{d}$, $u$ 
(with a~slightly smaller height increase in~$u$) which also satisfies the following conditions:
\begin{enumerate}
\renewcommand{\theenumi}{(\alph{enumi})}
\renewcommand{\labelenumi}{\theenumi}
\item $d$ is large enough (of magnitude indicated by the first
  combinatorial property);
\item
the factors $u$ and ${-}^{d}$ are derived during overlapping time intervals,
\item
the factor $u$ sits to the left of ${-}^d$ and inside $w$, and
\item
the sub-fragment
associated with the factor between $u$ and ${-}^d$
has  width strictly smaller than the width of the entire fragment~$F$. 
\end{enumerate}
These conditions enable us to upper-bound the distance
between $u$ and ${-}^d$ through a~function of $L(W,k-1)$.
On the other hand,
this distance is lower-bounded 
by the second combinatorial property.
Comparing the bounds shows how to bound
$|w|$, and thus $L(W,k)$, from above
by a function of $L(W, k-1)$.

\section{An application: Lower bounds for commutative NFA}
\label{s:oca}

In this section we link the re-pairing problem for well-formed (Dyck) words
to the descriptional complexity (number of states in NFA)
of the Parikh image of languages recognized by  one-counter automata (OCA).

We consider a slightly simplified version of \emph{complete} languages $(\mathcal H_n)_{n \ge 2}$
introduced by Atig et al.~\cite{AtigCHKSZ16} (see Section~\ref{s:intro}).
Each of them is over the alphabet
$
\{c_{ij}: 0\leq i<j<n\}\cup \{a_i: 0\leq i<n\}
$
and can be recognized by an OCA with $n + 2$~states.
We will assume throughout that $n$~is even.
In what follows, we need only
the Parikh image~$U_n$ of this language. We call~$U_n$
the \emph{universal one-counter set}
(since it is provably the hardest case for translations from OCA to
Parikh-equivalent NFA).
This is the set of $n(n+1)/2$-dimensional vectors
$(y_{ij}, x_i) = (y_{i j} \colon 0 \le i < j < n; \ x_i \colon 0 \le i < n)$
of nonnegative integers
that satisfy the following conditions:

\begin{enumerate}
\renewcommand{\theenumi}{U\arabic{enumi}}
\renewcommand{\labelenumi}{(\theenumi)}
\item\label{U1}\label{U2}
$y_{ij}\in\{0,1\}$ and
the directed graph on vertices $[n] = \{0,1,\dots,n-1\}$
with the set of edges $\{(i,j) \text{\ such that\ } y_{ij}=1\}$
is a monotone path from $0$ to $n-1$ (i.e., one with $i<j$ in each edge);
we call such paths \emph{chains};
\item\label{U3}
(\emph{balance})
the vector $x = (x_i)$ belongs to \emph{the cone~$K$ of
balanced vectors:}
\begin{equation*}
  K = \Big\{(x_0,\dots, x_{n-1}) \colon
    \sum_{i=0}^{n-1} (-1)^ix_i=0;
  \ %
  \sum_{i=0}^k (-1)^ix_i\geq0, \ 
  0\leq k <n-1
    \Big\};
\end{equation*}
\item\label{U4}
(\emph{compatibility})
if $x_j>0$ for some $j > 0$, then $y_{ij}=1$ for some~$i$;
if $x_j>0$ for some $j<n-1$, then $y_{jk}=1$ for some $k$.
\end{enumerate}

We skip the standard definition of
nondeterministic finite automata (NFA).
The meaning of the numbers $x_i$ and $y_{i j}$ is that
they specify the number of occurrences of letters $a_i$ and $c_{i j}$
on accepting paths in (the transition graph of) the NFA.
An NFA \emph{recognizes a language with Parikh image~$U_n$} iff
for each vector from $U_n$ the NFA has an accepting path with
these counts of occurrences,
and for all other vectors no such path exists.
There exists~\cite{AtigCHKSZ16}
an NFA with $n^{O(\log n)}$~states
that recognizes a~language with
Parikh image~$U_n$.
Our goal is to prove that this superpolynomial dependency on~$n$
is unavoidable.

\begin{theorem}
\label{OCA2NFA}
Let $\sigma_n$ be
an arbitrary Dyck word of length $O(\sqrt{n})$.
Suppose an NFA $\A_n$ recognizes a
language with Parikh image~$U_n$. 
Then the number of states of $\A_n$
is at least $n^{\Omega(\Wd(\sigma_n))}$.
\end{theorem}

\begin{cor}
\label{cor:parikh-lb}
If an NFA $\A_n$ recognizes a language
with Parikh image~$U_n$, then
its number of states is
$
  n^{\Omega (\sqrt{\log n/\log\log n})}
$.
\end{cor}

Corollary~\ref{cor:parikh-lb}
follows from Theorems~\ref{lwr-sqrtlog} and~\ref{OCA2NFA}.

Since $U_n$ is the Parikh image of a language recognized
by an OCA with $n+2$~states, it follows that there is no
polynomial translation from OCA to Parikh-equivalent NFA.

Our proof of Theorem~\ref{OCA2NFA}
makes three steps:
\begin{enumerate}
\item
    In the NFA for $U_n$
    we find accepting paths $\pi_F$, parameterized by
    sets $F \sset \{0, 1, \ldots, \frac{n}{2}-1\}$, $|F| = |\sigma_n|$, and 
extract re-pairings of $\sigma_n$ from them.
Roughly speaking, the parameter~$F$ determines the set of positions~$j$
for which the path $\pi_F$ has $x_j > 0$.

Intuitively, as $\pi_F$ goes through
any strongly connected component (SCC) in the NFA,
the re-pairing erases pairs $(2 i, 2 j + 1)$ such that
a cycle in this SCC reads letters $a_{2 i}$ and $a_{2 j + 1}$.
To get a bijection between even and odd indices, we use
the Birkhoff---von Neumann theorem on doubly stochastic matrices
(see, e.g.,~\cite[p.~301]{Schrijver03}).

\item
    With every SCC $V$ in the NFA,
    we associate an auxiliary set $B(V)$.
    We show that
    each path $\pi_F$ visits an SCC $V_F$ for which
    $|B(V_F)| \ge \Wd(\sigma_n) $.
\item
    By making $F$ range in a family $\mathcal F$ of sets with low intersection,
    we ensure that no other path $\pi_{F'}$ can visit the SCC $V_F$.
    So the NFA has at least $|\mathcal F|$ SCCs, and therefore at least $|\mathcal F|$ states.
The low intersection property means that $|F_1\cap F_2|\leq d$ for all
 $F_1,F_2\in\F$. 
We choose $d =\Wd(\sigma_n)-1$; 
the  family $\mathcal F$ of size $n^{\Omega(d)}$ can be obtained
by the Nisan---Wigderson construction~\cite{NW}.

\end{enumerate}

For the details of the proof see Appendix, Section~\ref{app:oca}.

\section{Open problems}

Our work suggests several directions for future research.
The first is computing the width of $Z(n)$ as well as of other words,
closing the gap between the upper and lower bounds.
Obtaining super-constant lower bounds (for infinite families of words,
both constructively and non-constructively)
seems particularly difficult.
Our lower bound on the width of $Y(n)$ leaves a gap
between $n^{\Omega(\sqrt{\log n / \log \log n})}$
and $n^{O(\log n)}$
for the size of blowup in an OCA to Parikh-equivalent NFA translation,
and our second problem is to close this gap.

The third problem is to recover a proof of Rozoy's statement
that the Dyck language~$D_1$ is not generated by any matrix grammar
of finite index~\cite{Rozoy-j}, or equivalently by any two-way deterministic
transducer with one-way output tape~\cite{Rozoy-c}.
We expect that our lower bound construction for the width
can be extended appropriately.

Last but not least, our re-pairing game corresponds to the following
family of deterministic two-way transducers~$\mathcal T_k$ generating Dyck words.
The input to a transducer $\mathcal T_k$ encodes a~derivation tree of width~$k$,
in the sense defined in section~\ref{s:lwr-bnd}.
Symbols corresponds to layers of the tree;
there are $O(k^2)$ symbols in the alphabet that encode the branching and
$O(k^2)$ symbols that encode the positions of a pair of brackets
($+$ and $-$). The transducer~$\mathcal T_k$ simulates a traversal
of the tree and outputs the generated word; it has $O(k)$ states.
All words of width at most~$k$ are generated by~$\mathcal T_k$.

Our final problem is to determine if there exist smaller transducers that
generate all Dyck words of length~$n$, $D_1 \cap \{ {+}, {-} \}^n$,
and do not generate any words outside~$D_1$. Here $n$ is such that
all words of length~$n$ have width at most~$k$.

\section*{Acknowledgment}

We are grateful to Georg Zetzsche
for the reference to Rozoy's paper~\cite{Rozoy-j}.

This research has been supported by the Royal Society
(IEC\textbackslash{}R2\textbackslash{}170123).
The research of the second author
has been funded by the Russian Academic Excellence Project
`5-100'. Supported in  part by RFBR grant 17--51-10005 and by 
the state assignment topic no. 0063-2016-0003.

\bibliographystyle{alpha}
\bibliography{brackets-trees}

\newpage

\appendix

\section{On matrix grammars of finite index}
\label{app:rozoy}

Rozoy's 1987 paper~\cite{Rozoy-j}
is devoted to the proof that no matrix grammar
can generate all words in the Dyck language~$D_1$ using bounded-index derivations
without also generating some words outside~$D_1$. This amounts
to saying that no finite-index matrix grammar can generate~$D_1$.

Unfortunately, the proof in that paper seems to be flawed.
The main tool in the proof is an inductive statement,
relying on several technical claims.
The quantifiers in the use of one of these claims
are swapped when compared
against the formulation (and proof) of the claim.
(The earlier conference version~\cite{Rozoy-c} of that paper
uses an equivalent formalism of deterministic two-way transducers instead of matrix grammars;
while the high-level structure of the results is present in the conference
version, it does not contain the details of the proofs.)

In more detail, the main tool in Rozoy's proof is a certain inductive statement, $P(h)$,
which is part of the proof of Proposition~6.2.1.
This statement~$P(h)$, $h = 1, 2, \ldots, 2 k$, relies on several technical claims,
and in particular on Corollary~6.1.2.
This Corollary says, roughly, that, given a (part of~a) derivation tree,
a position~$x$ in the derived word, and a family of disjoint time
intervals~$C_1$, \ldots, $C_p$,
there exists an index~$i$ such that a certain quantity $\mathcal O(x | C_i)$
enjoys what one can call an averaging upper bound.
The quantity $\mathcal O(x | C_i)$ is the height of the position~$x$ computed
relative to just the symbols generated at times~$t \in C_i$, but not at other
times.

Unfortunately, the use of this Corollary does not match its formulation.
This occurs in the proofs of the base cases~$P(1)$ and~$P(2)$, and in the proof
of the inductive step.
In each of these three cases,
time intervals~$C_j$, $j \in J$, are chosen to correspond
to occurrences of the factor $Z(m)$ in~$Z(n)$, denoted by $Z(m)_j$, $j \in J$.
But no position~$x$ is chosen.
Then the Corollary is invoked to find an index~$i$ such that
$\mathcal O(x | C_i)$ enjoys an averaging upper bound,
where $x$ denotes the beginning of $Z(m)_i$, the $i$th occurrence of $Z(m)$.
It appears as if this use of Corollary would require a different formulation:
instead of ``for every position~$x$ there exists an index~$i$ such that
$\mathcal O(x | C_i)$ is upper-bounded'',
it would need a ``given positions~$x_j$, $j \in J$,
there exists an~$i$ such that for the pair $C_i, x_i$
the quantity $\mathcal O(x_i | C_i)$ is upper-bounded''.

Given the existing formulation, it seems possible that
for each of the positions $x_j$ (the start positions of the occurrences
of $Z(m)$), there indeed exists a suitable~$i$, but this~$i$ is different from~$j$.
In other words, the quantities $\mathcal O(x_j | C_i)$ may be bounded
for $j \ne i$, but not for $j = i$.
A formulation with the swapped quantifiers, ``there exists an~$i$ such that
for all~$x$'', would suffice to rule out this bad case, but
it is not clear if it is possible to prove such a statement.

Rozoy's construction has certainly influenced our lower bound construction
(for the width of re-pairings).
The family of statements~$P(h)$ in her construction amounts to
an inductive upper bound on the length of the factors~$Z(m)$ that can occur in
the derivation tree for~$Z(n)$ if this tree has width at most~$k$
(we are glossing over some technical details here).
In our proofs, instead of factors of the form~$Z(m)$ we consider arbitrary
factors in Definition~\ref{def:L}:
the quantity $L(W, k)$ upper-bounds the length of all factors (of~$W$)
that can be associated with fragments of width at most~$k$.
We do not restrict the width of the entire tree either.
The proof of the main inductive step (Lemma~\ref{main-length})
then crucially depends on an upper bound
for width~$k-1$ being available for \emph{all} words, not just for words
of a particular form, $Z(m)$.
We do not know if it is possible to use a weaker inductive statement, one
mentioning words of the form $Z(m)$ only, for a lower bound on~$\Wd(Z(n))$.

\section{Trees, Dyck words, and re-pairings}
\label{app:def:trees}

In this section we review the various use of trees in our work.
The interpretation of Dyck words as trees (subsection~\ref{s:trees:dyck}) is required
by our results on simple re-pairings and pebble games
(Section~\ref{app:s:simple} below and Section~\ref{s:simple} in the main text).
The tree representation of re-pairings (subsection~\ref{app:s:trees})
is required for our lower bounds on the width of re-pairings
(Section~\ref{app:s:lwr-bnd} below and Section~\ref{s:lwr-bnd} in the main text).

We will use \emph{ordered rooted trees}.
Let us recall the standard definitions.

A tree with an empty set of nodes (an empty tree) is denoted by~$\eps$.
Suppose the set of nodes~$I$ is nonempty; then it has a distinguished
node~$\root$ (the \emph{root} of the tree).
Consider a mapping $a\colon I\sm\{\root\}\to I$.
The node $a(v)$ is the \emph{immediate ancestor} (\emph{parent})
of a node~$v$, and the node~$v$ is the \emph{immediate descendant}
(\emph{child}) of~$a(v)$.
The transitive closure of the parent (resp.\ child) relation
is the ancestor (resp.\ descendant) relation.
No node is its own descendant,
and all non-root nodes are descendants of the root.
These are exactly the defining properties of a rooted tree.
The set of children of each node is linearly ordered.
A \emph{leaf} is a node without descendants.

A node $v$ and the set of its descendants form an ordered rooted tree
in a natural way, which will be referred to as the subtree rooted at~$v$.

If every node in a tree has at most two children, the tree is \emph{binary}.
If $a^{-1}(v) = (v_1, v_2)$, then the node $v_1$ is the \emph{left child}
and $v_2$ is the \emph{right child}.

For a non-root node~$v$, denote by $s(v)$ the \emph{sibling} of $v$,
that is, the node with the same parent
different from~$v$, if such a node exists in the tree.

\subsection{Dyck words and trees}
\label{s:trees:dyck}

Well-formed words are naturally \emph{associated} with
\emph{ordered rooted forests} (i.e., with sequences of ordered
rooted trees) in the following way.
Well-formed words are exactly those generated by the unambiguous
grammar
\[
S \to \eps \ | \ {+}S{-}S.
\]
The rooted forest~$T(\sigma)$ associated with the word~$\sigma $
is defined recursively using the derivation in this grammar:
$T(\eps)=\eps$, and
$$T({+}\sigma_1{-}\sigma_2) = (T({+}\sigma_1{-}), T(\sigma_2)),$$
where $T({+}\sigma_1{-})$ is the tree in which the children of the root
are the roots of the trees from the rooted forest~$T(\sigma_1)$.
The two symbols $+$ and $-$ in the word ${+}\sigma_1{-}$
are associated with the root of the tree~$T({+}\sigma_1{-})$;
these symbols are also said to be matched to each other.

This way all symbols in the word~$\sigma$ are split into pairs
of symbols matched to each other, and every pair is associated
with a node in the rooted forest~$T(\sigma)$.

\subsection{Tree representation of re-pairings}
\label{app:s:trees}

To prove lower bounds on the width,
we will need a different representation
of re-pairings.
Informally, this representation will track
the sequence of mergers of erased intervals.
This sequence is naturally depicted as an ordered
rooted binary tree.
The width of the tree will match the width
of the re-pairing (in that the two numbers will
differ by at most~$1$).
This tree is essentially a derivation tree for the word
in an appropriate matrix grammar.

\subsubsection{Trees and traversals}
\label{app:s:trees:traversal}

Here and below, all trees wil be binary and
edges in trees will be directed from root to leaves.
That is, in an ordered rooted binary tree,
the set~$E$ of edges consists of ordered pairs of the form~$(a(v),v)$.
A ranking function~$T$ will be defined on this set,
and we will call it the time function.
For edges~$e$ departing from the root,
the time~$T(e)$ is equal to~$1$;
for all pairs of edges $(x,y)$, $(y,z)$ with a common endpoint
the time function satisfies $T(x,y)+1 = T(y,z)$.
Whenever $T(e) =t$, we will say that the edge~$e$
\emph{exists} at time~$t$.

The \emph{width of a tree}~$(I,E)$
%
%
%
with the set of nodes~$I$ and set of edges~$E$
is now defined as the maximum (over all time points~$t$)
number of edges existing at time~$t$, i.e.,
\[
\Wd(I, E) = \max_{t}\big|T^{-1}(t)\big|.
\]
The \emph{width} of an arbitrary subset of edges, $\Wd(I, E')$
for $E' \subseteq E$, is defined analogously.

Recall the definition of the
\emph{(left-to-right depth-first) traversal}
of the tree
(we denote this traversal by $\tau(I,E)$).
This is a sequence of edges 
in which every edge occurs
twice; it is defined recursively as follows.
For a tree that consists of the root node only (and no edges),
the traversal is the empty sequence.
If $v_1$ and $v_2$ are the two children of the root,
then
\[
\tau(I,E) = (\root,v_1)\cdot \tau(I_1, E_1)\cdot (\root,v_1)\cdot
(\root,v_2)\cdot \tau(I_2, E_2)\cdot (\root,v_2),
\]
where $\cdot$ denotes the concatenation of sequences,
and $(I_1,E_1)$ and $(I_2,E_2)$ are the subtrees rooted at $v_1$
and $v_2$, respectively.
Similarly, if there is only one child~$v_1$, then
\[
\tau(I,E) = (\root,v_1)\cdot \tau(I_1, E_1)\cdot (\root,v_1),
\]
where $(I_1,E_1)$ is the subtree rooted at $v_1$.
Every tree has exactly one traversal.

An example of a traversal is shown in Fig.~\ref{app:pic:traverse}.

\begin{figure}[!h]
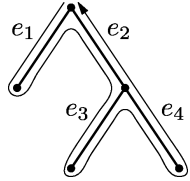

  \centering
    \mpfile{tr}{1}
  \caption{Tree traversal: $\tau(I,E) = (e_1,e_1, e_2,e_3,e_3,e_4,e_4, e_2)$}
  \label{app:pic:traverse}
\end{figure}

For every~$t$, edges that exist at time~$t$
cut the traversal as follows.

\begin{prop}
\label{app:t-division-of-traverse}
The edges $e_1$, $\dots$, $e_k$ of the tree $(I,E)$
that exist at time~$t$ occur in the traversal
in pairs:
\begin{equation}\label{app:t-decomposition}
\tau(I,E)= U_0 e_1 D_1 e_1 U_1 e_2 D_2 e_2\dots
U_{k-1} e_k D_k e_k U_k\;.
\end{equation}
All edges from $U_i$ exist at times before~$t$,
and all edges from $D_i$ at times after~$t$.
\end{prop}

\begin{proof}
Assume that the edges $e_i$ are sorted
according to the indices
of their first occurrences in the traversal.

Let $v_i^<, v_i^>$ be the two endpoints of the edge~$e_i$,
so that $v_i^<$ is closer to the root than $v_i^>$.
It is clear from the definition of the traversal that
the part of the traversal between the two occurrences of $e_i$---%
denote this part by~$D_i$---%
is the traversal of the subtree rooted at~$v_i^>$;
therefore, all edges in that part exist at times strictly greater
than (after)~$t$.

Since an undirected tree has exactly one path between any two edges,
we have $D_i\cap D_j=\es$ whenever $i\ne j$.

For the same reason, the part~$U_i$ of the traversal
between the second occurrence of~$e_i$ and the first occurrence
of~$e_{i+1}$ does not touch any edges from subtrees rooted at~$v_k^>$.
The same also holds for the initial and final parts of the traversal:
$U_0$ from the beginning of the traversal up until~$e_1$
and
$U_k$
from the second occurrence of~$e_k$ to the end of the traversal.
\end{proof}

\subsubsection{Tree derivations}
\label{app:s:trees:col}

Tree derivations, or derivations,
are similar to re-pairings defined in section~\ref{s:def}.
Informally, at any time point~$t$,
one~$+$ and one~$-$ are placed on edges
that exist at time~$t$;
the plus must be placed to the left of the minus.
A sign can be placed on an edge in two ways:
on the left side of the edge or on the right side of it.
This is shown  in Fig.~\ref{app:pic:col-tree}.

\begin{figure}[!h]
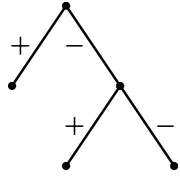

  \centering
    \mpfile{tr}{2}
  \caption{Tree derivation}
  \label{app:pic:col-tree}
\end{figure}

The formal definition is as follows.
Let $(I, E)$ be a ordered rooted binary tree.
The traversal $\tau(I,E)$ defines a function
\[
\tau(I,E)\colon \{1,\dots, 2|E|\}\to E
\]
in natural way (any sequence is formally a function of this form). 

Consider partial functions of the form
\[
\pi\colon \{1,\dots, 2|E|\}\to \{-1,+1\}.
\]
These functions are \emph{partial words} (or patterns) over the
alphabet~$\{-1,+1\}$. We define by $\Dom\pi$ the \emph{domain} of the
partial function, i.e., the set of all $i$ such that $\pi(i)$ is defined.

For any subset $S\subseteq \{1,\dots, 2|E|\}$ a partial word $\pi$ defines
a~word $\pi(S)$  over the
alphabet~$\{-1,+1\}$ by the rule: if $S\cap \Dom\pi = \{i_1,i_2,\dots,
i_\ell\}$, where $i_1<i_2<\dots<i_\ell$, then 
\[
\pi(S) = \pi(i_1)\cdot\pi(i_2)\cdot\ldots\cdot \pi(i_\ell).
\]

Let $E_t$ be the set of places in the traversal $\tau(I,E)$
occupied by edges existing at time~$t$. Formally,
\[
E_t = \{i : T(\tau(I,E)(i) ) = t\}.
\]
For brevity we  use notation  $\pi_t = \pi(E_t)$.

Given a tree~$(I,E)$,
a \emph{tree derivation}
(or a \emph{derivation})
is a partial word $\pi$ 
such that, for all points in time~$t$,
either the word~$\pi_t$ is empty,
or it is ${+}{-}$.

Every edge $e$ occurs in the traversal exactly twice and thus may
\emph{derive} at most two signs in a tree derivation~$\pi$. Let
$\tau(I,E)^{-1}(e) = \{i,j\}$, where $i<j$. If  $\pi(i)$ is defined,
then the sign $\pi(i)$ is said to be placed on the left side of the
edge. If $\pi(j)$ is defined, the sign $\pi(j)$
is said to be placed on the right side of the edge.





For an example of a derivation  shown in Fig.~\ref{app:pic:col-tree},
the corresponding partial function is as follows:
\[
\begin{array}{l|c|c|c|c|c|c|c|c|}
  i&1&2&3&4&5&6&7&8\\
\hline
\pi&+1&&-1&+1&&&-1& 
\end{array}
\]
Blank spaces in the table mean that the function is not defined.

Let $\pi$ be a derivation
based on the tree~$(I,E)$. 
We will say that the word~$\sigma(\pi, I,E) = \pi(\{1,2,\dots, 2|E|\})$
is \emph{generated} by the derivation~$\pi$,
or, alternatively, that the word is \emph{derived} by the tree.
We will also sometimes say that a sign in the word
is \emph{derived} at time~$t$
if it is derived by an edge that exists at time~$t$.

\subsubsection{Fragments}
\label{app:s:trees:fragments}

Consider an arbitrary factor of the tree traversal~$\tau(I, E)$.
The set of edges that occur in this factor at least once
forms a subgraph in the tree~$(I, E)$, which we will call
a \emph{fragment} of this tree.
It is easy to check that all fragments
are (weakly) connected subgraphs of the graph~$(I, E)$.
Thus, a fragment is a tree in graph theory terms. 

Let $\pi$ be a derivation based on a tree $(I, E)$. An~interval
$[i,j]\subseteq \{1,2,\dots, 2|E|\}$ 
defines a~factor  $w = \pi([i,j])$ of the word $\sigma(\pi, I,
E)$ and each factor of  $\sigma(\pi, I, E)$ can be represented in this
form. Moreover, w.l.o.g. we will assume that $i,j\in\Dom\pi $.

Suppose that  $w=\pi([i,j])$ for a factor $w$ of the word $\sigma(\pi,
I, E)$. The interval $[i,j]$ selects a~factor in the tree
traversal. We say that the fragment $F_w$ corresponding to this factor of
the tree traversal \emph{is associated} to the factor~$w$. 

Note that the only one fragment is associated to a~factor of
$\sigma(\pi, I, E)$ but the converse is not true. A~fragment is a set
of edges by definition. It may correspond to different factors of the
tree traversal because each edge occurs twice in the traversal.

\subsection{Derivations and re-pairings}
\label{app:s:trees:link}

We now show how to link re-pairings of words and derivations.

Let $\pi$ be a derivation based on a tree $(I, E)$.
Consider the word $\sigma(\pi, I, E)$.
Take any point in time~$t$ for which
this word has two (\emph{paired}) signs;
they are said to be \emph{derived} at this time point.
Denote by $\ell_t$ the index of the $+$ sign
from this pair in the word $\sigma(\pi, I,E)$;
and by $r_t$ the index of the $-$ sign.
We obtain a sequence
\[
p_\pi = ((\ell_{t_{\max}},r_{t_{\max}}),\;\dots, \;(\ell_{t_{1}},r_{t_{1}})),
\]
where $t_{\max}$ is the maximum point in time
at which edges of the tree exist.
The time in this sequence flows in direction opposite
to that in the tree; indeed,
as we already mentioned at the beginning of section~\ref{app:s:trees},
the tree depicts a sequence of mergers of erased intervals.

Fig.~\ref{pic:col-tree->word}(a) shows an example
where a re-pairing of a word is obtained
from a tree derivation.

\begin{prop}
The word $\sigma(\pi, I,E)$ is well-formed,
and the sequence $p_\pi$ is its re-pairing.
The width of~$p_\pi$ does not exceed
the width of the tree~$(I,E)$.
\end{prop}

\begin{proof}
It is easy to see that the word $\sigma(\pi, I,E)$ is well-formed
and the sequence $p_\pi$ satisfy the requirements on re-pairings.
Indeed, if  $\pi_t={+}{-}$, then the plus
is placed to the left of the minus by definition of words $\pi(S)$.

Consider a time point~$t$.
Suppose $k$~edges exist at this time,
$e_1$, \ldots, $e_k$ if read from left to right.
It follows from Claim~\ref{app:t-division-of-traverse} that
the re-pairing $p_\pi$ has, at this time point,
at most~$k$ erased intervals in the word $\sigma(\pi, I,E)$.
Indeed, in terms of the factorization~\eqref{app:t-decomposition}
from Claim~\ref{app:t-division-of-traverse},
each interval is formed by the symbols of the word $\sigma(\pi, I,E)$
that correspond to the edges from the subsequence $D_i$
(this is the traversal of the subtree of descendants
 of the edge~$e_i$).

Some of these intervals can be empty,
however, they do in any case cover all signs
that have been erased by time~$t$.
\end{proof}

We now describe a link 
in the opposite direction.
Let $p$ be a re-pairing of the word~$\sigma$.
We will construct a tree $(I,E)$ of width~$\leq \Wd(p)+1$
and a derivation~$\pi$ based on this tree such that~$p_\pi=p$.

Let us first construct a tree $(I',E')$
in which the nodes are pairs of the form
\[
\text{(maximal erased interval, time point).}
\]
The time here is specified according to the sequence~$p$.
The intervals are specified by their start and end positions.
This way a node $([s,f],t)$ corresponds to a maximal
interval~$[s,f]$, between positions~$s$ and~$f$,
which has been erased completely by time~$t$.

The ancestor---descendant relation on the pairs of this form
is the conjunction of the set inclusion for intervals
and the ordering on time points:
a pair $([s_1,f_1], t_1)$ is a descendant of a pair $([s_2,f_2], t_2)$
if and only if
\[
[s_1,f_1] \subseteq [s_2,f_2]\quad \text{and}\quad t_1<t_2.
\]
Here we have taken into account that the time on the tree
and in the re-pairing flows in opposite directions.

All possible ways in which the erased intervals can evolve
when a new pair $(\ell_t,r_t)$ of signs is erased
in the re-pairing~$p$ are shown in Fig.~\ref{app:pic:evolution}.

\begin{figure}[!h]
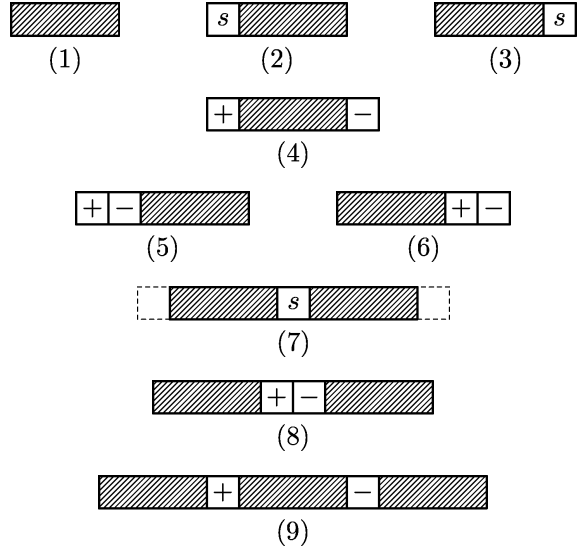

  \centering
    \mpfile{tr}{50} \hskip1cm\mpfile{tr}{51} \hskip1cm \mpfile{tr}{52} 

    \vskip3mm

    \mpfile{tr}{53}

    \vskip3mm

 \mpfile{tr}{54} \hskip1cm \mpfile{tr}{55}

    \vskip3mm

    \mpfile{tr}{56}

    \vskip3mm

\mpfile{tr}{57}

 \vskip3mm

 \mpfile{tr}{58}

  \caption{Possible steps in the evolution of the intervals, $s\in\{+1,-1\}$}
  \label{app:pic:evolution}
\end{figure}

All intervals not adjacent to the freshly erased (paired) signs
remain unchanged (case~(1)).
The node of the tree $(I',E')$ that corresponds to this interval
at the time right after these two signs are paired
has exactly one child.
The cases~(2)--(6) are similar:
in these cases the interval evolves, but does not merge with other intervals.
For this reason, the nodes corresponding to such intervals at this time
have exactly one child too.
As a special case, these scenarios can capture
a new interval emerging
(the dashed area in the picture is empty),
that is,
a new leaf appearing in the tree $(I',E')$.

In the case~(7) one of the freshly erased (paired) signs
is adjacent to two intervals.
Dashed lines in the picture indicate positions of signs
that may also be paired at this time.
The node of the tree $(I',E')$ that corresponds to the new
interval has two children in this case.
Case~(8) shows another way two intervals can merge,
in which there are two signs between the intervals,
and they are paired with one another at this point.
There are no further cases of two-interval mergers:
indeed, since just two signs are erased at a single time point,
there may be at most two signs between the merging intervals.

Finally, case~(9) shows the last possible scenario,
one where three intervals are merged into a single one.
The relative position of signs is determined uniquely in this case.

We will now transform the tree $(I',E')$ into a binary tree $(I,E)$
and a derivation~$\pi$ based on the latter tree
such that $p_\pi=p$.
As seen from Fig.~\ref{app:pic:evolution},
the tree $(I',E')$ may have nodes with three children,
in which case auxiliary nodes need to be inserted so that
the tree $(I,E)$ would be binary.

The way the signs are positioned on the edges of the tree $(I,E)$
is based on Fig.~\ref{app:pic:evolution} and shown in
Fig.~\ref{app:pic:pairing}.

\begin{figure}[!h]
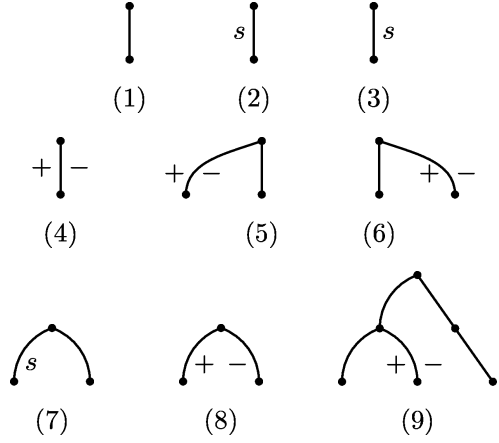

  \centering
    \mpfile{tr}{60} \hskip1cm\mpfile{tr}{61} \hskip1cm \mpfile{tr}{62} 

    \vskip 3mm

    \mpfile{tr}{63}\hskip1cm  \mpfile{tr}{64} \hskip1cm \mpfile{tr}{65}

    \vskip3mm

    \mpfile{tr}{66} \hskip1cm \mpfile{tr}{67}\hskip1cm \mpfile{tr}{68}

  \caption{Constructing the re-pairing $\pi$, $s\in\{+1,-1\}$}
  \label{app:pic:pairing}
\end{figure}

Note that due to the case (9) there is no bijection between the points in time
in the tree and the points in time in the re-pairing.
Still, the ordering of erased pairs in time in the tree
is the reverse of the ordering in the re-pairing~$p$.
This ensures the equality $p_\pi=p$.

Another observation is that in the cases~(5)--(6)
the width of the tree becomes greater by~$1$ than
the width of the re-pairing~$p$, because the signs are placed
on (auxiliary) edges that lead to auxiliary leaves.
Informally, one can think of this as of
a new interval (of $+$ and $-$) being ``born'' first
and merged with an old interval second.
In the re-pairing~$p$ these events occur simultaneously.

Combining all the arguments above,
we obtain a new characterization of the width
of a~well-formed word.


\begin{theorem}
\label{app:th:pm1}
The difference between
the width of a well-formed word~$\sigma$
and the minimum width of a tree~$(I,E)$
for which there exists a derivation~$\pi$
generating the word~$\sigma$, i.e., $\sigma = \sigma(\pi, I,E)$,
is at most~$1$.
\end{theorem}

\section{Proofs for Section~\ref{s:simple}}
\label{app:s:simple}

\subsection{Proof of the logarithmic upper bound}
\label{app:s:simple:log-ub}

In this subsection we prove Theorem~\ref{log-upbnd}.
We will use several observations.

\begin{prop}
\label{app:forest}
Let $\sigma = \sigma_1\cdot \sigma_2\cdot \ldots\cdot \sigma_t$
be a factorization of a~well-formed word into well-formed words.
Then
  \[
  \Wd(\sigma)\leq 1+ \max_{1\leq i\leq t}\big(\Wd(\sigma_i)\,\big).
  \]
\end{prop}

\begin{proof}
Denote by $p_i$ an optimal re-pairing of a word $\sigma_i$.

Consider a re-pairing $p$ of $\sigma$ which erases
the words $\sigma_1$, $\sigma_2$, \ldots, $\sigma_t$
consecutively and
according to the optimal re-pairings~$p_1, \ldots, p_t$.
The width of this re-pairing
at all times when the word $\sigma_i$ is being re-paired
cannot exceed the number of erased intervals in the re-pairing~$p_i$
plus possibly an additional interval
that consists of the fully erased word~%
$\sigma_1\cdot\ldots\cdot \sigma_{i-1}$.
\end{proof}

\begin{prop}
\label{app:simple-cut}
Let a well-formed word $\sigma$ factorize as $L\cdot \pi_1\cdot R$,
where $\pi_1$ is a well-formed word.
Then the word $\pi_2=LR$ is also well-formed
and
\[
\Wd(\sigma)\leq \max(\Wd(\pi_1),1+ \Wd(\pi_2)).
\]
\end{prop}

\begin{proof}
The first statement is obvious.

Consider a re-pairing $p$ of the word $\sigma$
which first erases the word $\pi_1$ (optimally, using a re-pairing $p_1$)
and then the word $\pi_2$ (also optimally, using a re-pairing $p_2$).

Before the beginning of $\pi_2$'s re-pairing, the width
of $p$ does not exceed $\Wd(p_1)$.
From that point on, the word $\pi_1$ is erased completely,
so the width of $p_2$ can be increased by at most~$1$
(which corresponds to the erased interval~$\pi_1$).
\end{proof}

\begin{proof}
[Proof of Theorem~\ref{log-upbnd}]

Use induction on the length of the well-formed word,
$N = |\sigma|$.
Construct a sequence of nested factors of $\sigma$ in the following way.

Any Dyck word is a concatenation of Dyck primes, i.e. words of the
form ${+}w{-}$, where $w$ is a Dyck word.

If $\sigma=\sigma_1 \ldots \sigma_m$, where $\sigma_i$ are Dyck
primes, then pick  the factor $\sigma^{(1)}=\sigma_j$ of maximum width.

Since $\sigma^{(1)}$ is  a Dyck prime, $\sigma^{(1)} 
= {+}\sigma_1^{(1)}\dots\sigma_{m(1)}^{(1)}{-}$,
where  $\sigma_i^{(1)}$ are Dyck primes.

If some $\sigma_j^{(1)}$ has length greater than $N/2$,
set $\sigma^{(2)} = \sigma_i^{(1)}$. Note that
$\sigma^{(1)} = L_2\sigma^{(2)} R_2$, and  $|L_2R_2|<N/2$.

Repeating this procedure produces a sequence of Dyck primes
$\sigma^{(i)}$: if $\sigma^{(i)}
={+}\sigma_1^{(i)}\dots\sigma_{m(i)}^{(i)}{-} $ and $|\sigma_j^{(i)}|>N/2$ 
for some $1\leq j\leq m$, then 
$\sigma^{(i+1)} = \sigma_j^{(i)}$. 
It is easy to check that $\sigma^{(1)} = L_{i+1}\sigma^{(i+1)}
R_{i+1}$ and $|L_{i+1}R_{i+1}|<N/2$.

The proccess stops on  a Dyck prime $\sigma^{(f)}$ such that
  \[
  \sigma^{(f)} ={+}\sigma_1^{(f)}\dots\sigma_{m(f)}^{(f)}{-}
  \]
where $|\sigma_{j}'|\leq N/2$ for all $1\leq j\leq m(f)$.
We have $\sigma^{(1)} = L_f\sigma^{(f)} R_f$ and $|L_fR_f|<N/2$.

Applying Claim~\ref{app:forest} to the factorization of~$\sigma$,
Claim~\ref{app:simple-cut} to the factorization $\sigma^{(1)} = L_f\sigma^{(f)} R_f$,
and then again Claims~\ref{app:simple-cut} and~\ref{app:forest}
to the factorization of $\sigma^{(f)}$,
obtain the inequality
  \[
  \Wd(\sigma)\leq 1+\Wd(\sigma^{(1)})\leq 1+1+1+\max_{|\tau| \leq N/2}\Wd(\tau).
  \]
It is now easy to see that
$\max\limits_{|\sigma| = N} \Wd(\sigma)\leq 3\log_2 N$.
\end{proof}

\subsection{Simple re-pairings and pebble games: Proof of Theorem~\ref{th:peb}}
\label{app:s:simple:peb}

In this subsection we prove Theorem~\ref{th:peb}.

We will rely on
the interpretation of Dyck words as trees, discussed in Section~\ref{s:trees:dyck}.

%
%
We will use slightly modified rules of the black-and-white pebble game,
where instead of moves of the form~\ref{bw-move:remove-white}
we will only permit those of the form
\begin{enumerate}
\renewcommand{\theenumi}{(M\arabic{enumi}')}
\renewcommand{\labelenumi}{\theenumi}
\setcounter{enumi}{3}
\item\label{bw-move:replace-white}
replace a white pebble on a node with a black pebble,
provided that all its immediate predecessors
carry pebbles.
\end{enumerate}
It is easy to see that replacing the rule~\ref{bw-move:remove-white}
with the rule~\ref{bw-move:replace-white} leaves
the value of~$\mathrm{bw}(G)$ unchanged.

Note that the variant of the game considered by Lengauer and Tarjan~\cite{LengauerT80}
also permits two types of moves where pebbles can be moved.
Such a move can be replaced by two moves that our rules admit;
the value of $\mathrm{bw}(G)$ either remains unchanged,
or increases by~$1$.
Since we are only interested in asymptotic bounds,
all results carry over between the two variants.

Consider a simple re-pairing~$p$ of a word~$\sigma$.
Construct, based on~$p$, a new re-pairing~$p'$ of the same word
as follows.
Each time $p$ erases a pair of signs associated with
a non-root node~$v$
(in the tree associated with~$\sigma$),
the new re-pairing~$p'$ will erase \emph{two} pairs of signs
associated with nodes $v$ and $s(v)$
(if these signs have not been erased previously
 and if, for $s(v)$, this node exists in the tree;
 see Fig.~\ref{app:fig:signs-triple}).
It is easy to see that the re-pairing~$p'$ is simple.

\begin{figure}[!h]
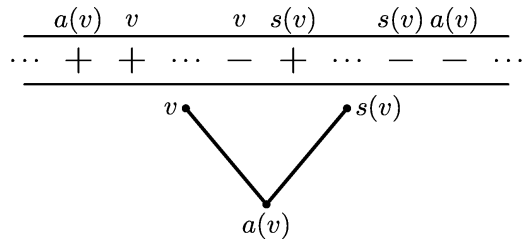

\centering
\mpfile{tr}{8}
  \caption{Signs associated with a node and its children}
  \label{app:fig:signs-triple}
\end{figure}

\begin{prop}
\label{app:prop:peb:modify-peb}
$\Wd(p') \le 3 \cdot \Wd(p)$.
\end{prop}

\begin{proof}
Consider a point in time when $p'$ erases two pairs of signs
associated with nodes~$v$ and~$s(v)$.
One of the signs associated with $s(v)$ is adjacent to
one of the signs associated with~$v$ and,
therefore, does not increase the number of erased intervals
at any point compared to~$p$.
The other sign associated with~$s(v)$ can
join an existing interval or form a new, separate interval.
In the latter case we will say that this sign \emph{becomes special}.
It \emph{stops being special} when
the sign adjacent to it and associated with the node~$a(v)$
is erased.
The position of these signs is shown in Fig.~\ref{app:fig:signs-triple}.

It is easy to see that at all times
the number of erased intervals in~$p'$ cannot be
greater than the sum of the number of erased intervals in~$p$
and the number of special signs.
Indeed, every sign that has stopped being special
has already joined another interval.
(Note that the converse does not hold.)

We now bound the number of special signs from above.
Consider an arbitrary point in time.
Let the sign $\sigma(i)$ be special
and associated, as previously, with the node~$s(v)$.
Then the signs associated with nodes~$v$ and~$s(v)$
are already erased, and
the signs associated with~$a(v)$ are not.
Observe that in this case the position between
the two adjacent signs associated with~$v$ and~$a(v)$
is an endpoint of an erased interval at this point in time.
Moreover, this position is an endpoint of an erased interval
in the re-pairing~$p$ too.
If we assign this position to the special sign~$\sigma(i)$,
then no position will be assigned to two or more signs,
and so the number of special signs will be at most
double the number of erased intervals in the re-pairing~$p$.
Combined with the argument above, this
gives the inequality
$\Wd(p') \le 3 \cdot \Wd(p)$.
\end{proof}

Based on the re-pairing~$p'$,
construct a strategy in the black-and-white pebble game
on the associated tree,
using the following rules:
\begin{enumerate}
\renewcommand{\theenumi}{(P\arabic{enumi})}
\renewcommand{\labelenumi}{\theenumi}
\item
\label{app:bw-r:place-white}
when a pair of signs is erased,
place a white pebble on the associated node~$v$;
\item
\label{app:bw-r:replace-black}
every time and everywhere,
when a node~$v$ carries a white pebble
and all children of~$v$ carry pebbles,
replace the pebble on~$v$ with a black one;
\item
\label{app:bw-r:path-black}
every time and everywhere,
when the tree has a pair of nodes~$v$ and~$w$
such that $v$ is a descendant of~$w$,
both carry black pebbles,
and the path between them has no pebbles,
remove the pebble from~$v$.
\end{enumerate}
Rules~\ref{app:bw-r:replace-black} and~\ref{app:bw-r:path-black}
are always applied until neither is applicable;
in other words, rule~\ref{app:bw-r:place-white} is applicable
only if rules~\ref{app:bw-r:replace-black} and~\ref{app:bw-r:path-black} are not.
We emphasize that our strategy must comply with
this restriction.

Let $v$ be a node in the tree~$D$.
Consider the path from~$v$ to the root of the tree.
Let $u$ be the first node on this path that carries a pebble.
We will use the following notation:
$\chi(v)$ is $1$ if the pebble on~$u$ is black
and $0$ if the pebble is white or if there is no such node~$u$.

\begin{prop}
At all points in time the following \emph{invariant} holds:
for all pairs of (matching) signs in the word~$\sigma$,
the pair has been erased 
if and only if the associated node~$v$
carries a pebble or satisfies the equality $\chi(v) = 1$.
\end{prop}

\begin{proof}
Use induction on time (the number of performed moves).
At the beginning nothing is erased and $\chi(v) = 0$ for all~$v$.
We will now consider an application of each of the three rules
that define our construction of the strategy in the
pebble game.

First suppose rule~\ref{app:bw-r:place-white} is applied.
By the inductive hypothesis, we have $\chi(v) = 0$,
so placing a pebble onto~$v$ does not change $\chi(v)$
for any node in the tree.

Now suppose rule~\ref{app:bw-r:replace-black} is applied.
Under the conditions of the replacement,
the set of nodes that carry pebbles is unchanged,
and for all other nodes the value of~$\chi$ remains the same
as previously.

Finally, suppose rule~\ref{app:bw-r:path-black} is applied.
It is clear that the invariant will continue to hold for the node~$v$;
and for the nodes~$u$ that had $\chi(u) = 1$ due to
the black pebble on~$v$ this value will remain unchanged
due to the black pebble on~$w$.
\end{proof}

\begin{prop}
\label{app:prop:peb:end}
After completion of the re-pairing,
the root carries a black pebble
and there are no white pebbles on the tree.
\end{prop}

\begin{proof}
Consider the point in time when the re-pairing is complete
and rules~\ref{app:bw-r:replace-black} and~\ref{app:bw-r:path-black} are not
applicable.
The root carries a pebble by the invariant,
because it has no outgoing edges.
Let us prove that there are no white pebbles on the tree.
Assume for the sake of contradiction that
$v$~is a node which carries a white pebble
and the descendants of which carry no white pebbles.
This node cannot be a leaf, because, if it were,
rule~\ref{app:bw-r:replace-black} would be applicable.
On the other hand, all children of~$v$ carry pebbles,
since the associated pairs of signs are erased
and the invariant holds.
But then rule~\ref{app:bw-r:replace-black} is applicable
to~$v$. This contradiction concludes the proof.
\end{proof}

Claim~\ref{app:prop:peb:end} means that the constructed strategy
is successful.

\begin{prop}
\label{app:prop:peb:peb-from-col}
The number of pebbles used by
the strategy does not exceed
$4 \cdot \Wd(p') + 3$.
\end{prop}

\begin{proof}
First consider an arbitrary point in time
when rules~\ref{app:bw-r:replace-black} and~\ref{app:bw-r:path-black} are
not applicable and, for every node~$v$,
if the associated pair of signs has been erased,
then the pair of signs associated with~$s(v)$ has also been erased
(if this node exists in the tree).
By our choice of the re-pairing~$p'$,
the rule~\ref{app:bw-r:place-white} is applied at most twice
between any two consecutive points in time with this property.

We now find two adjacent signs in the word~$\sigma$,
exactly one of which has been erased,
(a)~for each white pebble and
(b)~for each black pebble lying on a node
which is not a root and which has no descendants that carry white pebbles.

First consider scenario~(a).
If a node~$v$ carries a white pebble,
then $v$ cannot be a leaf, because of rule~\ref{app:bw-r:replace-black}.
By the same rule, at least one child~$u$ of this node
carries no pebble.
Since the invariant holds,
the two signs associated with the node~$u$
have not been erased yet,
but the signs associated with~$v$ have.
This gives us the two desired signs:
the erased one is associated with the node~$v$,
and the surviving one with its child~$u$.

Now consider scenario~(b).
Let $v$ be a black pebble satisfying the condition
in question.
The signs associated with~$v$ are erased;
by our choice of the re-pairing~$p'$,
so are the signs associated with~$s(v)$,
if this node exists in the tree.
Now observe that the signs associated with~$a(v)$
have not been erased yet:
otherwise, due to the invariant,
either the node~$v$ carries a black pebble
(this pebble cannot be white because of rule~\ref{app:bw-r:replace-black}),
or the condition $\chi(a(v)) = 1$ is satisfied
and there is a black pebble on one of the ancestors of~$v$;
but then in both cases rule~\ref{app:bw-r:path-black} is applicable.
Therefore, in this scenario we have also found two adjacent signs,
associated with nodes~$v$ and~$a(v)$, where the first has been
erased and the second has not.

It is straightforward that all the pairs of adjacent signs
that we have found are distinct.
Indeed, in both scenarios these signs are associated with
pairs of nodes in the tree, one of which is a child of the other.
But the erased sign is associated with the parent
node in scenario~(a), and with the child node in scenario~(b).

Therefore, the sum of the number of white pebbles
and the number of black pebbles,
except those that lie on the root or on a node
which has a descendant carrying a white pebble,
does not exceed double the number of the erased intervals.
Now observe that for all excepted black pebbles
(apart from one on the root---if there is such a pebble)
the path to the descendant carrying a white pebble
cannot meet a black pebble before a white one
(due to rule~\ref{app:bw-r:path-black}).
Charge this white pebble to our black pebble;
clearly, each white pebble can be charged at most one
black pebble.
So the total number of pebbles cannot exceed by more than~$1$
the number of erased intervals times~$4$.

We have considered only 
a subset of (the set of) all time points.
At all other times the number of pebbles cannot be greater
by more than~$2$, because of the form of rule~\ref{app:bw-r:place-white}.
This completes the proof.
\end{proof}

Combining
Claims~\ref{app:prop:peb:modify-peb} and~\ref{app:prop:peb:peb-from-col}
(as well as Claim~\ref{app:prop:peb:end})
now gives the lower bound in Theorem~\ref{th:peb}.

The upper bound is given by the following statement:

\begin{prop}
For every successful strategy in the black-and-white pebble
game on a binary tree
there exists a re-pairing of the associated word
with width at most twice the number of pebbles.
\end{prop}

\begin{proof}
Use the results of the papers~\cite{Loui79,Meyer-a-d-H81,LengauerT80}
(see also~\cite[p.~526--528]{savage}):
if the graph~$G$ is a tree, then for every successful strategy
in the black-and-white pebble game on~$G$ there exists another
successful strategy that uses at most twice the original number of
pebbles, but no white pebbles at all.
Given this fact,
it suffices to construct a re-pairing with the required width
just for strategies using black pebbles only.

With no loss of generality, we consider only the strategies having
 an auxiliary property.
Every time a black pebble is removed from a node~$v$
using rule~\ref{bw-move:remove-black},
we will assume that
the path from~$v$ to the root has a black pebble.
Indeed, if this is not the case, then the strategy
in the pebble game can be modified as follows:
instead of removing the pebble from~$v$,
we could have refrained from putting it on~$v$ altogether;
this pebble will necessarily be placed on~$v$ in the future
anyway with the goal of reaching the root.
This modification yields a strategy that uses
the same or lower number of pebbles and, at the same time,
fewer moves. 
Thus, repeating modifications of this sort leads to  a~strategy satisfying the 
property stated above.

We now construct the required re-pairing in the following natural way:
every time a node~$v$ receives a pebble,
we will erase the two signs associated with~$v$
(unless they have been erased previously).
We prove by induction on the number of moves that
the number of erased intervals cannot exceed the number of used pebbles
and, moreover,
every node that has carried a pebble is associated with
a pair of signs in the word~$\sigma$, between which all signs have already
been erased.

Indeed, this condition holds at the very beginning.
When a black pebble is put on a node~$v$ according to
rule~\ref{bw-move:place-black}, the children of that node
(if they exist in the tree) already carry pebbles
and the associated intervals have already been erased
and become one erased interval.
Erasing the pair of signs associated with~$v$ can only extend
this interval (and perhaps attach it to another interval).
Finally, when a black pebble is removed from a node
according to rule~\ref{bw-move:remove-black},
we know from the property above that the path from this node
to the root has a black pebble. This means that the intervals
associated with these two nodes are nested in one another,
so the removal of the smaller will not violate the inequality
in question.
These two are the only rules in the game that use black pebbles
only, so this completes the proof.
\end{proof}

This concludes the proof of Theorem~\ref{th:peb}.

\section{Proofs for Section~\ref{upbndZ}}
\label{app:upbndZ}

\subsection{Correctness of the definition}
\label{app:s:z:desc}

The only possible issue with the definition of the re-pairing~$p(q,n,k)$ is
the choice of subsequence~$Z'_t$. All other actions are specified 
either by recursively defined re-pairings, or by the direct rule `pair the
current minus with the leftmost non-erased plus'. 

We will prove  that the re-pairing~$p(q,n,k)$, $q\leq
(n+1)/2$, is well-defined by induction on~$n$.

The  base case, $n\leq2$, is trivial because there are no stages at all
in this case.

For the induction step, assume that for all $n<n_0$ the
re-pairings~$p(q,n,k)$, $q\leq (n+1)/2$, are
well-defined.

We prove that the re-pairing~$p(q,n,k)$ is well-defined 
by inner induction on stages. The base of the inner induction is
obvious, since $k_1=0$ by definition. 

For the inner induction step, assume that the re-pairing $p(q,n_0,k)$
is well-defined for all stages $t'<t$. 
At the beginning of the stage~$t$,  no sign in the factor $Z_t$ has been
erased. Indeed, this is obvious for the minuses, because
$p(q,n,k)$ picks them in a greedy way, i.e., in the same order
in which they occur in the word. But the property in question also extends to
the pluses, because they cannot be paired  with minuses to the left
of~$Z_t$.

We now show that we can choose,
from the sequence of non-erased signs, a subsequence $Z'_t$,
as specified in the definition of $p(q,n,k)$.
Recall that at each moment in time the non-erased signs form a
Dyck word. So, to pick up $k_t$ pluses to the left of the factor~$Z_t$
(a Dyck word itself)
it suffices to pick up $k_t$ minuses to the right of the factor.

By definition~\eqref{k_t-def}, $k_t\leq q-1$. On the other hand, the
word $Z(n)^{(k)}$ has the suffix $-^{k+n-q}$ that lies to the right of
any factor $Z(q)$. Taking into account the restriction on the
parameter values $q\leq (n+1)/2$ we get
\[
k_t\leq q-1\leq \frac{n-1}2\leq n-  \frac{n+1}2\leq 
n-q.
\]
This inequality guarantees that there are at least $k_t$  minuses 
to the right of the factor~$Z_t$, for each~$t$.

\begin{Remark}
Although the sequence $Z'_t$ contains $k_t$ minuses to the right of the factor~$Z_t$,
only the minuses to the left of~$Z_{t+1}$ are erased during stage~$t$,
by our definition of stages.
\end{Remark}

\subsection{Properties of the re-pairing}
\label{app:s:z:properties}

Let 
$w_t$ be
the factor between~$Z_t$ and $Z_{t+1}$, i.e., the (unique) word such that
$ Z_t\cdot w_t\cdot Z_{t+1} $ is a~factor of $Z(n)$.

\begin{prop}\label{app:Zm-Zm}
$w_0=+^{r_0}$, where $r_0 =k+n-q$; $w_{N}=-^{r_N}$,
$r_N = r_0$;  and
$w_i=-^{r_i}+^{r_i}$ for $1\leq i\leq N-1$. Here
  \begin{equation}\label{r_t-def}
    r_t = \max \{\, j \colon   2^j \text{\textup{\ divides\ }} t\,\},\ \text{where }1\leq t\leq N-1.
  \end{equation}
\end{prop}

\begin{proof}
Use induction on~$n$.
The base case, $n=q$, is trivial.
For the inductive step, consider the factorization
of $Z(n+1)^{(k)}$ into
\[
+^{k+1} Z(n)Z(n)-^{k+1}.
\]
The arithmetic condition on $r_0=r_N$ is obvious. There are exactly $2^{n-q}$
occurrences of $Z(q)$ in $Z(n)$. Thus, we need to check only the word between the last occurrence of~$Z(q)$
in the first factor~$Z(n)$ and the first occurrence of~$Z(q)$
in the second factor~$Z(n)$. This word is the concatenation of the suffix~$-^{n-q}$
from the first factor~$Z(n)$ and the prefix~$+^{n-q}$
from the second factor~$Z(n)$.
So it does have 
the right form, and $r_{2^{n-q}}=n-q$. This proves~\eqref{r_t-def}.
\end{proof}

Note that in the final part of a
stage~$t$, possibly empty, we pair minuses in the word~$w_t$ with the
leftmost non-erased pluses. 

In the sequel, we say that a stage is a \emph{narrow-frame} stage
if the final part
of it  is non-empty. Otherwise, we call a stage  \emph{wide-frame}.
This classification will be crucial in the following
analysis. 

From Claim~\ref{app:Zm-Zm} we conclude that a~stage~$t$ is narrow-frame iff
$r_t>k_t$. Let us compare the values of $r_t$ and $k_t$.
Note that
the sequence $k_t$ is periodic. For 
$t = 2^q \cdot s+a $, $0< a\leq 2^q$, we get
\begin{equation}\label{rk-rel}
\begin{aligned}
  &r_t \geq q> q-1= k_t&& \text{for}\ a=2^q;\\
  &r_t=k_t=0 &&\text{for}\ a=1;\\
  &r_t = j> j-1 = k_t&&\text{for}\ a=2^j,\ 1\leq j<q;\\
  &r_t<j = k_t &&\text{for}\ 2^j<a< 2^{j+1},\ 1\leq j<q.
\end{aligned}
\end{equation}

\begin{prop}
\label{app:monotone-property}
For $1\leq q \leq n/3$, $n\geq15$, and $k\leq n$,
the re-pairing  $p(q, n, k)$
erases all
pluses from the prefix~$+^k$ of the word $Z(n)^{(k)}$
before the first
minus from the suffix $-^k$.
\end{prop}

\begin{proof}
We prove that the pluses from the prefix~$+^k$
are exhausted before the factors~$Z(q)$ in the second half of the word
start getting erased.
It is clear that at this point in time no minus from the suffix
has been paired.

It follows from~\eqref{rk-rel} that there are exactly $q$ narrow-frame
stages during a period
$2^q \cdot s<t\leq 2^q \cdot (s+1)$. Thus, 
at least~$q$ leftmost non-erased pluses in total are erased during these stages.
The first~$k$ pluses, therefore,
will be exhausted after at most $\lceil k/q\rceil$ periods. There are
$2^{n-1-q}$ factors $Z(q)$ in the 
first half of the word $Z(n)$. So it is sufficient
to satisfy the inequality
\[
\left\lceil\frac{k}q\right\rceil\cdot 2^q < 2^{n-q-1}
\quad\Leftrightarrow\quad
\left\lceil\frac{k}q\right\rceil 2^{2q} <  2^{n-1}.
\]
Note that for $q=1$ we have $\lceil\frac{k}q\rceil = k$. 
So for $n\geq15$, $k\leq n$  we get
\[
k2^{2q} \leq n 2^{2n/3} < 
2^{n-1}.
\]
For $q\geq2$, we have $\lceil\frac{k}q\rceil \leq k/2+1$. 
So for $n\geq12$, $k\leq n$ and $1< q\leq n/3$ we get
\[
\left\lceil\frac{k}q\right\rceil 2^{2q} 
 \leq\Big(\frac n2+1\Big) 2^{2n/3} < 
2^{n-1}.\qedhere
\]
\end{proof}

\begin{cor}\label{app:left-correctness}
  For $15\leq q \leq (n+1)/2$ and $k\leq n$,
  all pluses in $Z'_t$ from the prefix ${+}^{k_t}$ (i.e., those preceding the factor $Z_t$) are
  erased during the stage~$t$ of the re-pairing~$p(q,n,k)$. 
\end{cor}

Now we are going to prove that the re-pairing $p(q,n,k)$ maintains
a specific shape of the sequence of non-erased pluses.
We call a stage~$t$ \emph{normal} if, at the beginning of the stage,
the following conditions hold:
{
\def\theenumi{(N\arabic{enumi})}
\def\labelenumi{\theenumi}
\begin{enumerate}
\item\label{app:runs-cond} Non-erased pluses to the left of the
factor~$Z_t$ are divided, for some $s$, into (possibly empty)
$s$ intervals (runs of `skipped' pluses) $P_1$, \dots,
$P_s$ of length~$m_1$, \ldots, $m_s$, respectively, and a group
of $k_t$ pluses that form a prefix of $Z'_t$.
(The latter do not have to form a contiguous interval, i.e.,
 a factor of $Z(n)$, and can be scattered inside $Z(n)$.)
\item\label{app:local-cond} The last group  lie in the factor $Z_{t-1}\cdot
  w_{t-1}$. 
\item\label{app:subfacs-cond} Each $P_i $ is a subfactor of $w_{j_i}$,
  $j_i<t$. (Or, equivalently, all factors $Z_i$, $i<t-1$, have been completely erased.)
\item \label{app:long-cond} $m_i\geq q$ for $i>1$.
\end{enumerate}

}

\begin{lemma}\label{app:normal-stage}
  For $15\leq q \leq (n+1)/2$ and $k\leq n$, each stage of the
  re-pairing $p(q,n,k)$ is normal.
\end{lemma}
\begin{proof}
  Use induction on stages. The base case, $t=1$, is trivial.

  For the induction step, assume that a stage $t-1\geq1$ is
  normal. The condition~\ref{app:local-cond} for it and
  Corollary~\ref{app:left-correctness} imply that at the beginning of
  stage $t$ all symbols in the factor $Z_{t-2}$, if any, have been
  erased. Thus, the condition~\ref{app:subfacs-cond} holds for stage~$t$.
  In addition, everywhere below we will also rely on the following
  consequence of Claim~\ref{app:monotone-property}:
  the group of $k_{t-1}$~pluses from the beginning of~$Z'_{t-1}$
  gets erased completely by the beginning of stage~$t$.

  The key intuition for the rest of the proof is based
  on the definition of the sequence $Z'_t$.
  In particular, note that all minuses that can be
  erased on stage~$t-1$ and earlier lie in the factors~$Z_{t-1}$
  and $w_{t-1}$.
  Since $w_{t-1} = {-}^{r_{t-1}} {+}^{r_{t-1}}$,
  it follows that none of the pluses in ${+}^{r_{t-1}}$ can be
  erased by the beginning of stage~$t$.
  These are the $r_{t-1}$~rightmost non-erased pluses to the left of $Z_t$,
  and recall that the definition of $Z'_t$ requires $k_t$~such pluses.
  So for the rest of normality conditions, 
  we need  relations between $r_{t-1}$ and $k_t$,
  where $t= 2^q \cdot s+a$, $1\leq a\leq 2^q$. They are as follows:
  \begin{equation}\label{r(t-1)k(t)-rel}
    \begin{aligned}
      &r_{t-1}=0, &&k_t = q-1&& \text{for}\ a=2^q;\\
      &r_{t-1}\geq q,  &&k_t = 0 &&\text{for}\ a=1;\\
      &r_{t-1} = 0, &&k_t = 0 &&\text{for}\ a=2;\\
      &r_{t-1} = j, && k_t= j&&\text{for}\ a=2^j+1,\ 1\leq j<q;\\
      &r_{t-1} <j, && k_t=j &&\text{for}\ 2^j+1<a< 2^{j+1}+1,\ 1\leq j<q.
    \end{aligned}
  \end{equation}

  If $t= 2^q \cdot s+1$, $s\geq1$, then the second line
  of~\eqref{r(t-1)k(t)-rel} shows 
  that stage $t-1$ is narrow-frame and, in fact, a group of $\geq
  q$ skipped pluses appears at the beginning of the stage~$t$. It
  implies the conditions~\ref{app:long-cond} and~\ref{app:runs-cond}. The
  condition~\ref{app:local-cond} is trivial in this case, $k_t=0$.

  If $t= 2^q \cdot s$, $s\geq1$, or  $t= 2^q \cdot s+a$, $2^j+1<a< 2^{j+1}+1$,
  $1\leq j<q$, then the first and the last lines
  of~\eqref{r(t-1)k(t)-rel} show that $k_t>r_{t-1}$. Note that
  $k_{t-1}=k_t$ in this case by equation~\eqref{rk-rel}. Applying
  Corollary~\ref{app:left-correctness} once more, we deduce that
  $k_{t-1}-r_{t-1}$ pluses remain non-erased in $Z_{t-1}$ and all of them lie
  in~$Z'_t$. It gives the conditions~\ref{app:local-cond}
  and~\ref{app:runs-cond}.  The condition~\ref{app:long-cond} also
  holds because the runs of `skipped' pluses do not change in this case.

  If $t= 2^q \cdot s+2$ or $ 2^q \cdot s+a$, $a=2^j+1$, $1\leq j<q$, then all signs
  in $Z_{t-1}$ are erased at the beginning of the stage~$t$. The runs
  of `skipped' pluses do not change. The $k_t=r_{t-1}$ pluses from
  $w_{t-1}$ are included into $Z'_{t}$. Thus, the
  conditions~\ref{app:local-cond}, \ref{app:runs-cond} and~\ref{app:long-cond} hold.
\end{proof}

Note that  the properties mentioned in Section~\ref{upbndZ} follow from
the conditions of
normality. Thus, Lemma~\ref{app:normal-stage} implies them.

\subsection{Proof of the recurrence for the width}

We now prove Claim~\ref{claim:work-rec} from Section~\ref{upbndZ}.

\begin{prop}
  For $15\leq q\leq(n+1)/2$, $1\leq q'\leq q/3$, $k\leq n$,
\begin{equation}\label{work-rec}
\Wd(p(q,n,k))\leq
\frac{2n}{q} + 2\max_{0\leq k'\leq q}\Wd(p(q',q,k')) + 3.
\end{equation}
\end{prop}
\begin{proof}
To prove~\eqref{work-rec} we estimate the width of $p(q,n,k)$ at any
moment in time. 
Recall that the erased signs constitute a well-formed word at any point in
time as well as the non-erased signs. The height of any position in the
entire word is the sum of the heights of this position  with respect
to the erased signs only and to the non-erased signs only.

In the re-paring $p(q,n,k) $, the skipped runs of pluses form a prefix
of the word that consists of all non-erased signs. From the above
argument we conclude that the height of this
prefix, i.e. the number of signs in the skipped runs of pluses, is at
most $n+k\leq 2n$ (the maximal height of a position in the framed word $Z(n)^{(k)}$
is $n+k$).

But the skipped runs of pluses $P_1, \ldots, P_s$, except possibly the first one, have size at
least~$q$ due to normality condition~\ref{app:long-cond}.
This means that there are at most $2n/q+1$ of them.
These skipped (non-erased) runs separate
at most $2n/q+2$~erased intervals.

At any point in time at most two factors~$Z(q)$, namely $Z_{t-1}$ and
$Z_t$, are being erased, by normality
condition~\ref{app:subfacs-cond}. Moreover, the only non-erased signs
in $Z_{t-1}$ at this moment are pluses.

The case analysis in the proof of
Lemma~\ref{app:normal-stage} shows that if both of these factors are partially
erased at the moment, then $k_t=k_{t-1}$. It implies that 
in  the re-pairings
of $Z'_{t-1}$ and $Z'_t$ these pluses are erased in a consistent way.
 In the former they  are erased from the left to the right during the
 last (narrow) stage 
of the re-pairing. In the latter they are located in the first run of
`skipped' pluses, and are thus erased from the left to the right.

This consistency of the re-pairings ensures that contribution of these
factors to the number of erased intervals is at most 
\[
\Wd(p(q',q,k_{t-1})) + \Wd(p(q',q,k_{t}))\leq 
2\max_{0\leq k'\leq q}\Wd(p(q',q,k')). 
\]
(We ignore in this upper bound that $q'$ is chosen optimally at each stage.)

Finally, one more erased interval separated from
the ones listed above may occur inside~$w_t$.
Putting everything together, we obtain~\eqref{work-rec}. 
\end{proof}

Note that the right-hand side of~\eqref{work-rec} does not depend on~$k$.
To get the recurrence of Claim~\ref{claim:work-rec},
observe that
\begin{align*}
 W_n &=\min_{15\leq q\leq n/3}\  \max_{0\leq k \leq n}\Wd(p(q,n,k)) \\
 &\le
 \min_{15\leq q\leq n/3}\Big(\frac{2n}{q}  + 
 2\min\limits_{1\leq q'\leq q/3}\max _{0\leq k'\leq q}
   \Wd(p(q',q,k')) + 3\Big).
\end{align*}
If $q < 45$, then the inner minimum in this expression
is, by definition, equal to $W_q$, and so
\begin{equation}\label{trueW-rec}
W_n \le
\min\limits_{15\leq q\leq n/3}\Big(
   \displaystyle\frac{2n}{q} + 2W_q + 3\Big)  .
\end{equation}
Otherwise, if $q \ge 45$,
notice that the minimum
over a set decreases if a set becomes larger, so
\begin{multline*}
\min_{15\leq q\leq n/3}\Big(\frac{2n}{q}  + 
 2\min\limits_{1\leq q'\leq q/3}\max _{0\leq k'\leq q}
   \Wd(p(q',q,k')) + 3\Big) \le \\
\min_{15\leq q\leq n/3}\Big(\frac{2n}{q}  + 
 2\min\limits_{15\leq q'\leq q/3}\max _{0\leq k'\leq q}
   \Wd(p(q',q,k') + 3\Big) \leq \\
\min\limits_{15\leq q\leq n/3}\Big(
   \displaystyle\frac{2n}{q} + 2W_q + 3\Big),
\end{multline*}
which gives us equation~\eqref{trueW-rec} again.

\subsection{Proof of the upper bound on the width}

It remains to obtain the upper bound from
the recurrence~\eqref{trueW-rec},
completing the proof of Theorem~\ref{Zn-upbnd}.
We begin by simplifying the recurrence a little.

We will bound the sequence of $W_n$s from above
by the sequence $x_n-3$, where $x_n$ is given by
\begin{equation}\label{app:exact-rec}
x_i = W_i+3, \ 3\leq i< 45; \quad
x_{n}=\min_{15\leq q\leq n/3}\Big( \frac{2n}{q} + 2x_q \Big), \ n\geq45 .
\end{equation}

\begin{prop}\label{app:W-x-rel}
   $W_n\leq x_n-3$.
\end{prop}

\begin{proof}
  For $n< 45 $, the equality  $W_n = x_n-3$ holds.
  For greater values of~$n$ the inequality is proved by induction,
  using the recurrences~\eqref{trueW-rec}
  and~\eqref{app:exact-rec}:
  \begin{equation*}
  W_{n}\leq\min_{15\leq q\leq n/3}\Big( \frac{2n}{q} + 2W_q +
  3\Big)\leq 
  \min_{15\leq q\leq n/3}\Big( \frac{2n}{q} + 2x_q\Big)-3 = x_n-3.\qedhere
  \end{equation*}
\end{proof}

To obtain asymptotic bounds on~$x_n$,
we need a different representation of these numbers.
Denote by~$\Z_n$ the set of sequences
\[
z= (z_0,\ z_1,\ \dots, z_s)
\]
of positive integers satisfying the conditions
(i)~$z_0=n$,
(ii)~$ 15\leq z_s\leq 44$, and
(iii)~$z_i/z_{i+1}\geq 3$ for $i = 0, \ldots, s-1$.
For the sequence~$z$ we denote by~$s(z)$
its length minus~$1$.
Finally, let the function $\xi$ be defined
on sequences by
\begin{equation}\label{app:xi-def}
\xi(z) = \sum_{i=1}^{s(z)} 2^i\cdot\frac{z_{i-1}}{z_i} + 2^{s(z)}x_{z_s}.
\end{equation}

\begin{prop}\label{app:global-min}
  $x_n = \min \limits_{z\in\Z_n}\xi(z)$ for all $n\geq45$.
\end{prop}

\begin{proof}
Use induction on~$n$.
The base case, $n=45$, is easily checked---%
the set~$\Z_{45}$ 
contains a unique sequence~$(45,15)$.

The induction step is given by the direct computation
(which is actually Bellman's principle of optimality)
as follows:
  \begin{multline*}
    \min_{z\in\Z_n} \Big(
     \sum_{i=1}^{s(z)} 2^i\cdot\frac{z_{i-1}}{z_i} + 2^{s(z)}x_{z_s}
     \Big) = \\
\min\Bigg( \min_{45\leq m\leq n/3} \Big(2\cdot\frac{n}m + 2
     \min_{z\in\Z_m}\xi(z)\Big),
\min_{15\leq m\leq 44} \Big(2\cdot\frac{n}m + 2x_m\Big)\Bigg)
 = \\
     \min_{15\leq m\leq n/3} \Big(2\cdot\frac{n}m + 2 x_m
     \Big) = x_n.
  \end{multline*}
The second equality is the inductive hypothesis,
and the last is the definition of the sequence~$x_n$.
\end{proof}

It is easier to estimate the minimum of~$\xi(z)$
over a bigger set~$\R_n$, defined analogously to~$\Z_n$
but in which the terms of the sequences are arbitrary
real numbers (instead of positive integers only).
More exactly, 
\[
z= (z_0,\ z_1,\ \dots, z_s)\in \R_n
\]
if 
(i)~$z_0=n$,
(ii)~$ 15\leq z_s\leq 44$ and $z_s$ is integer, and
(iii)~$z_i/z_{i+1}\geq 3$ for $i = 0, \ldots, s-1$.

Notice that the properties (i)--(iii) imply that
for all sequences~$z\in \R_n$ the length~$s(z)$
does not exceed $\log (n/15)/\log 3$.
Therefore, for every fixed~$n$ the set~$\R_n$
is a compact set, and in particular all
functions continuous in~$z_i$ attain a minimum
on this set.

Since $\Z_n\subseteq \R_n$,
we have
$\min\limits_{z\in \R_n}\xi(z)\leq \min\limits_{z\in \Z_n}\xi(z) $.
To obtain (a variant of) the converse inequality
for the minima of $\xi(z)$ over the sets $\Z_n$ and $\R_n$,
we will need two simple observations about real numbers.

\begin{prop}
\label{app:round-lwr}
Let $x$ and $y$ be real numbers
such that $y\geq 1$ and $x/y\geq3$.
Then $\lfloor x\rfloor/\lfloor y\rfloor\geq3$.
\end{prop}

\begin{proof}
For the sake of contradiction,
suppose $\lfloor x\rfloor< 3\lfloor y\rfloor$,
or equivalently $\lfloor x\rfloor\leq  3\lfloor y\rfloor-1$.
Since $\lfloor x\rfloor + \{x\}\geq 3(\lfloor y\rfloor + \{y\})$,
we obtain the inequality
  \[
  3\lfloor y\rfloor-1\geq 3\lfloor y\rfloor + 3\{y\} - \{x\},
  \]
and hence $1\leq \{x\} -3\{y\}<1 $.
This is a contradiction.
\end{proof}

\begin{prop}
\label{app:round-up}
If $x\geq 1$ and $ y\geq 2$,
then
  $$\frac{\lfloor x\rfloor}{\lfloor y\rfloor}< \frac32\cdot\frac{x}y\,.$$
\end{prop}

\begin{proof}
The required inequality is implied by
the chained inequality
\[
  \frac{\lfloor x\rfloor}{\lfloor y\rfloor}- \frac{x}y = 
  \frac{\lfloor x\rfloor (\lfloor y\rfloor + \{y\}) - 
  \lfloor y\rfloor (\lfloor x\rfloor + \{x\}) }
  {\lfloor y\rfloor (\lfloor y\rfloor + \{y\})} \leq \\
  \frac{\lfloor x\rfloor\cdot \{y\}}
       {\lfloor y\rfloor(\lfloor y\rfloor + \{y\})} \leq
  \frac{\lfloor x\rfloor}
       {\lfloor y\rfloor} 
       \cdot
  \frac{ \{y\}}
       {2 + \{y\}} <
  \frac{\lfloor x\rfloor}
       {\lfloor y\rfloor} 
       \cdot\frac13
\qedhere
\]
\end{proof}

\begin{prop}
\label{app:ZR-ineq}
$\min \limits_{z\in\Z_n}\xi(z)\leq \frac32\min \limits_{z\in\R_n}\xi(z)$.
\end{prop}

\begin{proof}
Let $z^*= (z_0^*,\dots,z_s^*)$ be the sequence
at which $\xi(z)$ attains its minimum on the set~$\R_n$.

Consider the sequence $z'=(z'_0,\dots,z'_s)$
with $z'_i = \lfloor z_i^*\rfloor$.
By Claim~\ref{app:round-lwr}, we have $z'\in\Z_n$.
Applying the inequality from Claim~\ref{app:round-up}
to each term in equation~\eqref{app:xi-def},
we obtain the inequality $\xi(z')\leq \frac32\xi(z^*)$.
The desired inequality for the minima follows.
\end{proof}

\begin{prop}\label{app:xn-bound}
  $x_n = 2^{\Theta(\sqrt{\log n})}$.
\end{prop}

\begin{proof}
We obtain the upper bound by
choosing an appropriate sequence from~$\R_n$
and applying Claim~\ref{app:ZR-ineq}.
Let $s$ be the maximum integer~$x$ satisfying
the inequality
  \[
  n\geq 6\cdot 2^{s^2/2},
  \]
i.e., $s= \Theta(\sqrt{\log n})$.
Consider the sequence~%
$z_{s-i} = 2\cdot2^{(i+1)^2/2}$,
$2\leq i\leq s-1$.
For $z_0=n$,
the sequence $(z_0, z_1,\dots , z_{s-2}, 15) $
belongs to the set~$\R_n$,
since $z_{s-2} =2\cdot 2^{9/2}>3\cdot15$
and for all $i \ge 2$ it is the case that
  \[
  \frac{z_{s-i}}{z_{s-i+1}} = \frac{2^{(i+1)^2/2}}{2^{i^2/2}} = 
  2^{i+1/2}>3.
  \]
But for this sequence we have
\[
  \sum_{i=1}^{s-1} 2^i\cdot\frac{z_{i-1}}{z_i} =
  \sum_{i=2}^{s  } 2^{s-i+1}\cdot\frac{z_{s-i}}{z_{s-i+1}} \geq \\
  \sum_{i=2}^{s-2} 2^{s-i+1}\cdot
  2^{i+1/2}
  =
  2 \sqrt2\cdot (s-3)\cdot 2^s.
\]
By Claim~\ref{app:ZR-ineq}, we have
  \[
  x_n \leq\frac32\cdot( 2\sqrt2\cdot (s-3)\cdot 2^s + 2^{s-1}\cdot \max_{15\leq
  i\leq 44} x_i ) = 2^{O(\sqrt{\log n})}.
  \]

We now prove the lower bound.
For $z\in \Z_n$ introduce the notation
  \[
  r(z) =\max_i \frac{z_{i}}{z_{i+1}}.
  \]
It follows from Claim~\ref{app:global-min} that
  \[
  x_n \geq \min_{z\in \Z_n}\max\big(r(z), 2^{s(z)+1}\big) + 2^{s(z)} \cdot \max_{15\leq
  i\leq44} x_i.
  \]
Since $44\cdot r(z)^{s(z) }\geq n$,
we have
  \[
  x_n \geq \min _s \max\big(2^{\log (n/44)/s}, 2^s\big).
  \]
The minimum is attained at $s = \sqrt{\log(n/44)}$.
We get the bound $x_n=\Omega( 2^{\sqrt{\log n}})$,
and the lower bound of the Claim follows.
\end{proof}

Claim~\ref{Wn-ub}
now follows from Claim~\ref{app:W-x-rel} and Claim~\ref{app:xn-bound}.

\section{Proofs of lower bounds on the width of re-pairings}
\label{app:s:lwr-bnd}

%
In this section we prove lower bounds in equation~\eqref{eq:lb}.
In particular,
Theorem~\ref{lwr-sqrtlog} will follow
from the following equation:
\begin{equation}\label{WdY}
\Wd(Y(\ell)) = \Omega(\ell),
\end{equation}
where the words $Y(\ell)$ are defined in Section~\ref{s:lwr-bnd}.
For the words~$Z(n)$ associated
with complete binary trees,
the same method will give us the following bound:
\begin{equation}\label{WdZ}
  \Wd(Z(n)) = \Omega\left(\frac{\log n}{\log\log n}\right).
\end{equation}
Recall that $Z(n) = X(1, \ldots, 1) = Y(1, n)$.

We will rely, throughout the section, on
the tree representation of re-pairings, discussed in Section~\ref{app:s:trees}.

\subsection{Proof overview}
\label{s:lb:strategy}

Fix a well-formed word~$W$.
We will restrict our attention to
\emph{symmetric} words,
i.e., those that stay unchanged if they are reversed
and all signs are flipped (inverted).

It is immediate, for instance, that all words~$Y(m, \ell)$ are symmetric.
We will rely on the fact that this symmetry
transforms a re-pairing (resp.\ a derivation)
into another re-pairing (resp.\ another derivation)
of the same width.

We already introduced the function $L(W,k)$, somewhat informally, in
Section~\ref{s:lwr-bnd}. Now we give a more formal definition,
based on the ones in Section~\ref{app:s:trees:col}. 

Let $\pi$ be a  derivation of a word $W$ based on a tree $(I, E)$, i.e.
$W =\sigma(\pi,I,E)$.
In Section~\ref{app:s:trees:fragments} we associate a~fragment $F_w$
to each factor $w$ of the word $\sigma(\pi,I,E)$. Note that $F_w$
depends on $\pi$, $I$, $E$ as well.

\begin{defi}
\label{def:L}
In the above notation,
\[
L(W,k) = \max_{(\pi,I,E): W = \sigma(\pi, I,E)}\ \max_{w: 
\Wd(F_w)\leq k} |w|.
\]
\end{defi}

Recall the informal meaning of the definition.
Given a word~$W$,
consider all possible tree derivations that generate~$W$.
Fragments
of width at most~$k$ in the trees (on which the derivations
are based)
are associated with factors of the word~$W$,
and $L(W, k)$ is the maximum length of such a factor.
Note that in this definition
the width of the (entire) trees is not restricted.

It is clear from the definition that the sequence of numbers~$L(W,k)$
is non-decreasing:
\[
L(W,1)\leq L(W,2)\leq \dots \leq L(W,k).
\]
Furthermore, if $\Wd(W)\leq k$, then $|W|\leq L(W,k+1)$
by Theorem~\ref{app:th:pm1}.

We will obtain \emph{upper} bounds on the numbers~$L(W,k)$ by induction.
For $L(W, 1)$ we will use a bound based on a simple fact
characterizing derivations of well-formed words.
We will say that a pair of symbols is \emph{separated} by a position~$i$
if the first symbol from the pair occurs to the left of~$i$
and the second symbol to the right of~$i$.

\begin{prop}
\label{prop:few-pairs}
Let $i$ be a position in a well-formed word~$\sigma$.
Then every re-pairing of~$\sigma$ has exactly $\Ht(i)$~pairs
that are separated by this position.
\end{prop}
\begin{proof}
First consider all pluses to the left of~$i$ that are \emph{not}
paired with minuses to the right of~$i$.
If the number of the pairs in the statement of the Claim exceeds $\Ht(i)$,
then to the left of~$i$ there are more minuses than the considered pluses.
But then these signs cannot be paired up among themselves, which is
a contradiction.

Second, the number of pluses to the left of~$i$ is greater
by $\Ht(i)$ than the number of minuses to the left of~$i$.
Therefore, some $\Ht(i)$ among these pluses have to be paired
with the minuses to the right of~$i$. This completes the proof.
\end{proof}

By $\Htmax(u)$ we will denote the maximum height of a position
in a word~$u$.

\begin{lemma}
\label{base-length}
For every symmetric word~$W$,
$L(W,1)\leq 4\Htmax(W)$.
\end{lemma}
\begin{proof}
If a factor $w$ is associated with a fragment of width~$1$,
then this fragment is a simple path.
Factorize the word~$w$ into $w_1\cdot w_2$,
where the symbols $w_1$ are generated by the left sides
of the edges and the symbols $w_2$ by the right sides.

Assume without loss of generality that $|w_1|\geq |w_2|$.
(If this is not the case, consider the symmetric re-pairing,
obtained by `reflecting' the indices of signs with respect
to the middle of~$W$ and by inverting all signs.)

Since the signs from the factor~$w_1$ are generated by the left
sides of the edges, they cannot be paired with one another
in the re-pairing (or, equivalently, in the derivation)
of the word~$W$.
So each minus from the factor~$w_1$ is paired with a plus
to the left of the start position of~$w_1$;
similarly, each plus from~$w_1$ is paired with a minus to the right
of the end position of~$w_1$.
In other words, every sign from the factor~$w_1$ is paired with a sign outside
$w_1$, and this pair of signs is separated by either the start
or end position of~$w_1$.

The start and end positions of $w_1$ have height not exceeding~$\Htmax(W)$,
the maximum height of a position in~$W$.
Therefore, these positions cannot separate more than $\Htmax(W)+\Htmax(W)$ signs,
and so $|w_1|\leq 2\Htmax(W)$, i.e., $|w| \leq 4\Htmax(W)$.
\end{proof}

We use Claim~\ref{prop:few-pairs}
not only in the proof of Lemma~\ref{base-length}, 
but also in the proof of the inductive upper bound on~$L(W,k)$.
Besides this Claim, we will need two
%
`distance' 
characteristics of the word~$W$ and their properties.

Before we introduce them,
let us fix some further notation.
Given a word over the alphabet~$\{+1,-1\}$,
define $\Delta(i,j)=\Ht(j)-\Ht(i) $ the increase in height
from position~$i$ to position~$j$; $i \le j$.
Similarly, for a factor~$w$,
we write $\Delta(w) = \Delta(i,j)$
where $i$~and $j$ are the start and end positions of~$w$,
respectively.%
\footnote{We assume that a \emph{position} in a word points to a place
  \emph{between} symbols. Formally, a position corresponds to a partitioning of the word~$w$
into a prefix and suffix: $w=p\cdot s$. Thus, 
an $n$-symbol word has $n+1$ positions, including the start position
(empty prefix) and the end position (empty suffix). }

We will now introduce the following definitions.
For a word $w\in\{{+},{-}\}$, we denote by $\mu(w)$ the maximum length
of a factor of $w$
that consists of minuses only: 
\[
\mu(w) = \max \{\,d \colon -^d\ \text{is a factor of}\ w\,\}.
\]

\begin{defi}\label{def:phi}
$\displaystyle\phi(W,x)=\min\{\,\mu(w) \colon \text{$w$ is a factor of $W$ and}\ |w|\geq x \,\}$.
\end{defi}
Informally, the value of $\phi(W,x)$
shows how many minuses in a row are guaranteed to occur in factors
of length~$x$ or more.
\begin{defi}\label{def:psi}
 $\psi(W,x) = \min \{\,|uv{-}^x| \colon uv{-}^x\ \text{is a factor of $W$ and}\ 
\Delta(u)\geq x\,\}-1 $.
\end{defi}
Informally, the value of $\psi(W,x)$
shows how far away to the left (from a~run of minuses)
one needs to jump in order to find
a~factor on which the height increase
matches the height drop on this run of minuses.
(More precisely, $\psi(W,x)$ determines the \emph{minimum} size of the jump,
 over all factors ${-}^x$ of $W$.)
Since the word~$W$ is well-formed,
such a factor has to exist somewhere on the left,
if ${-}^x$ is indeed a factor of~$W$.

In the proofs below we will also use another interpretation of $\psi(W,x)$,
as  the largest number~$\ell$ satisfying the following property:
every factor of~$W$ of the form~$u {-}^{x}$,
where $|u {-}^{x}|\leq\ell$,
contains no sub-factor~$u'$ such that $\Delta(u')\geq x$.

To obtain an upper bound on $L(W,k)$,
we will use a recurrence of the following form.

\begin{lemma}
\label{main-length}
There exist constants $A$, $C$, $N$ such that
for all symmetric words~$W$ of maximum height at least~$N$,
for all $k\geq 2$,
if $L(W,k-1)\geq C \cdot \Htmax(W) \cdot k$,
then
  $$
  \psi \left(W, \ \frac{1}{6k} \cdot \phi\Big(W, \ \frac {L(W,k)}{A \cdot \Htmax(W)
  \cdot k}\Big)\right)
  \leq 3 k \cdot L(W,k-1).
  $$
One can take $A = 17$, $C=17^2$, and $N=18$.
\end{lemma}

\begin{Remark}
For the (sequences of) words~$W$ that we consider,
the functions~$\phi$ and~$\psi$ are
`approximate inverses' of each other:
for big enough values of parameters,
the following bounds hold:
\begin{align*}
\phi(Z(n), x) &\geq \log x - O(1), &
\phi(Y(\ell), x) &= \Omega(\ell\cdot (x/9)^{1/\ell}), \\
\psi(Z(n), x) &\geq 2^x, &
\psi(Y(\ell), x) &= \Omega((x/2\ell)^\ell).
\end{align*}
Given these bounds, the left-hand side of the inequality in Lemma~\ref{main-length}
is approximately bounded from below
by the fraction $L(W, k) / (A \cdot \Htmax(W) \cdot k)$.
The quality of this approximation is determined by the lower bounds on~$\phi$ and~$\psi$ above
and by the magnitude of the coefficient~$1 / 6 k$.
\end{Remark}

Lemma~\ref{main-length}
will enable us to obtain an inductive upper bound on~$L(W,k)$
in terms of\linebreak $L(W,k-1)$, if the latter is not too small.
A proof of the bounds~\eqref{WdY} and~\eqref{WdZ} can then follow.

The proof of Lemma~\ref{main-length} is based on two auxiliary lemmas.
Consider a factor~$w$ of the word~$W$ associated with a fragment of
width at most~$k$ (in a tree derivation that generates~$W$).  We
show 
(see subsection~\ref{s:lb:raise-n-decline} below)
that the existence of a long enough sub-factor of the form~$-^{d}$ in
the factor~$w$ ensures the existence of another sub-factor~$u$ which
is \emph{derived during the same time period} as~$-^{d}$ and which has a big
enough height increase (of at least~$d / 6 k$).  On the other hand,
the distance (in the word~$W$) between sub-factors of the factor~$w$
that are derived during a single time period cannot be big
(subsection~\ref{s:time-space}).
We discuss the proof of
Lemma~\ref{main-length} in subsection~\ref{main-proof}.

In order to apply Lemma~\ref{main-length},
we will need to estimate the value of the functions~$\phi$
and~$\psi$.
These estimates are given in 
subsection~\ref{app:phi-psi-bounds}. 
Finally, in subsection~\ref{s:lb:final-calculations} below,
we discuss the proofs of the lower bounds~\eqref{WdY} and~\eqref{WdZ}
(see p.~\pageref{WdY})
and Theorem~\ref{lwr-sqrtlog}.

\subsection{Increases and drops in height within a single time period}
\label{s:lb:raise-n-decline}

Consider a derivation~$\pi$ based on a tree~$(I,E)$.
Restrict the tree to just the edges that exist
during the time period between some~$t_1$ and~$t_2$.
These edges form a forest, see Fig.~\ref{app:pic:time-zone}.
\begin{figure}[!h]
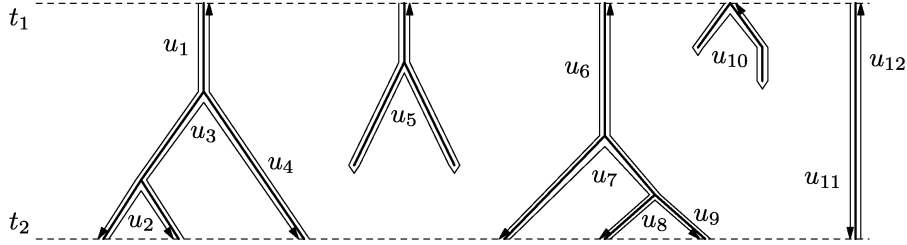

  \centering
   \mpfile{col}{2}

  \caption{A forest formed by the edges that exist
           during the time period between~$t_1$ and $t_2$}
  \label{app:pic:time-zone}
\end{figure}

The (weakly) connected components of this forest
occur in the traversal of the tree~$(I,E)$ consecutively,
as the following factors:
(in Fig.~\ref{app:pic:time-zone})
from the root of the component at time~$t_1$
to the leftmost edge at time~$t_2$ (using the left side of the edges),
then from the leftmost edge at time~$t_2$ to the second leftmost
edge at the same time (using the right side of the edges at first
and then switching to the left side), and so on.

These  factors of the traversal (a~sequence of length $2|E|$) define a system of  intervals
$[i_\al,j_\al]\subseteq\{1,\dots, 2|E|\}$. We assume that the
intervals are ordered: $j_\al<i_{\al+1} $. 
In the derivation~$\pi$ based on the tree~$(I, E)$, this system
generates  factors $u_\al=\pi([i_\al,j_\al])$ of the word~$\sigma(\pi, I,E)$.
Denote by $S(t_1,t_2)$ the sequence formed by
these factors%
,
that is,
\[
S(t_1,t_2) = (u_1,u_2,\dots, u_s).
\]
This is shown in Fig.~\ref{app:pic:time-zone}.

We will say that factors $u_i$ and $u_j$, $i<j$,
in the sequence $S(t_1,t_2)$
are \emph{independent}
if no plus from~$u_i$ is derived by~$\pi$
at the same time point as a minus from~$u_j$.

By the balance of paired signs in the derivation,
the following condition is satisfied.
(Recall that $\cdot$ denotes the concatenation of sequences.)

\begin{prop}
\label{app:branched}
Suppose $S(t_1,t_2) = S_0\cdot S_1\cdot S_2\cdot S_3$,
where all factors in~$S_0$ are independent from all factors in~$S_1\cdot S_2$;
and all factors in~$S_2$ consist of minuses only.
Then
  \[
  \sum_{u\in S_1}\Delta(u)\geq \sum_{v\in S_2}  |v|.
  \]
\end{prop}

\begin{proof}
Suppose $S_1$ and $S_2$ have the form
\[
S_1 = (u_1, u_2,\dots, u_\ell),\qquad
S_2 = (v_1, v_2,\dots, v_r).
\]
Denote by $m_i$ the number of minuses in the factor~$u_i$
that are paired (in the derivation) with symbols outside~$u_i$;
and by~$p_i$ the analogous number of pluses.
Then the height increase on the factor~$u_i$
is $\Delta(u_i) = p_i-m_i$,
because the signs paired inside~$u_i$ do not contribute
to the height change.

Since $S_0$ is independent from $S_1$ and $S_2$,
%
%
the concatenation of all factors in
the sequence $S_1\cdot S_2$ is a prefix of a well-formed word.
So any minus from~$u_j$ can only be paired with a symbol from
the words~$u_1$, \ldots, $u_{j}$.
The factors~$v_i$ consist of minuses only,
and these minuses can only be paired with pluses
from the factors $u_1$, $\dots$, $u_\ell$.
Since the height of all positions in a well-formed word
is nonnegative,
we obtain the following `balance' conditions:
\[
\begin{aligned}
  m_1& = 0,\\
  m_2&\leq p_1-m_1 = \Delta(u_1),\\
  m_3&\leq p_1-m_1+p_2-m_2 = \Delta(u_1)+\Delta(u_2),\\
  \dots\\
  m_{\ell}& \leq   \Delta(u_1)+\Delta(u_2)+\dots +\Delta(u_{\ell-1}),\\
  |v_1|& \leq  \Delta(u_1)+\Delta(u_2)+\dots +\Delta(u_{\ell})
  \,,\\
  |v_2|& \leq  \Delta(u_1)+\Delta(u_2)+\dots +\Delta(u_{\ell}) - |v_1|
  \,,\\
  \dots\\
  |v_r|& \leq  \Delta(u_1)+\Delta(u_2)+\dots +\Delta(u_\ell)
   - |v_1|-|v_2|-\dots - |v_{r-1}|.
\end{aligned}
\]
The last inequality is the desired one.
\end{proof}

This balance condition will give us the first ingredient in the
proof of the inductive upper bound on~$L(W,k)$.

\begin{lemma}
\label{raise2left}
Suppose a factor~$-^d$ of the word~$\sigma=\sigma(\pi, I,E)$
is derived
during time period from~$t_1$ to~$t_2$.

Also suppose that $S(t_1,t_2) = S_0\cdot S_{1,2}\cdot S_3$,
where
all factors in~$S_0$ are independent from all factors in~$S_{1,2}$;
and all symbols from the factor~$-^d$
are contained in the factors from~$S_{1,2}$.

Finally, suppose the width of the set of edges
that
derive
(all symbols in) the factors in~$S_{1,2}$
does not exceed~$k$.

Then there exists a factor of the word~$\sigma$
which is derived during the same time period as~$-^d$,
occurs to the left of $-^d$ in the word~$\sigma$,
is a sub-factor of one of the factors in~$S_{1,2}$,
and on which the height increase is at least~$d/6k$.
\end{lemma}
\begin{figure}[!h]
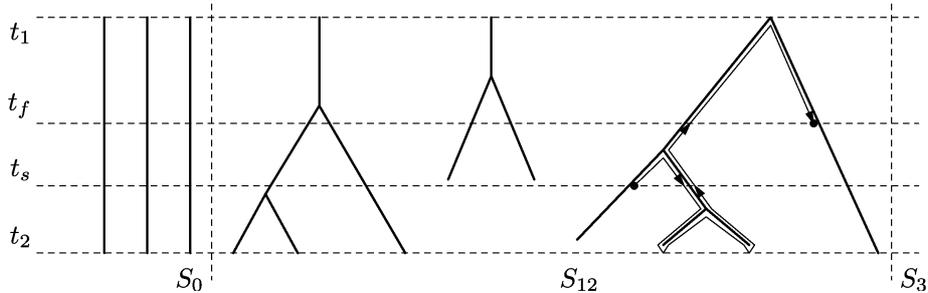

  \centering
   \mpfile{col}{3}

  \caption{Time period during which $-^d$ is derived}
  \label{app:pic:subword}
\end{figure}

\begin{proof}
Denote by $t_1$ the earliest time point at which
minuses from the factor~$-^d$ are derived,
and by $t_2$ the last such time point.
Also denote by $t_s$ the time the first minus from~$-^d$
is derived, and by~$t_f$ the time the last minus
from~$-^d$ is derived.
(See Fig.~\ref{app:pic:subword}.)

The time period from~$t_1$ to~$t_2$
is divided by the time points~$t_s$ and~$t_f$ into
at most three intervals.
Each of these intervals, denote it $[t',t'']$,
%
%
induces a factorization $S_0'\cdot S_1'\cdot S_2'\cdot S_3'$;
here the factors from~$S_2'$ are sub-factors of the factor~$-^d$.
(See Fig.~\ref{app:pic:subword}.)
Since the factors in~$S_0' $ are independent from the factors in~$S_1'\cdot S_2'$,
for each of these three time intervals Claim~\ref{app:branched} is applicable.
During at least one of these intervals,
at least $d/3$~minuses from~$-^d$ are derived;
apply Claim~\ref{app:branched} to this particular interval.
We get
  \[
  \Delta(u_1) + \Delta(u_2) +\dots+ \Delta(u_\ell)\geq \frac{d}{3}\,,
  \]
where all the factors~$u_i$ are derived during the time period
$[t',t'']\subseteq[t_1,t_2]$ and are, by construction,
sub-factors of the factors in~$S_{1,2}$
(in the original factorization of $S(t_1,t_2$);
see Fig.~\ref{app:pic:subword}.

The number of such factors is at most~$2k$,
because the number of their endpoints is at most
the number of sides of edges that exist at times~$t_1$ and~$t_2$,
i.e., $2k+2k=4k$; and each factor has exactly two endpoints.

Therefore, the height increase on at least one of these
factors~$u_i$ is at least~$d/6k$.
\end{proof}

\subsection{Distances within a single time period}
\label{s:time-space}

Consider, as previously, a derivation~$\pi$ based on a tree~$(I,E)$,
and a factor~$w$ of the word~$W=\sigma(\pi,I,E)$
associated with a fragment of width~$k$.
Denote this fragment~$F$.
Assume with no loss of generality
that $w$ is the maximal factor associated with~$F$.

Recall that edges of~$F$ form a (weakly) connected subgraph,
that is, a tree.
The root of this tree is the node closest to the root
of the entire tree~$(I,E)$.
Pick in~$F$ any path of maximum length from the root
and consider the point on this path furthest away from the root.
This point corresponds to a position
in the factor~$w$ that
splits it into $w = w_1\cdot w_2$.
By symmetry,
we can assume without loss of generality that
$|w_1|\geq |w|/2$.

Our strategy is to prove an upper bound on
the length of the factor~$w$
(and, as a result, on the value of~$L(W, k)$; see subsection~\ref{main-proof}).
In the present subsection we obtain an intermediate result:
we identify a long enough factor~$w_1$ inside~$w$
and show that the distance between the symbols of this factor~$w_1$
derived at a single time point is bounded from above
(this bound is in terms of $L(W,k-1)$).

We will first identify some structure inside the tree of~$F$.
Specifically, we will need a `trunk', spawning branches of smaller
width.

Pick any path of maximum length from the root as this \emph{trunk}.
Removing the edges of the trunk would make our binary tree into
a forest. Each (weakly) connected component of this forest
is a rooted tree, and its root sits on the trunk.
We will refer to these components as \emph{branches} of the tree.

Branches are categorized as left or right, depending on their
position relative to the trunk, see Fig.~\ref{app:pic:trunk}.

\begin{figure}[!h]
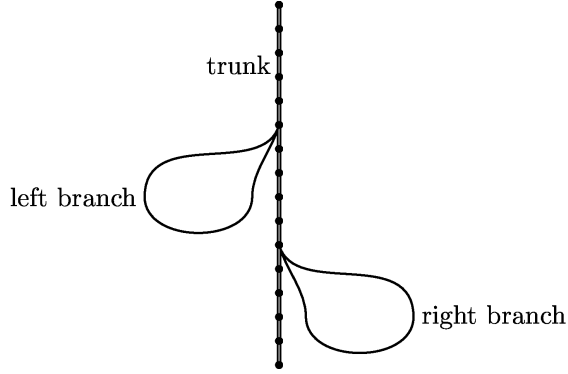

  \centering
   \mpfile{col}{4}

  \caption{Trunk and branches}
  \label{app:pic:trunk}
\end{figure}

Left branches are ordered by the point of their attachment
to the trunk. 
This matches the order of their occurrence
in the tree traversal: whenever one branch starts before
another, the tree traversal encounters the former before
the latter.

\begin{prop}
The width of each branch
is strictly less than
the width of the tree.
\end{prop}

\begin{proof}
This is obvious from Fig.~\ref{app:pic:trunk}:
at each point in time
when an edge of a branch exists,
there also exists an edge of the trunk (it is a path of maximum length from
the root).
\end{proof}

The point of the trunk that furthest away from the root (the lowest point)
corresponds to a position
%
%
in the factor~$w$;
this position splits this factor into $w = w_1\cdot w_2$.
Here the symbols of~$w_1$ are derived by the (sides of) edges
that occur before this lowest point of the trunk
in the traversal of the tree.

We will from now on focus on the word~$w_1$,
assuming without loss of generality that $|w_1|\geq |w_2|$.
Recall that we only consider symmetric words~$W$,
so if this inequality fails (that is, if $|w_2|>|w_1|$),
we can instead take a symmetric derivation of~$W$.
This derivation is obtained by reflecting the indices of symbols
in each erased pair with respect to the middle of the word~$W$
(and flipping the signs).
In this derivation, the factor~$w$ now corresponds to
another factor, $\bar w^R$, obtained by flipping all the signs
and reversing the word. This factor can be written as
$\bar w^R = \bar w_2^R\cdot \bar w_1^R$,
and in this factorization the length of the prefix
is greater than the length of the suffix.
Applying the argument below to the sub-factor~$\bar w_2^R$
(instead of~$w_1$), we will obtain the desired bound
for~$w_2$.

So, we have $|w_1| \geq |w_2|$.
The left branches factorize the word~$w_1$ as follows:
\[
w_1 =b_0\cdot b_1\cdot b_2\cdot \ldots\cdot b_N.
\]
The factor~$b_0$ is derived by
the edges of the trunk, which correspond to the prefix of~$w_1$.
Each factor~$b_i$, $i>0$,
is derived by the edges of a branch
and the edges of the trunk preceding the next branch.
Factorize $b_i= b'_i\cdot b''_i$,
according to which symbols are derived by the edges of the branch
resp.\ the edges of the trunk;
see Fig.~\ref{app:pic:branches}.
For~$b_0$, the prefix~$b'_0$ in this factorization
is the empty word.

\begin{prop}
$|b'_i| \le L(W, k-1)$ and
$|b''_i| \le 3 \Htmax(W)$
for all~$i$.
\end{prop}

\begin{proof}
First,
the factors~$b'_i$ are derived by fragments
of width strictly smaller than~$k$,
so the length of each of them is at most~$L(W, k-1)$.

Second,
observe that a plus from a factor~$b''_i$ can only be paired with a minus
outside~$w_1$, because this minus occurs to the right of it
(and is derived at the same time point),
and therefore to the right of the lowest point of the trunk---%
which is where the factor~$w_1$ ends.
Applying Claim~\ref{prop:few-pairs} to the lowest point
of the trunk, we obtain that there are at most~$\Htmax(W)$ such pluses
On the other hand, the absolute value of the change in height on the
factor~$b''_i$ cannot be greater than~$\Htmax(W)$.
It follows that the length of~$b''_i$ is at most~$3\Htmax(W)$.
\end{proof}

\begin{figure}[!h]
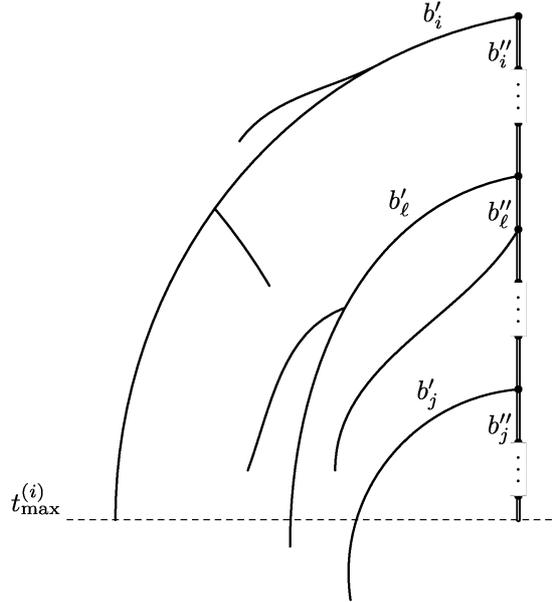

  \centering
   \mpfile{col}{5}

  \caption{Position of branches in time}
  \label{app:pic:branches}
\end{figure}

\begin{lemma}
\label{app:branches-dist0}
Suppose the time periods during which the symbols of the factors~$b_i$
and $b_j$, $i<j$, are derived, overlap.
Then the distance
between the start position of~$b_i$
and end position of~$b_j$
(in the word~$W$)
is at most $ (k-1)(2L(W,k-1)+3\Htmax(W))$.
\end{lemma}

\begin{proof}
Denote by $t^{(i)}_{\max}$ the last point in time when
there exists some edge of the branch that derives the word~$b'_i$;
see Fig.~\ref{app:pic:branches}.
The width of the entire fragment $F$ is at most $k$. Therefore, at
the point in time $t^{(i)}_{\max}$,
there exist edges of at most $k-1$~branches
to the right of the $i$th branch (and including it)
and to the left of the trunk.
(Indeed, edges of all these branches and 
an edge from the trunk exist at  time point $t^{(i)}_{\max}$.)
If an edge of a branch exists at time point $t^{(i)}_{\max}$, then the
branch will be referred to as \emph{long}.

Suppose the long branch following the $i$th branch has
index~$\ell$. The factor
  \[
   u= b_{i+1}\cdot\ldots\cdot b_{\ell-1}
  \]
is derived by the edges that belong to a fragment of width
strictly smaller than~$k$, because
edges of the branch~$b_i$ exist to the left of this fragment.
So the length of this (small) fragment cannot exceed~$L(W,k-1)$.

Therefore, every long branch, together with the factor
up to the next long branch, have total length of at most
  \[
  (L(W,k-1) +3\Htmax(W))+L(W,k-1).
  \]
Since there are at most $k-1$ such branches,
the statement of the lemma follows.
\end{proof}

An immediate consequence of this lemma is the following statement.

\begin{lemma}
\label{branches-dist}
If two factors in the word~$w_1$ are derived during overlapping time
intervals,
then the shortest factor containing both
has length at most~$(k-1)(2L(W,k-1)+3\Htmax(W))$.
\end{lemma}

\subsection{Proof of Lemma~\ref{main-length}}
\label{main-proof}

Consider a derivation~$\pi$ based on a tree~$(I,E)$
and a factor~$w$ of the word~$W=\sigma(\pi,I,E)$
associated with a fragment of width~$k$.
Denote this fragment by~$F$;
as in the previous subsection assume with no loss of generality
that $w$ is the maximal factor associated with this fragment~$F$.
Suppose the length of~$w$ is~$L(W,k)$.

In the previous subsection,
we have identified a long enough factor~$w_1$
inside~$w$; more specifically, we have $|w_1|\geq |w|/2$.

The overall idea of the proof
of the upper bound in Lemma~\ref{main-length}
is to find,
in our factor of the word~$W$,
a long enough sub-factor that consists of minuses only
and apply Lemma~\ref{raise2left} to this sub-factor.
This lemma will yield another sub-factor~$u$,
one with big increase in height.
On the other hand,
Lemma~\ref{branches-dist} shows that
the distance
(in the word~$W$)
between the sub-factor~$u$ and the sub-factor
that consists of minuses only
cannot be very big.
This will give us the desired upper bound.

An issue with this plan is that
a direct application of Lemma~\ref{raise2left}
will give us no information on whether
the word~$u$ is indeed a sub-factor of~$w_1$.
Indeed, the word~$u$ may well sit to the left
of~$w_1$ inside the entire word~$W$.

In order to avoid this issue, we will apply
Claim~\ref{prop:few-pairs} again.

Denote $H = \Htmax(W) \ge 1$.
Split the word~$w_1$ into $8kH+1$ factors
with approximately equal lengths:
\[
w_1 = w'_1\cdot w'_2\cdot \ldots \cdot w'_{8kH+1}.
\]
More precisely, the length of each~$w'_i$
is chosen to be at least
\[
L = \left\lfloor 
\frac{|w_1|}{8kH+1}
\right\rfloor
\geq
\frac{L(W,k)}{16kH+2} -1.
\]
Note that if $H$ is sufficiently big, the inequality
\[
L\geq \frac{L(W,k)}{17kH}\,
\]
holds---indeed,
one can check by direct calculation that
it is sufficient that the following two inequalities hold:
$kH\geq 36$ and $L(W,k)\geq L(W,k-1)\geq 17^2 kH$.
So we will use constants $N=18$ (restricting ourselves to $k\geq2$),
$A=17$, and $C = 17^2$
in the statement of Lemma~\ref{main-length}.

By  Definition~\ref{def:phi} of the function~$\phi$ (see p.~\pageref{def:phi}),
each of the factors~$w'_i$ contains a sub-factor~$-^{d}$
with
\[
d \geq \phi\left(W, \ \frac{L(W,k)}{17kH}\right).
\]

The start and end positions of the word~$w_1$
separate at most~$2H$ pairs of signs in the re-pairing~$p_\pi$
(by Claim~\ref{prop:few-pairs}).
At each time point when one of these `separated' signs
is derived (there are at most~$4H$ of these time points),
there exist at most $k$~edges in the left branches and trunk
of the tree (in the fragment~$F$) that may be deriving symbols
from the words~$w'_i$.
Each edge may be deriving symbols by its left or right side.

We now conclude that there is a word~$w'' = w'_j$ with the following property:
the signs of this word are not paired with symbols outside~$w_1$.
This word will be said to be \emph{isolated}.
Note, however, that the signs of~$w'_j$ may well be paired with
signs outside~$w'_j$.

So, inside the isolated word~$w''$ there is a factor~$-^d$
but no signs paired with signs outside~$w_1$.
This enables us to apply Lemma~\ref{raise2left}.
Indeed,
let the factor~$-^d$ be derived during time period~$[t_1,t_2]$;
we have $[t_1, t_2] \subseteq [\tau_1, \tau_2]$,
where $[\tau_1, \tau_2]$ is the time interval when
edges of the fragment~$F$ exist.
%
%
Now, since the bigger word~$w$ was chosen to be maximal
(among all words associated with the fragment~$F$),
each factor in the sequence $S(t_1, t_2)$ either
is completely inside~$w_1$ or sits completely before
or after it.
Then, in the statement of the lemma,
let $S_0$ consist of all the factors that sit to the
left of~$w_1$,
and let $S_{1,2}$ consist of all the factors that
are sub-factors of~$w_1$.
It follows from Lemma~\ref{raise2left} that
the word~$w_1$ contains a sub-factor~$u$
which is derived within the same time interval~$[t_1,t_2]$
and on which the height increase is $\Delta(u)\geq d/6k$.

By Lemma~\ref{branches-dist},
the shortest factor of the word~$W$
containing both of these sub-factors~$u$ and~$-^{d}$
has length at most $(k-1)(2L(W,k-1)+3H)$.
But by  Definition~\ref{def:psi} of the function~$\psi$
(see p.~\pageref{def:psi}),
this length must be strictly greater than
\[
\psi\left(W, \ \frac d{6k}\right).
\]
The function~$\psi$ is easily seen to be non-decreasing;
hence, we obtain the inequality
\[
\psi \left(W, \ \frac{1}{6k}\cdot\phi\Big(\frac {L(W,k)}{17 H
  k}\Big)\right)\leq (k-1)(2L(W,k-1)+3 H), 
\]
the right-hand side of which does not exceed $ 3kL(W,k-1)$
under the conditions of the lemma
(it suffices that the inequality $L(W,k-1)\geq 3H$ hold).

This completes the proof.

\subsection{Combinatorial properties of the words $Z(n)$ and $Y(m,\ell)$}
\label{app:phi-psi-bounds}

The inductive definition of the words
$$Z(n) =
X(\underbrace{1,1,\dots,1}_{\text{\tiny $n$}})\quad\text{and}
\quad Y(m,\ell) = X(a_0,\dots, a_{m\ell-1}),
$$
where $a_i = 2^{\lfloor i/\ell\rfloor}$,
can be interpreted as follows.
The signs of these words are grouped at the nodes
of complete binary trees of height~$n-1$ and~$m\ell-1$, respectively.
At the node of tree at distance~$r \ge 0$ from the leaves,
one pair of opposite signs is grouped in the case of~$Z(n)$,
and $a_r$~pairs in the case of~$Y(m,\ell)$
(all the pluses in this group form a single run,
 and so do all the minuses).

Therefore, the length of the word~$Z(n)$ is~$2\cdot(2^n-1)$,
double the number of nodes in a complete binary tree of height~$n - 1$.

We now estimate the length of the word~$Y(m,\ell)$.

\begin{prop}
\label{app:lenY}
$2^{m\ell}\leq |Y(m,\ell)|\leq 3\cdot 2^{m\ell}$ if $\ell>1$.
\end{prop}

\begin{proof}
Note that $Y(1,\ell) = Z(\ell)$, so
  \[
  |Y(1,\ell) | = 2^{\ell+1}-2> 2^{1\cdot\ell}
  \]
if $\ell>1$.

The set of nodes of the complete binary tree of height~$m\ell-1$
is split into~$m\ell$ levels.
The root and the $\ell-1$ levels closest to the root,
taken together,
contribute
  $$2\cdot 2^{m-1}\cdot(2^{\ell}-1)$$
signs to the length of $Y(m,\ell)$.
All other nodes are divided into $2^{\ell}$ subtrees
of height~$(m-1)\ell$ each, one subtree corresponding to
a factor~$Y(m-1,\ell)$.
This gives us a recurrence for the length of $Y(m,\ell)$:
  \[
  |Y(m,\ell)| = 2^{\ell}\cdot|Y(m-1,\ell)| +  2^{m}\cdot(2^{\ell}-1),
  \quad
  m \geq 2.
  \]
It is now easy to obtain the exact expression
for $|Y(m,\ell)|$:
  \begin{multline*}
    |Y(m,\ell)| = (2^\ell-1)(2^m+2^{m-1}\cdot 2^\ell + \dots +
     2^1\cdot 2^{(m-1)\ell}) = \\
     = 
     (2^\ell-1)\cdot 2^m\cdot (1+2^{\ell-1}+\dots + 2^{(m-1)(\ell-1)})
     = \\
     =
     (2^\ell - 1) \cdot 2^m \cdot \frac{2^{(\ell - 1) \cdot m} - 1}{2^{\ell - 1} - 1}
     = 
     \frac{2^\ell-1}{2^{\ell-1}-1}\cdot(2^{m\ell}-2^m).
  \end{multline*}
The inequalities of the Claim now follow from this expression.
Indeed, since $\ell>1$,
we have $(2^\ell-1)/(2^{\ell-1}-1)\leq 3$,
and the upper bound follows;
the lower bound inequality
  \[
 \frac{2^\ell-1}{2^{\ell-1}-1}\cdot(2^{m\ell}-2^m)> 2^{m\ell}
  \]
also holds, since the conditions $\ell>1$ and $m \geq 1$ imply
that
  \[
  2^{\ell-1}\cdot 2^{m\ell} > (2^\ell-1)\cdot 2^m.\qedhere
  \]
\end{proof}

To apply Lemma~\ref{main-length}, we need bounds on the values
of the functions~$\phi$ and~$\psi$ for the words~$Z(n)$ and~$Y(m,\ell)$.
The arguments in both cases will be similar to each other.
We will rely on several simple facts, which we formulate
for general words of the form $X(a_0,\dots,a_{k-1})$
and will then use for $Z(n)$ and $Y(m,\ell)$.

We first generalize a property of~$Z(n)$ formulated previously
in section~\ref{app:upbndZ}, see Claim~\ref{app:Zm-Zm}.

\begin{fact}
\label{app:subswordsX}
For all $0\leq s\leq k-1$,
the word $X(a_0,\dots,a_{k-1})$ factorizes as
  \[
  w_0 X_1 w_1 X_2 w_2 X_3 \dots w_{2^{s}-1} X_{2^s} w_{2^s},
  \]
where the factors $X_i$ are equal to $X(a_0,\dots,a_{k-s-1})$,
the factor~$w_0=+^{r_0}$, $w_i=-^{r_i}+^{r_i}$ for $0< i<2^s$,
$w_{2^s}=-^{r_{2^s}} = -^{r_0}$, and
no number~$r_i$ exceeds the maximum height of positions
in the word~$X(a_0,\dots,a_{k-1})$.
\end{fact}

\begin{proof}
Applying the inductive definition $s$~times,
we  find $2^s$~factors of the form~$X(a_0,\dots,a_{k-s-1})$
in~$X(a_0,\dots,a_{k-1})$.
Put differently,
these occurrences of $X(a_0,\dots,a_{k-s-1})$ in~$X(a_0,\dots,a_{k-1})$
correspond to the nodes of a complete binary tree
at distance~$s$ from the root.

One can see from here that these factors are separated from
one another by words of the form~$-^{r_i}+^{r_i}$, where
each~$r_i$ is at most the maximum height of positions
in the word\linebreak $X(a_0,\dots,a_{k-1})$;
the same holds for the prefix and suffix~$w_0$ and~$w_{2^s}$.
\end{proof}

\begin{fact}
\label{app:heightX}
The height of all positions in the word~$X(a_0,\dots, a_{k-1})$
is at most
\[
a_0+a_1+\dots+a_{k-1},
\]
and this bound is tight.
The increase in height on factors that are not prefixes
of~$X(a_0,\dots, a_{k-1})$ is strictly smaller than this bound.
\end{fact}

\begin{proof}
Use induction on~$k$.
The base case is~$k=1$; it is obvious that
that the maximum height of positions in
the word $X(a) = +^a-^a$ is equal to~$a$.
It is as clear that every factor of~$X(a)$,
unless it is a prefix, has fewer than $a$~pluses,
and so the increase in height on this factor
is strictly smaller than~$a$.

For the induction step, assume that
for a word~$X(a_0,\dots, a_{k-1})$
the inductive statement holds.
In the word
  \begin{equation*}
  X(a_0,\dots, a_{k-1},a_k) = 
  +^{a_k}X(a_0,\dots, a_{k-1})X(a_0,\dots, a_{k-1})-^{a_k},
  \end{equation*}
the maximum height of positions does not
exceed the height of positions
inside either of the factors~$X(a_0,\dots, a_{k-1})$
increased by~$a_k$.
This bound is tight for those positions of the factors~$X(a_0,\dots, a_{k-1})$
that have maximum height inside these factors.

Suppose a non-empty factor~$w$ is not a prefix of~$X(a_0,\dots, a_{k-1},a_k)$.
Consider the following three cases.

1. The factor~$w$ starts in the prefix~$+^{a_k}$, but not from
   the start position of 
the word.
   Then, using the same argument as above,
   the height increase on the factor~$w$ is strictly smaller
   than the maximum height of~$X(a_0,\dots, a_{k-1})$
   increased by~$a_k$.

2. The factor starts inside either $X(a_0,\dots, a_{k-1})$.
   Then the height increase on this factor cannot exceed
   the maximum height of positions in the word
   $X(a_0,\dots, a_{k-1})$.

3. The factor starts inside the suffix~$-^{a_k}$.
   Then the height increase on this factor is negative.
\end{proof}

We conclude from Fact~\ref{app:heightX} that the maximum height
of positions in the words~$Z(n)$ and $Y(m,\ell)$ satisfies
the equalities
\begin{align}
\label{app:maxHtZ}
\Htmax(Z(n)) &= n,\\
\label{app:maxHtY}
\Htmax(Y(m,\ell)) &= \ell\cdot1+\ell\cdot 2^1+\dots+\ell\cdot 2^{m-1}
= \ell\cdot(2^m-1).
\end{align}

Since the functions~$\phi$ and~$\psi$ are defined using
factors that consist of minuses only,
we will need the following fact on the relative
position of such factors inside $X(a_0,\dots, a_{k-1})$.

\begin{fact}
\label{app:minusesX}
The suffix of $X(a_0,\dots, a_{k-1})$ of length~$a_0+\dots+a_{k-1}$
consists of minuses only.

Suppose a factor of $X(a_0,\dots, a_{k-1})$ has the form~$-^x$,
and a number~$s$ is such that
  \[
  a_0+a_1+\dots+a_s\leq x.
  \]
Then this factor \textup{(}i.e., this occurrence of $-^x$
in $X(a_0,\dots, a_{k-1})$\textup{)}
can be extended leftwards to an occurrence
of a word~$X(a_0,\dots, a_{s})-^r$
\textup{(}for some $r \ge 0$\textup{)}.
\end{fact}

\begin{proof}
The first assertion is obvious from the inductive definition
of $X(a_0,\dots, a_{k-1})$.

For the second assertion,
we will show by induction on~$k$ that,
for all $k>1$,
(length-)maximal occurrences of factors that consist of minuses only
(in the word~$X(a_0,\dots, a_{k-1})$)
are suffixes of factors of the form~$X(a_0,\dots, a_{s})$.

The base case is $k=2$, i.e., words of the form
  \[
  X(a,b) =
  +^b+^a-^a+^a-^a-^b = +^{a+b}-^a+^a-^{a+b}.
  \]
Here the assertion is checked easily for the two maximal
factors that consist of minuses only.

For the induction step, note that, in the word
  \begin{equation*}
  X(a_0,\dots, a_{k-1}, a_k)  =
  +^{a_k} X(a_0,\dots, a_{k-1})  X(a_0,\dots, a_{k-1}) -^{a_k},
  \end{equation*}
every (length-)maximal factor that consists of minuses only
either is a suffix of the entire word~$X(a_0,\dots, a_{k-1}, a_k) $
(and then $s \in \{0, 1, \ldots, k\}$),
or ends inside either $X(a_0,\dots, a_{k-1}) $
and is bounded by the plus from the right---
in which case we use the inductive assumption.
\end{proof}

We now establish the bounds on the functions~$\phi$ and~$\psi$
for the words~$Z(n)$ and~$Y(m,\ell)$.
They are proved under certain assumptions on the parameters
and arguments of these functions;
in all our cases these assumptions will hold.

\begin{prop}
\label{app:phiZ}
If $x\geq6 n$, then $\phi(Z(n), x)\geq \log (x/12)$.
\end{prop}

\begin{proof}
We will find a condition on the length of a factor~$u$
of the word~$Z(n)$ which guarantees that the factor~$u$
contains a sub-factor~$Z(s)$ for big enough~$s$.
Fact~\ref{app:minusesX} will then imply that
this sub-factor contains~$-^s$.

By Fact~\ref{app:subswordsX},
every factor of length at least $2|Z(s)| + 2n$
contains a sub-factor of the form~$Z(s)$.
Indeed, the distance between two consecutive occurrences
of~$Z(s)$ in~$Z(n)$ cannot exceed~$2n$.
So for every factor of length~$2|Z(s)| + 2n$
one the following two conditions holds.
Either this factor begins between the occurrences of~$Z(s)$---%
and then its prefix of length~$2n+|Z(s)|$ contains
a sub-factor~$Z(s)$ already;
%
%
or the start position of this factor is inside
an occurrence of~$Z(s)$---%
in which case the prefix of this factor of length
at most~$2|Z(s)| + 2n-1$ contains a sub-factor~$Z(s)$.

Therefore, $\phi(Z(n), x)\geq s$ if $x \geq2|Z(s)| + 2n$.
Assuming $x\geq6 n$, we have
  \[
  x \geq 2 \cdot \frac{x}{3} + 2 n,
  \]
and under this assumption it would suffice
that the inequality $|Z(s)|\leq x/3$ hold---%
and this is equivalent to $2^{s+1}-2\leq x/3$,
i.e., $s\leq \log (x/3+2)-1$.
So we can pick
$s = \lfloor \log( x / 3 + 2 ) \rfloor - 1 > \log (x / 3 + 2) - 2$.
We then obtain
  \[
  \phi(Z(n), x)\geq \log (x/3+2)-2\geq \log (x/12).\qedhere
  \]
\end{proof}

\begin{prop}
\label{app:psiZ}
If $x \geq 2$ and $n \geq 1$, then $\psi(Z(n), x)\geq 2^x$.
\end{prop}

\begin{proof}
From Facts~\ref{app:heightX} and~\ref{app:minusesX} we conclude that
a word with height increase at least~$x$
cannot occur to the left of the factor~$-^x$
earlier than the factor~$Z(x)$ occurs.
But the length of this factor is $2^{x+1}-2$.
For $x\geq 2$ the inequality $(2^{x+1}-2) - 1\geq 2^x$ holds,
and the desired inequality on~$\psi(Z(n), x)$ follows.
\end{proof}

We next obtain bounds on~$\phi(Y(m,\ell),x)$ and~$\psi(Y(m,\ell), x)$.

\begin{prop}
\label{app:phiY}
If $x\geq\max(6 \Htmax(Y(m,\ell)), 9\cdot 6^\ell)$,
then
  \[
  \phi(Y(m,\ell), x) \geq \frac{\ell}3\cdot\Big(\frac{x}9\Big)^{1/\ell}.
  \]
\end{prop}

\begin{proof}
Reasoning as in the proof of Claim~\ref{app:phiZ}
and using Facts~\ref{app:heightX} and~\ref{app:minusesX},
we conclude that a sub-factor which consists of $\Htmax(Y(s,\ell))$~minuses
must occur in every factor of length at least
$2|Y(s,\ell)| + 2\Htmax(Y(m,\ell))$,
and, moreover,
for this inequality to hold under the given assumption on~$x$,
it would suffice that $|Y(s,\ell)| \leq x/3$.

The upper bound from Claim~\ref{app:lenY} means that
the condition $3\cdot 2^{s\ell}\leq x/3$ is sufficient.
This condition holds, if
  \[
  s=\left\lfloor \frac1\ell\cdot \log\Big(\frac x9\Big)\right\rfloor
  \geq \frac1\ell\cdot \log\Big(\frac x9\Big)-1.
  \]
Applying equation~\eqref{app:maxHtY} for the maximum height
of positions in the word $Y(s,\ell)$ (see p.~\pageref{app:maxHtY}),
we obtain
  \begin{equation*}
  \phi(Y(m,\ell), x) \geq \ell\cdot(2^{ \frac1\ell\cdot \log(\frac
  x9)-1} -1) =
  \ell\cdot\Bigg(\frac12\cdot\Big(\frac{x}9\Big)^{1/\ell}-1\Bigg)\geq
  \frac{\ell}3\cdot\Big(\frac{x}9\Big)^{1/\ell}. 
  \end{equation*}
The last inequality holds by the assumption that $x\geq 9\cdot 6^\ell$.
\end{proof}

\begin{prop}
\label{app:psiY}
The following inequality holds:
 $$ \psi(Y(m,\ell), x)\geq \Big( \frac{x}{2\ell}\Big)^\ell .$$
\end{prop}

\begin{proof}
The idea is the same as in the proof of Claim~\ref{app:psiZ}.

Let $s$ be such that $x>\ell\cdot(2^s-1)$.
Then from Fact~\ref{app:minusesX} and equation~\eqref{app:maxHtY}
for the maximum height of positions in the word $Y(s,\ell)$
(see p.~\pageref{app:maxHtY}) it follows that
the factor $-^x$ can be extended leftwards to an occurrence
of a word~$Y(s,\ell)$;
and then by Fact~\ref{app:heightX} this word 
has no sub-factor
on which the height increases by~$x$ or more.

The inequality $x>\ell\cdot(2^s-1)$ holds for
  \[
  s = \left\lceil\log\Big( \frac{x}\ell +1\Big)\right\rceil-1
  \geq\log\Big( \frac{x}\ell +1\Big)-1> \log\Big( \frac{x}{2\ell}\Big) .
  \]

Since the length of the word $Y(s,\ell)$ is at least
$2^{s\ell}\geq\big( \frac{x}{2 \ell} )^\ell $,
this gives us the desired bound.
\end{proof}

\subsection{Proof of Theorem~\ref{lwr-sqrtlog}}
\label{s:lb:final-calculations}

Apply Lemma~\ref{main-length} to the words $Y(\ell) = Y(m,\ell)$,
where $m=\lfloor \ell\cdot\log \ell\rfloor$,
for big enough~$\ell$.
We will rely on the bounds on the functions $\phi(Y(\ell),x)$
and $\psi(Y(\ell), x)$ obtained in the previous subsection.

In this part of the present subsection,
we will denote $L(k) := L(Y(\ell),k)$
and $H := \Htmax(Y(\ell))$.

We obtain from Lemma~\ref{main-length} that,
if
\begin{equation}
\label{app:eq:applicability:scaled}
L(k-1)\geq 17 k H \cdot \max (17 H, 9\cdot 6^\ell),
\end{equation}
then the following inequalities hold:
\begin{multline*}
   3kL(k-1)\geq 
   \psi \left(Y(\ell), \ \frac{1}{6k}\phi\Big(\frac {L(k)}{17kH}\Big)\right)
  \geq \\
\psi\left(Y(\ell), \ \frac{1}{6k}\cdot \frac\ell3\cdot\Big(\frac
   {L(k)}{9\cdot17kH}\Big)^{1/\ell}\right)\geq \\
\left(\frac1{2\ell}\cdot\frac{1}{6k}\cdot \frac\ell3\cdot\Big(\frac
   {L(k)}{9\cdot17kH}\Big)^{1/\ell}\right)^\ell = 
\frac1{(36k)^\ell}\cdot \frac
   {L(k)}{9\cdot17kH}\,,
\end{multline*}
These inequalities imply the following recurrence for the numbers~$L(k)$:
\begin{equation*}
  L(k)\leq 27\cdot17k^2H(36k)^\ell L(k-1).
\end{equation*}

Apply this recurrence while it is still applicable,
i.e., while the condition~\eqref{app:eq:applicability:scaled} holds.
It suffices to use a rather crude approximation for the contribution
of different factors.
If $\ell\geq 2$, then the inequalities $\ell+2\leq 2\ell$ and
\begin{equation}\label{app:L-upbnd}
L(k)\leq (36k)^{2(k-2)\ell}H^{k-2}\cdot 17^2 k H^2 \cdot 9\cdot 6^\ell
    \leq (36 k)^{2 k \ell} H^k
\end{equation}
hold.
(Note that $L(1) \leq 4 H$ by Lemma~\ref{base-length} (see p.~\pageref{base-length}),
 so for $k = 2$ the inequality~\eqref{app:eq:applicability:scaled} fails
 and, therefore, for some $k' \geq 2$ the following upper bound holds:
 $L(k') < 17 k' H \cdot \max(17 H, 9 \cdot 6^\ell)$.
 This implies, for instance, that the recurrence is applied at most $k - 2$ times.)

Choose $k = \lfloor \ell/4 \rfloor \leq \ell/4$.
Then the first factor in~\eqref{app:L-upbnd} is at most
\[
(36k)^{2k\ell} \le (9\ell)^{\ell^2 / 2}.
\]
The factor $H^{k}$ in~\eqref{app:L-upbnd} can be bounded,
using the expression~\eqref{app:maxHtY} for the maximum height of positions
in the word $Y(m,\ell)$,
by
\[
H^{k} = \Big(\ell \cdot (2^m-1) \Big)^{k} < \ell^{\ell/4}\cdot
2^{\ell\cdot\log\ell \cdot \ell/4} =
2^{(\ell \cdot \log \ell + \log \ell) \cdot \ell / 4}.
\]

Now recall that if $L(k)<|Y(\ell)|$, then $\Wd(Y(\ell)) \ge k$.
Compare the bound~\eqref{app:L-upbnd} with the lower bound on the length
of the word $Y(\ell)$, taking into account the upper bounds
for the individual factors above.
Taking the logarithms of both sides,
we get
$\Wd(Y(m,\ell)) \geq \lfloor \ell/4 \rfloor \geq \ell/4 - 1$,
provided that
\begin{equation*}
\frac{\ell^2}{2} \cdot \log(9 \ell) + (\ell \cdot \log \ell + \log \ell) \cdot \frac{\ell}{4}
<
(\ell \cdot \log \ell - 1) \cdot \ell
\end{equation*}
(the right-hand side of this inequality does not exceed
 the logarithm of the length of $Y(\ell)$).

The main terms on both sides of this inequality
are those with $\ell^2\cdot\log\ell$.
On the left-hand side, this expression has coefficient~$1/2+1/4=3/4$,
and on the right-hand side, coefficient~$1$.
All other terms are $o(\ell^2\log\ell)$, so the inequality
holds for all big enough~$\ell$.

This gives us the upper bound $\Wd(Y(\ell)) \ge \ell/4 - 1$,
which holds for big enough~$\ell$.
The inequality~\eqref{WdY}
(see p.~\pageref{WdY})
follows.

Let us now show that this inequality implies Theorem~\ref{lwr-sqrtlog}.
From Claim~\ref{app:lenY} we obtain an upper bound on the length
of the word $Y(\ell)$:
\[
|Y(\ell)| \leq 3\cdot 2^{m\ell} <3\cdot 2^{\ell^2\cdot
 \log\ell}. 
\]
It follows from this bound that
\[
\ell = \Omega\left(\sqrt{\frac{\log |Y(\ell)|}{\log\log
    |Y(\ell)|}}\,\right), 
\]
which completes the proof of Theorem~\ref{lwr-sqrtlog}.

\bigskip

The lower bound of Theorem~\ref{lwr-sqrtlog} is much higher
than the upper bound on the width of the words associated
with complete binary trees, $Z(n)$, given by Theorem~\ref{Zn-upbnd}.
The same method gives a much weaker
lower bound bound for $Z(n)$, namely the inequality~\eqref{WdZ}.

\subsubsection*{Proof of inequality~\eqref{WdZ}}

We apply Lemma~\ref{main-length} to the words $Z(n)$
for big enough~$n$.
We will rely on the bounds on the functions $\phi(Z(n),x)$ and
$\psi(Z(n), x)$ obtained in the previous subsection.

In this part of the present subsection,
we will denote $L(k) := L(Z(n),k)$.

If the conditions
\begin{gather}
\label{app:eq:applicability:complete}
L(k-1)\geq 17^2 k n^2, \qquad\text{and}\\
\label{app:eq:applicability:complete-additional}
\frac{1}{6 k} \cdot \log \Big( \frac{L(k)}{12 \cdot 17 k n} \Big) \ge 2
\end{gather}
are satisfied,
then the following inequalities hold:
\begin{multline*}
3kL(k-1)
\geq\psi\left(Z(n), \ \frac1{6k}\cdot\phi\Big(\frac{L(k)}{17kn}\Big)\right)
\geq \\
\psi\left(Z(n), \ \frac1{6k}\cdot \log \Big(\frac{L(k)}{12\cdot17kn}\Big)\right)
\geq 
\Big(\frac{L(k)}{12\cdot17kn}\Big)^{1/6k}.
\end{multline*}
The last inequality implies the following recurrence for the numbers~$L(k)$:
\[
L(k) \leq 12\cdot17kn\cdot (3k)^{6k} \left(L(k-1)\right)^{6k}. 
\]
Taking the logarithms of both sides and coarsening slightly,
we get
\begin{equation*}
\log L(k) \leq 
6k\cdot\log L(k-1) + 6k\cdot\log (3k) + 
\log (12\cdot17kn) < 7k\cdot\log L(k-1),
\end{equation*}
provided that
\begin{equation}
\label{app:eq:applicability:complete-final}
L(k-1)>17^2 k n^2\cdot(3k)^{6k}.
\end{equation}
Condition~\eqref{app:eq:applicability:complete-final}
is stronger than the condition~\eqref{app:eq:applicability:complete}
of applicability of Lemma~\ref{main-length},
but will suffice for us.

Apply this recurrence while it is still applicable,
i.e., while the conditions~\eqref{app:eq:applicability:complete-additional}
and~\eqref{app:eq:applicability:complete-final} hold
(notice that the former condition is equivalent
 to the inequality $L(k) \ge 12 \cdot 17 k n \cdot 2^{12 k}$).
This will give us the following upper bound on~$\log L(k)$:
\begin{equation*}
\log L(k) \leq
(7k)^k \log \max
\Big( 12 \cdot 17 k n \cdot 2^{12 k},\ 17^2 k n^2 \cdot (3 k)^{6 k} \Big)
\leq
(7k)^k \Big(12 k\cdot\log (3k) + 
\log (17^2 k n^2)\Big)
\end{equation*}

Plug in $k = \lfloor 0.9\log n/\log\log n \rfloor$.
For big enough~$n$, we have
\[
\begin{aligned}
  &k^k \leq\left(\frac{0.9\log n}{\log\log n}\right)^{0.9\log
  n/\log\log n} < n^{0.9} \text{\quad{} и}\\
  &7^k\Big(12 k\cdot\log (3k) + 
\log (17^2 k n^2)\Big) =  n^{o(1)}.
\end{aligned}
\]
Therefore, for big enough~$n$ we have
\[
\log L(k) < n^{0.9 +o(1)} <n < \log |Z(n)|,
\]
and then inequality~\eqref{WdZ} follows, i.e.,
\[
\Wd(Z(n)) = \Omega\left(\frac{\log n}{\log\log n}\right).
\]

\section{Proofs of lower bounds for commutative NFA}
\label{app:oca}

Recall that a~\emph{nondeterministic finite automaton} (NFA)
is a quintuple $(\Sigma, Q, q_0, \delta, F)$ where
$\Sigma$ is a (finite) input alphabet,
$Q$ a (finite) set of states,
$q_0 \in Q$ an initial state,
$\delta \sset Q \times \Sigma \times Q$ a transition relation, and
$F \sset Q$ a set of final states.
The \emph{transition graph} of the NFA
is a~directed graph with vertices~$Q$ that contains,
for each $(q, a, q') \in \delta$, an edge $(q, q')$ labeled
by the symbol~$a$.
The NFA accepts a word~$\sigma \in \Sigma^*$
if the graph has a path from~$q_0$ to some~$q \in F$
such that the labels of the edges on the path
form the word~$\sigma$.
The set of all accepted words is the \emph{language}
recognized by the NFA.

\subsection{Vertices and paths in the graph of strongly connected components}
\label{app:oca:nfa}

Let $\A = \A_n$ be an NFA which recognizes
a language $L(\A)$ with
Parikh image equal to $U_n$, the universal one-counter set (see Section~\ref{s:oca}).

Denote by $Q_0$, $Q_1$, \ldots, $Q_s$ the (strongly connected)
components in the transition graph of the automaton~$\A$
which are reachable from the initial state
and from which a final state is reachable.
Assume without loss of generality that $\A$ has no
other strongly connected components.
Suppose the initial state of~$\A$ belongs to the component~$Q_0$.

We first make two observations about the structure
of the transition graph.

The first observation is that every edge labeled $c_{ij}$
on a path from the initial state to a final state
goes from one strongly connected component to another.
Indeed, if this is not the case, then
this edge belongs to a cycle in the graph.
But then the Parikh image of $L(\A)$ violates condition~\eqref{U1}
in the definition of the universal one-counter set
(by including vectors that have $y_{ij} > 1$).

The second observation concerns the directed cycles
in the transition graph of the automaton~$\A$.
Let $C$ be such a cycle.
Denote by $x_i^{(C)}$ the number of edges in this cycle
that are labeled by the symbol~$a_i$.

\begin{prop}
\label{cycle-balanced}
For the edges of every cycle~$C$
in the transition graph of the automaton~$\A$,
the vector $(x_i^{(C)})$ belongs to the cone of balanced vectors.
\end{prop}

\begin{proof}
Consider paths from the initial state to a final state
that differ in the number of traversals of the cycle~$C$ only.
If the cycle~$C$ violates the balance condition,
then, after sufficiently many traversals,
the balance condition will be violated for the entire path too---%
but this contradicts the condition~\eqref{U3}
in the definition of the universal one-counter set.
\end{proof}

We construct, based on the transition graph of~$\A$,
another graph~$\B$, the \emph{condensation} of it
(or the \emph{graph of strongly connected components}).
The vertices of this graph are strongly connected components
in the transition graph of~$\A$.
The edges of $\B$ correspond to the edges of $\A$
between different components.
An edge in~$\B$ has label~$(i,j)$ if
the corresponding transition in~$\A$ has label~$c_{i,j}$;
and label~$\eps$ if the transition in~$\A$ has label~$\eps$
or~$a_i$ for some~$i$.
By construction, the graph~$\B$ may have parallel edges
and is acyclic.

Vertices of~$\B$ that contain initial (resp.\ final) state(s)
of the NFA~$\A$ are called \emph{initial} and \emph{final},
respectively.
A path from the initial vertex to a final vertex
is \emph{a complete path}.

Recall that chains are defined
on p.~\pageref{U2},
when we describe condition~\eqref{U2}.
Chains on~$[n]$ are in bijection with sets of numbers from~$[n]$
that they visit.
We will denote the set of numbers visited by the edges
of~$\pi$ by~$C(\pi)$.
We have $\{0,n-1\}\sset C(\pi) \sset [n]$.

\begin{prop}
\label{prop:path-to-chain}
Let $\pi$ be a complete path in the graph~$\B$.
Then the non-empty labels on (the edges of) this path
induce a chain on the set~$[n]$.
\end{prop}
\begin{proof}
Given~$\pi$, construct an accepting path~$\pi'$ in~$\A$
that consists of the edges that correspond to the edges of the path~$\pi$,
as well as edges inside the strongly connected components.
Since all edges labeled~$c_{ij}$ in the transition graph of~$\A$
connect states that lie in different components,
the set of non-empty labels on the edges of the path~$\pi$
corresponds to a set of labels $c_{ij}$ of edges on the path~$\pi'$.
Therefore, the claim follows from
the property~\eqref{U2} of the universal one-counter set.
\end{proof}

Let us now assign labels to vertices of~$\B$.
For each vertex~$v$, we define four subsets of
$[n] = \{0, 1, \ldots, n-1\}$ that we denote
$I(v)$, $L(v)$, $R(v)$, and $S(v)$.
The intuition is as follows:
\begin{itemize}
\item $S(v)$ is the set of all numbers~$i\in[n]$
such that the corresponding component in~$\A$
contains a edge with label~$a_i$;
\item $I(v)$ specifies the (endpoints of) disjoint intervals in the set $[n]$
that are induced by (incomplete) paths from the initial vertex
to~$v$;
\item $L(v)$ contains points from the interior of these intervals
that must be visited by (all) these incomplete paths;
\item $R(v)$ contains points from the exterior of these intervals
that must be visited by the completion of (all) these paths
(i.e., by paths from~$v$ to final vertices).
\end{itemize}
Rigorous definitions of the sets $I(v)$, $L(v)$ and $R(v)$ are given below.

The following property links vertices to paths. Let 
$B(v)= S(v)\cup I(v)\cup L(v)\cup R(v) $.

\begin{prop}
\label{prop:oca:vertices-and-paths}
For every vertex~$v$ in the graph~$\B$
and every complete path~$\pi$ that visits~$v$,
the inclusion
$B(v)\subseteq C(\pi)$
holds.
\end{prop}

We  prove
Claim~\ref{prop:oca:vertices-and-paths}, considering sets $ S(v)$,
$I(v)$, $ L(v)$, $ R(v)$ one by one.

Since a chain is a monotone path,
it is uniquely determined by the set of numbers from the set~$[n]$
that it visits
(i.e., the set of vertices that are incident to at least one
 edge on the path).
By Claim~\ref{prop:path-to-chain}, the edges of a complete path~$\pi$
induce a chain. We will denote the set of vertices visited by
this chain by~$C(\pi)$.

For every complete path, we have $\{0,n-1\}\sset C(\pi)$.
The converse also holds:

\begin{prop}
\label{S-path}
For all $\{0,n-1\}\subseteq S\subseteq[n]$
there exists a complete path~$\pi$
such that $C(\pi)=S$.
\end{prop}
%
%

\begin{proof}
Let $S$ consist of numbers
  \[
  0 = i_0<i_1<i_2<\ldots< i_t=n-1.
  \]
The vector of numbers
  \[
  y_{ij} = \left\{
  \begin{aligned}
    1& &&\text{if } i = i_\al,\ j=i_{\al+1},\quad\text{and}\\
    0& &&\text{otherwise,}
  \end{aligned}\right.\qquad x_i=0
  \]
belongs to the universal one-counter set.
Therefore, the automaton~$\A$ has an accepting path
with this Parikh image.
This path induces the required path in~$\B$.
\end{proof}

\begin{prop}
\label{inclusion}
Let a vertex~$v$ in the graph~$\B$
belong to a complete path~$\pi$.
Then $S(v)\subseteq C(\pi)$.
\end{prop}

\begin{proof}
This is a consequence of the compatibility condition~\eqref{U4}
in the definition of the universal one-counter set.
By definition,
if $i\in S(v)$, then the component~$Q_v$ in the transition graph of~$\A$
contains an edge~$e$ with label~$a_i$ and there is an accepting path~$\pi^{-1}$
in the transition graph of~$\A$ such that
$\pi^{-1}$ traverses this edge and 
corresponds to the path~$\pi$ in the graph~$\B$.
The Parikh image of the word 
induced by this path~$\pi^{-1}$
belongs to the universal one-counter set~$U_n$,
so it satisfies the compatibility condition.
We have $x_i > 0$;
therefore, $\pi^{-1}$ traverses an edge with label~$c_{i,x}$
or~$c_{x,i}$ for some~$x$,
so the chain on~$[n]$ induced by the path~$\pi$
visits the number~$i$. But then $i\in C(\pi)$.
\end{proof}

Consider, in the graph~$\B$,
a path from the initial vertex to a vertex~$v$
that is a prefix of some complete path.
Take the edges with labels of the form~$(i,j)$ on this (incomplete) path.
They form a subgraph of the chain
and are therefore a disjoint union of non-intersecting paths:
from the (vertex) number~$\ell_0$ to~$r_0$;
from~$\ell_1$ to~$r_1$; \ldots;
from~$\ell_t$ to~$r_t$.

\begin{prop}
Let $v$ be a vertex in the graph~$\B$.
Then for every path in~$\B$ from the initial vertex to~$v$
the set of numbers
  \[
  \ell_0<r_0<\ell_1<r_1<\dots< \ell_t<r_t
  \]
is the same.
\end{prop}

\begin{proof}
Observe that every path from~$v$ to a final vertex
is a completion of any path from the initial vertex to~$v$.

Assume for the sake of contradiction that
for paths~$\pi$ and~$\pi'$ from the initial vertex to~$v$
the sequences
  \begin{equation}\label{2Iseq}
    \begin{aligned}
    &\ell_0<r_0<\ell_1<r_1<\dots< \ell_t<r_t,\\
    &\ell'_0<r'_0<\ell'_1<r'_1<\dots< \ell'_{t'}<r'_{t'}    
    \end{aligned}
  \end{equation}
are different.
Consider the first difference and suppose it is the left endpoint
of a segment. Without loss of generality, $\ell_i< \ell'_i$.
Take some path~$\pi''$ from~$v$ to a final vertex.
The unions~$\pi\cup\pi''$ and~$\pi'\cup\pi''$ should induce chains
on~$[n]$; we will prove that this is impossible.

Let $\pi\cup\pi''$ and~$\pi'\cup\pi''$ induce two chains.
Then the path~$\pi''$ has no edges with labels of the form~$(\ell_i,x)$,
because such a label already exists on the path~$\pi$.
Since we picked the first difference between the sequences~\eqref{2Iseq},
there are no edges with such labels on the path~$\pi'$ either.
Therefore, there are no edges with such labels on the entire path~$\pi'\cup\pi''$.
So $\ell_i\ne0$ (since the chain starts from~$0$).

But then the path~$\pi''$ has an edge with label~$(y,\ell_i)$,
because $\pi\cup\pi''$ induces a chain.
It follows that $\pi'\cup\pi''$ does not induce a chain,
since the in-degree of~$\ell_i$ is at least~$1$, and
the out-degree, as shown above, is~$0$. This is a contradiction.

The case of the first difference in the right endpoint of a segment
is analogous.

Previous arguments assumed implicitly that $i\leq t'$. To complete the
proof, we note that in the case~$ t'< i$ (i.e., when the segment~$[\ell_i',r_i']$ is absent from the sequence
for the path~$\pi'$)  our analysis
of the first difference in a left endpoint applies as well.
\end{proof}

By this claim, for every vertex~$v$ in the graph~$\B$,
the  set $I(v)$ is well-defined:
\[
I(v) = \{\ell_0,r_0,\ell_1,r_1,\dots,\ell_t,r_t\}.
\]
The next claim is an immediate corollary of the previous one.

\begin{prop}
\label{intervals}
$I(v)\subseteq C(\pi)$
for all complete paths that visit vertex~$v$.
\end{prop}

The set $I(v)$ specifies a system of intervals inside $[0; n-1]$.
The numbers in the interior of any of the intervals
will be called \emph{internal} for a vertex~$v$,
and the numbers outside the intervals will be called \emph{external}.
More precisely,
a number~$j$ is internal for a vertex~$v$
if $\ell_i< j< r_i$ for some~$i$.
A number~$j$ is external for a vertex~$v$
if either $j<\ell_0$, or $r_i< j< \ell_{i+1}$ for some~$i$, or $j> r_t$.

To each vertex~$v$ we will associate two more subsets~%
$L(v)$, $R(v)$ of the set~$[n]$, defined as follows:
\begin{itemize}
\item $j\in L(v)$ iff
(i) $j$~is internal for~$v$ and
(ii)~every path~$\pi$ in the graph~$\B$ from the initial vertex to~$v$
has at least one edge with label of the form~$(x,j)$ or~$(j,x)$.
\item $j\in R(v)$ iff
(i) $j$~is external for~$v$ and
(ii)~every path~$\pi$ in the graph~$\B$ from~$v$ to a final vertex
has at least one edge with label of the form~$(x,j)$ or~$(j,x)$.
\end{itemize}

The following statement follows directly from these definitions.

\begin{prop}
\label{LRinC}
For every complete path~$\pi$ in the graph~$\B$ that visits a vertex~$v$,
the inclusions~$L(v)\subseteq C(\pi)$ and~$R(v) \subseteq C(\pi)$ hold.
\end{prop}

Now Claim~\ref{prop:oca:vertices-and-paths} follows from
Claims~\ref{inclusion}, \ref{intervals}, and~\ref{LRinC}. 

The sets~$L(v)$ and~$R(v)$ restrict possible elements of the sets~$S(v')$
for other vertices~$v'$ as follows:

\begin{prop}
\label{crossing}
Let a vertex~$v_1$ be visited by some complete path before another vertex~$v_2$.
Then the set~$L(v_1)$ contains all numbers from~$S(v_2)$ that are internal for~$v_1$,
and the set~$R(v_2)$ contains all numbers from~$S(v_1)$ that are external for~$v_2$.
\end{prop}

\begin{proof}
Suppose~$v_1$ occurs before~$v_2$ on some complete path~$\pi$;
also suppose that the number~$i$ is internal for~$v_1$ and belongs to~$S(v_2)$.

We show by contradiction that $i\in L(v_1)$.
Assume that there is a path~$\pi'$ from the initial vertex to~$v_1$
that has no edges with labels of the form~$(x,i)$ and~$(i,x)$.
Denote by~$\pi''$ the part of the path~$\pi$ from~$v_1$ to the final vertex.
Then $\pi'\cup\pi''$ is a complete path,
the vertex~$v_2$ is visited by this path,
but $i\notin C(\pi'\cup\pi'')$.
This contradicts Claim~\ref{inclusion}.

The second assertion is proved in a similar way.
\end{proof}

The strategy of the rest of the proof is as follows.
In the following subsection,
we describe a construction that
produces paths~$\pi$ with `predetermined' $C(\pi)$.
Our construction will have the following properties:
\begin{itemize}
\item Each path~$\pi$ is originally chosen based (in a certain way)
on a well-formed (Dyck) word~$\sigma$; we show that
from~$\pi$ one can obtain a re-pairing of this word~$\sigma$
(based on Claim~\ref{bijection} below).
\item The width of this re-pairing will bound the cardinality
of the 
$B(v)$ 
from below, for \emph{some} vertex~$u = u(\pi)$ on the path~$\pi$
(Claim~\ref{width-IS}).
\item Many (different) paths~$\pi$ will be chosen
and it will be ensured that the sets $C(\pi)$ have
low pairwise intersection
(subsection~\ref{app:oca:proof}).
\end{itemize}
Now suppose that the vertex~$u = u(\pi)$ for a path~$\pi$ is visited by
\emph{another} such path, say $\pi'$.
Then the intersection of the sets $C(\pi')$ and $C(\pi)$
includes 
$B(u)$
by Claim~\ref{prop:oca:vertices-and-paths}.
On the one hand,
the cardinality~$|B(u)|$ is greater than or equal to the width of the re-pairing;
on the other hand,
no two sets $C(\pi')$ and $C(\pi)$ may overlap a lot.
Therefore, with a careful choice of parameters,
vertices~$u = u(\pi)$ are not shared by the paths~$\pi$,
and so the NFA~$\A$ should have at least as many strongly connected
components as we can choose paths~$\pi$.

\subsection{From NFA and well-formed word to re-pairing}
\label{app:oca:many}

Let $\sigma$ be a well-formed (Dyck) word of length~$ s \leq \frac{n}{2}$.
As everywhere in section~\ref{s:oca}, we assume that $n$~is even.
Associate with the word~$\sigma$ a set of vectors from the universal
one-counter set. Each of these vectors will be determined by some
$ s $-element subset~$F\supseteq \{0,\frac{n}{2} - 1\}$
of the set~$[\frac{n}{2}]$ and by a nonnegative integer~$\ld$;
we will denote these vectors $(\vc y_F(\ld), \vc x_F(\ld))$.
(These vectors depend on the word~$\sigma$ as well,
 but our notation will not reflect this.)

\begin{defi}
[the set $C_F$ and vectors $(\vc y_F(\ld),\vc x_F(\ld))$]
\label{def:many}
Given an $ s $-element subset~$F$ of the set~$[\frac{n}{2}]$,
define a subset~$C_F$ of the set~$[n]$ of the same size~$ s $
as follows.
Sort the elements of~$F$ in ascending order.
If the $i$th least element is equal to~$j$,
then the set~$C_F$ contains the number~$2j$ if $\sigma(i) = +1$
and the number~$2j+1$ otherwise.
(Since every well-formed word begins with a~$+1$
and ends with a~$-1$,
and $F\supseteq \{0,\frac{n}{2} - 1\}$,
we will always have $\{0,n-1\}\sset C_F$.)

The vectors $(\vc y_F(\ld),\vc x_F(\ld))$ are then defined
by the following equations:
$y_{ij}=1$ if $i$~and $j$~are two adjacent elements (in ascending order)
in the set~$C_F$, otherwise $y_{ij}=0$;
$x_i=\ld$ if $i\in C_F$, otherwise $x_i=0$.
\end{defi}

\begin{Remark}
Notice that the correspondence between the elements of the set~$C_F$
and the symbols of the word~$\sigma$ defined in this way is a bijection.
We will refer to symbols in the word~$\sigma$ by specifying the corresponding
numbers from the subset~$C_F$.
\end{Remark}

\begin{ex}
Let $\sigma = Z(2) = {+}{+}{-}{+}{-}{-}$;
$s = |\sigma| = 6$ and $n = 16$;
$s \le \frac{n}{2} = 8$.
Consider the set $F = \{ 0, 1, 3, 5, 6, 7 \} \sset [ 8 ] = \{ 0, 1, \ldots, 7 \}$.
Then $C_F = \{ 0, 2, 7, 10, 13, 15 \}$,
the vector $(\vc y_F(\ld), \vc x_F(\ld))$ satisfies the equations
$y_{0,2} = y_{2,7} = y_{7,10} = y_{10,13} = y_{13,15} = 1$ and
$x_0 = x_2 = x_7 = x_{10} = x_{13} = x_{15} = \ld$,
and all other components of this vector are equal to~$0$.
\end{ex}

It is easy to see that the vectors $(\vc y_F(\ld), \vc x_F(\ld))$ defined in this way
belong to the universal one-counter set.
Conditions~\eqref{U1} and~\eqref{U4} hold by construction,
and condition~\eqref{U3} by the fact that the word~$\sigma$ is well-formed.

Therefore, for every~$\ld$,
the transition graph of the NFA~$\A$ has at least one path~$\tau(\ld)$
from the initial state to a final state
on which the labels form a word with Parikh image~$(\vc y_F(\ld),\vc x_F(\ld))$,
i.e., each symbol~$c_{i j}$ occurs $y_{i j}$~times
and each~$a_i$ occurs $x_i$~times in this word.
We can now decompose this path by separating (possibly empty) cycles~$\gamma_i(\ld)$:
\begin{equation}\label{decomposition2cycles}
\tau(\ld) = q_0 \gamma_0(\ld) q_1(\ld) \gamma_1(\ld)\dots q_T(\ld)
\gamma_T(\ld),  
\end{equation}
where all~$q_i(\ld)$ are distinct states of~$\A$.
It is clear that each cycle~$\gamma_i(\ld)$ lies in some
strongly connected component of the transition graph of~$\A$.

Since the set of possible values of the parameter~$\ld$ is infinite,
there exists a sequence $\tau^*_F=(q_0, q_1, \dots, q_T)$
of distinct states of the NFA~$\A$ that coincides
with infinitely many sequences $(q_0, q_1(\ld), \dots, q_T(\ld)) $
that occur in the decomposition~\eqref{decomposition2cycles}.

The sequence~$\tau^*_F$ corresponds to a complete path in the graph~$\B$
up to the choice of parallel edges.
Choose these edges, and thus a complete path~$\pi_F$,
in such a way that~$\pi_F$ corresponds (in the graph~$\B$)
to infinitely many values of the parameter~$\ld$.
Denote the vertices of this path\label{v-in-B}
\[
v_0^{(F)}, \ v_1^{(F)},\ \dots,\ v_T^{(F)}.
\]
The following claim holds by definition of the vectors~$\vc y(\ld)$.

\begin{prop}
\label{C(F)=C_F}
$C(\pi_F) = C_F$.
\end{prop}

In what follows,
we will use a connection between accepting paths in the automaton~$\A$
and re-pairings of the word~$\sigma$.
We construct
such a re-pairing
based on the sets~$S(v_i^{(F)})$. 

Also we need the following simple claim.

\begin{prop}
\label{K-gen}
The cone~$K$ of balanced vectors is generated by the vectors
    \[
    \vc e_{i,j} = (\underbrace{0,\dots,0}_{i},1,
    \underbrace{0,\dots,0}_{j-i-1},1,0,\dots,0), \ 
    \text{$i$ even, $j$ odd, $i<j$}.
    \]
\end{prop}
%

\begin{proof}
It is obvious that all vectors~$\vc e_{i,j}$ belong to~$K$.

Given a vector from~$K$,
we will show that it is a nonnegative linear combination
of the vectors~$\vc e_{i,j}$.
The proof is by induction on~$M$, the number of nonzero coordinates
in the vector.

The base cases are $M=1$ and $M=2$.
For $M = 1$, the balance conditions cannot be satisfied,
and for $M = 2$, all balanced vectors are proportional to~$\vc e_{i,j}$.

For the induction step, assume that the claim holds for all~$M<M_0$.
Given a vector~$\vc x\in K$ with~$M_0$ nonzero coordinates,
consider the smallest nonzero coordinate~$j$ with an odd index.
The sum of the coordinates with even indices~$i<j$ is at least~$x_j$
(by the balance inequality),
so by subtracting from~$x$ an appropriate vector of the form
    \[
    \sum_{2i<j} x'_{2i}\vc e_{2i,j}, \quad 0\leq x'_{2i}\leq x_{2i},\quad 
    \sum_{2i<j} x'_{2i} = x_j
    \]
we obtain a vector~$\vc x'\in K$ in which the number of nonzero
coordinates is smaller than~$M_0$.
\end{proof}

\begin{prop}
\label{bijection}
The set $C_F$ can be decomposed as a union
of non-overlapping pairs $\{\ell_i,r_i\}_{i=1}^{s/2}$
such that for all~$i$ there is a~$k$ with $\{\ell_i,r_i\}\subseteq S(v_k^{(F)})$.
For each pair~$(\ell_i,r_i)$,
the corresponding pair of signs in~$\sigma$
is formed by a plus and a minus,
and the plus occurs to the left of the minus.
\end{prop}

\begin{proof}
It follows from the definition of the vectors $(\vc y_F(\ld), \vc x_F(\ld) )$
that for each label~$(i,j)$ on the edges of the path~$\pi_F$ in the graph~$\B$
the inclusion~$\{i,j\}\subseteq C_F$ holds.
By Claim~\ref{inclusion}, $S(v_k^{(F)})\subseteq C(\pi_F)$,
and $C(\pi_F) = C_F$ by Claim~\ref{C(F)=C_F}.

Recall that, for every~$\ld$,
the transition graph of the NFA~$\A$ has at least one path~$\tau(\ld)$
from the initial state to a final state,
on which the labels form a word with 
Parikh image~$(\vc y_F(\ld),\vc x_F(\ld))$, 
i.e., each symbol~$c_{i j}$ occurs $y_{i j}$~times
and each~$a_i$ occurs $x_i$~times in this word.

Decompose the path~$\tau(\ld)$ into cycles in a greedy way,
by reading the sequence of states $q_0,q_1, \dots, q_t$ of the automaton~$\A$
specified by the path~$\tau(\ld)$, \emph{from left to right}.
For each state~$q_i$, we find the last occurrence of~$q_i$ in the sequence:
$q_k = q_i$ and $q_j\ne q_i$ for all~$j>k$.
The subsequence $q_i, q_{i+1},\dots, q_k$ forms a cycle
in the transition graph of the automaton~$\A$;
denote this cycle~$\gamma_i(\ld)$.
After this, pick~$q_{k+1}$ as the next state.
Alternatively,
if the state~$q_i$ never occurs further on the path~$\tau(\ld)$,
we set~$k=i$ and let~$\gamma_i(\ld)$ be the empty cycle.
We can now decompose this path by separating (possibly empty)
cycles~$\gamma_i(\ld)$, 
as indicated in~\eqref{decomposition2cycles}:
\[
\tau(\ld) = q_0 \gamma_0(\ld) q_1(\ld) \gamma_1(\ld)\dots q_T(\ld)
\gamma_T(\ld),  
\]
Now each cycle $\gamma_i(\ld)$ belongs to an SCC in the transition graph
of $\A$, and the sequence of states $q_0, q_1(\ld), \ldots, q_T(\ld)$
corresponds to a complete path through the (acyclic) graph $\B$.
As previously in this section,
we denote the vertices of this path $v_0^{(F)}, v_1^{(F)}, \ldots, v_T^{(F)}$.
Note that
the balance condition holds for every cycle in the transition
graph of the automaton~$\A$ (by Claim~\ref{cycle-balanced}),
and thus for each $\gamma_i(\ld)$.
Since the vector~$\vc x_F(\ld)$ belongs to the cone of balanced vectors
as well,
the balance condition also holds for the multiplicities
of the labels~$a_i$ on the other edges of the path~$\tau(\ld)$.
Moreover, 
the number of these (other) edges cannot exceed~$|Q(\A)|$,
simply because our decomposition separated all cycles.

Therefore, if we take the expansion of the vector~$\vc x_F(\ld)$
as a nonnegative linear combination of the generators of the cone~$K$,
this expansion can be split into two terms as follows:
  \begin{equation}\label{SCC-partition}
     \vc x_F(\ld) =\sum_{k} \sum_{(i,j): \{i,j\}\in S(v_k^{(F)})}
   u_{ij}^{(k)}(\ld)\vc e_{i,j} 
   +\sum_{(i,j)} v_{ij}(\ld)\vc e_{i,j}.
  \end{equation}
Here
the outer summation in the first term
enumerates all vertices $v_k^{(F)}$ on the path,
and
the second term puts together the contribution of the edges
not included in the cycles~$\gamma_i(\ld)$.
We thus have $v_{ij}(\ld) \leq |Q(\A)|$.

Denote by~$Z^{\text{even}}_F$ the intersection of~$C_F$
with the set of even integers
and by~$Z^{\text{odd}}_F$ the intersection of~$C_F$
with the set of odd integers.
Notice that if $u_{i j}^{(k)}(\ld) > 0$ or
               $v_{i j}      (\ld) > 0$,
then $i \in Z^{\text{even}}_F$ and
     $j \in Z^{\text{odd}}_F$,
and that $|Z^{\text{even}}_F| = |Z^{\text{odd}}_F| = s / 2$.
Also denote
  \begin{equation*}
  u_{ij}(\ld) = \sum_k u_{ij}^{(k)}(\ld),
  \end{equation*}
then the coefficients of the expansion~\eqref{SCC-partition}
satisfy the equations
  \[
    \begin{aligned}
      & \sum_{i\in Z^{\text{even}}_F} u_{ij}(\ld) +\sum_{i\in
  Z^{\text{even}}_F} v_{ij}(\ld)=\ld, &&
  j \in Z^{\text{odd}}_F, \\
  & \sum_{j\in Z^{\text{odd}}_F} u_{ij}(\ld) +\sum_{j\in Z^{\text{odd}}_F} v_{ij}(\ld)=\ld,
  &&
  i \in Z^{\text{even}}_F,
    \end{aligned}
  \]
by definition of the vector~$\vc x_F(\ld)$.
(Each equation in this system corresponds to one coordinate
 of $\vc x_F(\ld)$.)

Since there are infinitely many possible values of~$\ld$,
and the coefficients~$v_{ij}(\ld)$ are upper-bounded by the
number of states of the automaton~$\A$,
it follows that the matrices~$(u_{ij}(\ld)/\ld)$
of dimension $(s / 2) \times (s / 2)$
(in which the rows are indexed with even numbers from~$C_F$,
 and the columns by odd numbers from~$C_F$)
have a limit point,
$(u^*_{ij})$, as $\ld \to \infty$.
(Recall that we picked the path through $\B$
 in such a way that it corresponds to infinitely many values of $\ld$.)
We see from the expansion~\eqref{SCC-partition} that
if $u^*_{ij}>0$, then the set~$\{i,j\}$ is contained
in some~$S(v_k^{(F)})$. Moreover, we have
  \[
  \sum_{i\in Z^{\text{even}}_F} u^*_{ij}=1,\qquad \sum_{j\in Z^{\text{odd}}_F} u^*_{ij}=1.
  \]
These conditions mean that $(u^*_{ij})$ is a doubly stochastic
matrix of size $(s/2) \times (s/2)$.
  %
  %
  %
%
By the Birkhoff---von Neumann theorem
(see, e.g.,~\cite[p.~301]{Schrijver03}),
it is a convex combination
%
%
of permutation matrices.
Take some permutation matrix that occurs in this convex combination
with a positive coefficient.

This permutation matrix specifies a bijection
$\al\colon Z^{\text{even}}_F\to Z^{\text{odd}}_F  $
between even and odd indices from~$C_F$.
The bijection has the following properties.
If $\al(2i)=2j+1$, then $2j+1>2i$,
since $u_{2i, 2j+1}(\ld)=0$ for $i>j$
due to the form of the generators of the cone of balanced vectors
(Claim~\ref{K-gen}).
Furthermore, every pair $\{2i, \al(2i)\}$ is included in some
set~$S(v_k^{(F)})$ by expansion~\eqref{SCC-partition}.
Since $Z^{\text{even}}_F \cup Z^{\text{odd}}_F = C_F$,
we obtain the equality
  \[
   \bigcup_k S(v_k^{(F)}) = C_F
  \]
and the required partitioning of the set~$C_F$ into pairs~$(2i, \al(2i))$.

Now recall that even numbers in~$C_F$ correspond to pluses in the word~$\sigma$
and odd numbers to minuses.
This correspondence is bijective by the construction of the set~$C_F$.
This proves the last assertion of the claim.
\end{proof}

Define a linear order on the pairs of signs
identified by
Claim~\ref{bijection}.
Sort the pairs $\{\ell_i,r_i\}$ in the order in which they
(or rather their~$k$---the one for which $\{\ell_i,r_i\}\subseteq S(v_k^{(F)})$)
occur along the path~$\pi_F$;
pairs that belong to the same set~$S(v_k^{(F)})$
can be ordered arbitrarily.
This gives us a re-pairing~$p_F$ of the word~$\sigma$.

\begin{prop}
\label{width-IS}
$\Wd(p_F)\leq \max_i\big|B(v_i^{(F)})\big|$.
\end{prop}

\begin{proof}
In this argument,
we will identify the signs of the word~$\sigma$
with the corresponding numbers from the set~$C_F$.

Recall that $B(v) =  S(v)\cup I(v)\cup L(v)\cup R(v)$.

The idea behind the assertion is as follows.
When the re-pairing has (just) erased
all the pairs included in the sets
$S(v_1^{(F)})$, $S(v_2^{(F)})$, \ldots, $S(v_t^{(F)})$,
the set of all erased signs
is the union of intervals
whose endpoints are specified by the set~$I(v_t^{(F)})$,
plus perhaps the signs that correspond to $R(v_t^{(F)})$,
  but except the signs that correspond to $L(v_t^{(F)})$.
At all other points in time, the set of all erased signs
is almost the same---except possibly for signs 
in one of the sets $S(v_i^{(F)})$.
We will now make this precise and provide justification.

Consider a time point in the re-pairing when
the pairs included in the sets
$S(v_1^{(F)})$, $S(v_2^{(F)})$, \ldots, $S(v_t^{(F)})$
have all been erased.
The set~$I(v_i^{(F)})$ defines $k=\frac12\big|I(v_i^{(F)})\big|$~intervals
  \[
  [\ell_0,r_0];\ [\ell_1,r_1];\ \ldots\ [\ell_{k-1},r_{k-1}].
  \]
Suppose the number~$i\in C_F$ belongs to one of these intervals
and has not been erased yet.
We then obtain from Claim~\ref{crossing} that $i\in L(v_i^{(F)})\cup I(v_i^{(F)})$.
Similarly, if a~$j\in C_F$ does not belong to these intervals
and has been erased already, then $j\in R(v_i^{(F)})$ by the same Claim.
Therefore, every sign in the word~$\sigma$ that has been erased by this time
either is covered by one of the $k / 2$~intervals,
%
%
or belongs to the set~$R(v_i^{(F)})$;
all the signs inside these $k/2$~intervals have been erased,
except maybe the elements of the set~$L(v_i^{(F)})$.
But this means that the set of all erased signs is a union of at most
$\frac12\big|I(v_t^{(F)})\big|+\big|L(v_t^{(F)})\big|+\big|R(v_t^{(F)})\big|$ intervals;
so the width of the re-pairing at this time point does not exceed this quantity.
(Whether the endpoints of the intervals are erased is irrelevant.)

Now consider an intermediate time point,
when the pairs included in the sets
$S(v_1^{(F)})$, $S(v_2^{(F)})$, \ldots, $S(v_t^{(F)})$
have all been erased, and
so have \emph{some} of the pairs included in the set~$S(v_{t+1}^{(F)})$.
We claim that
the width of the re-pairing at this time point
cannot differ by more than~$\big|S(v_{t+1}^{(F)})\big|$
from~the same width at time points
(i)~when the pairs included in the sets
$S(v_1^{(F)})$, $S(v_2^{(F)})$, \ldots, $S(v_t^{(F)})$
have all been erased or
(ii)~when the pairs included in the sets
$S(v_1^{(F)})$, $S(v_2^{(F)})$, \ldots, $S(v_{t+1}^{(F)})$
have all been erased.
Indeed, when a pair is erased, this can change (increase or decrease)
the number of intervals by at most~$2$.
It remains to observe that if one takes~$\leq \ell$ steps
of length at most~$2$ to move from a point~$X$ to point~$Y$ on the line,
then after each of these steps the distance from either~$X$ or~$Y$
to the current position does not exceed~$\ell$.
\end{proof}

\subsection{Proof of Theorem~\ref{OCA2NFA}}
\label{app:oca:proof}

We now obtain a lower bound on the number of states of the NFA~$\A$
from the construction described in the previous subsections.
We will pick in the set~$[\frac{n}{2}]$ an appropriate big enough family~$\F$
of subsets with low pairwise intersection:
whenever $F_1,F_2\in\F$, it should be the case that $|F_1\cap F_2|\leq d$.
The value of the parameter~$d$ will be chosen later.

There exist such families of cardinality~$n^{\Omega(d)}$.
We will use a construction by Nisan and Wigderson~\cite{NW},
modified slightly.

Pick an (odd) prime~$p$ in the interval between $\sqrt{n/8-1/2}$ and~$\sqrt{n/2-2}$
and set $m = p-1$.
Also pick a subset~$D$ of size~$m$
in the finite field~$\FF_p$.

Embed $\FF_p\times \FF_p$ into $[\frac{n}{2}]\setminus \{0, \frac{n}{2}-1\}$.
The members of the family~$\F$ are obtained from
the graphs of polynomials of degree (strictly) less than~$d$
restricted to the subset~$D$,
by applying this embedding and adding the numbers~$0$ and~$\frac{n}{2}-1$.
There are $p^{d}$ such polynomials in total;
if $d< m$, then different polynomials give rise to different
subsets in the family~$\F$, because
the graphs of any two polynomials of degree~$<d$
can have an overlap of size~$< d$ only.
This gives us a family~$\F$ of $p^d = \Omega (n^{d/2})$ (sub)sets,
each of size~$\Omega(n^{1/2})$, with pairwise intersection
of at most~$d+1$.

For each subset $F\in \F$,
apply the construction described in Subsection~\ref{app:oca:many},
using the~Dyck word  $\sigma=\sigma_n$. For brevity, denote $w_n =
\Wd(\sigma_n)$. 
This gives, for each $F\in \F$, a path~$\pi_F$ in the graph~$\B$.
By Theorem~\ref{lwr-sqrtlog} and Claim~\ref{width-IS},
this path satisfies the inequality
\begin{equation*}
\max_i\big|B(v_i^{(F)})\big|
\geq w_n.
\end{equation*}
Therefore, there is a vertex~$v_i^{(F)}$ for which
$|B(v_i^{(F)})|$
is at least $w_n$.
By Claim~\ref{prop:oca:vertices-and-paths},
this set is
included into $C(\pi_F) = C_F$ (see Claim~\ref{C(F)=C_F}).
Pick
\begin{equation*}
d =  w_n-1  ,
\end{equation*}
then the vertex~$v_i^{(F)}$ cannot be visited by any other path~$\pi_{F'}\ne\pi_ F$
due to the upper bound on the pairwise intersection of the subsets from~$\F$.

We have thus identified for each path~$\pi_F$, $F\in \F$,
a unique vertex of the graph~$\B$
that is visited by no other path from this family.
So the number of vertices in the graph~$\B$,
i.e., the number of strongly connected components in
the transition graph of the NFA~$\A$,
cannot be less than the number of members of the family~$\F$,
that is,
\begin{equation*}
p^d = 
n^{\Omega (w_n)} = n^{\Omega(\Wd(\sigma_n))}.
\end{equation*}
The number of states of~$\A$ cannot be less than that either,
and this is exactly the assertion of Theorem~\ref{OCA2NFA}.








\end{document}